\documentclass[english, letterpaper]{article}

\usepackage[totalwidth=480pt, totalheight=680pt]{geometry} 
\makeatletter
\renewcommand\subparagraph{\@startsection{subparagraph}{5}{\parindent}{0.25ex \@plus1ex \@minus .2ex}{-1em}{\normalfont\normalsize\bfseries}}
\makeatother

\title{Competitive Kill-and-Restart and Preemptive Strategies\\
for Non-Clairvoyant Scheduling}

\author{%
 Sven~Jäger\footnote{RPTU Kaiserslautern-Landau, Paul-Ehrlich-Straße~14, 67663 Kaiserslautern, Germany. \newline \texttt{ sven.jaeger@rptu.de}}
 \and Guillaume~Sagnol\footnote{Institute for Mathematics, Technische Universität Berlin, Straße des 17. Juni 136, 10623 Berlin, Germany. \newline \texttt{ \{sagnol,dschmidt\}@math.tu-berlin.de}}~\footnotemark[4]
 \and Daniel~Schmidt~genannt~Waldschmidt\footnotemark[2]~\footnotemark[4]
 \and Philipp~Warode\footnote{Humbold-Universität Berlin, Unter den Linden~6, 10099 Berlin, Germany. \newline \texttt{philipp.warode@hu-berlin.de}}~\footnote{supported by the Deutsche Forschungsgemeinschaft (DFG, German Research Foundation) under Germany's Excellence Strategy --- The Berlin Mathematics Research Center MATH+ (EXC-2046/1, project ID: 390685689).}%
}

\newcommand{\colorpalette}{r}

\newif\ifhidecomments
\hidecommentsfalse

\usepackage[utf8]{inputenc} 

\usepackage{amsmath}
\usepackage{amssymb}
\usepackage{amsthm}
\usepackage{mathtools,dsfont,bm,nicefrac,stmaryrd}

\usepackage[dvipsnames]{xcolor}
\usepackage{tikz}
\usetikzlibrary{calc,decorations.pathreplacing,patterns,positioning}

\usepackage{caption}
\usepackage[colorlinks,breaklinks=true,urlcolor=black]{hyperref}
\usepackage[noabbrev,capitalise]{cleveref}
\crefname{equation}{equation}{equations}
\usepackage{etoolbox}
\usepackage[numbers,longnamesfirst]{natbib}
\makeatletter
\pretocmd{\NAT@citexnum}{\@ifnum{\NAT@ctype>\z@}{\let\NAT@hyper@\relax}{}}{}{}
\makeatother

\usepackage{algorithm}
\usepackage{algpseudocode}
\algrenewcommand\algorithmicrequire{\textbf{Input:}}
\algrenewcommand\algorithmicensure{\textbf{Output:}}

\usepackage{thm-restate}

\newtheorem{theorem}{Theorem}[section]
\newtheorem{lemma}[theorem]{Lemma}
\newtheorem{proposition}[theorem]{Proposition}
\newtheorem{claim}{Claim}

\definecolor{colorsven}{RGB}{119,182,186}
\definecolor{colorguillaume}{RGB}{82,17,214}

\newcommand{\E}{\mathbb{E}}
\newcommand{\R}{\mathbb{R}}
\newcommand{\N}{\mathbb{N}}
\newcommand{\Z}{\mathbb{Z}}

\renewcommand{\vec}[1]{\boldsymbol{#1}}

\newcommand{\opt}{\ensuremath{\mathrm{OPT}}}
\makeatletter
\newcommand{\alg}[1][b]{\@ifstar{\ensuremath{\mathfrak{D}}}{\ensuremath{\mathfrak{D}_{#1}}}}
\newcommand{\ralg}{\@ifstar{\ensuremath{\mathfrak{R}}}{\ensuremath{\mathfrak{R}_b}}}
\makeatother

\newcommand{\wsetf}{\ensuremath{\mathrm{WSETF}}}
\newcommand{\pwspt}{\ensuremath{\mathrm{PWSPT}}}
\newcommand{\wspt}{\ensuremath{\mathrm{WSPT}}}
\newcommand{\fgipp}{\ensuremath{\mathrm{F\text{-}GIPP}}}
\newcommand{\rr}{\ensuremath{\mathrm{RR}}} 
\newcommand{\wrr}{\ensuremath{\mathrm{WRR}}} 

\newcommand{\threefield}[3]{%
\def\temp{#2}%
$#1 \, \vert 
\ifx\temp\empty{} \else \, #2 \, \fi
 \vert \, \sum #3$}

\newcommand{\jobrank}{q}

\newcommand{\tALG}{T'}
\newcommand{\tSETF}{T''}

\newcommand{\minrank}{\jobrank_{\min}}
\newcommand{\maxrank}{\jobrank_{\max}}

\newcommand{\infprobingjobs}{\widehat{J}}

\renewcommand{\epsilon}{\varepsilon}

\definecolor{mppetrol}{RGB}{0,114,125}
\definecolor{mpgreen}{RGB}{190,214,0}
\definecolor{mpgray}{RGB}{87,97,87}

\definecolor{tured}{cmyk}{.2,1,1,0}
\definecolor{tuyellow}{cmyk}{0,.5,1,0}
\definecolor{tugreen}{cmyk}{.85,.3,.6,.1}
\definecolor{tublue}{cmyk}{1,.35,.3,0}
\definecolor{tugray}{cmyk}{0,0,0,.7}

\if\colorpalette m
\colorlet{color1}{mppetrol}

\colorlet{color2}{color1!85}
\colorlet{color3}{color1!70}
\colorlet{color4}{color1!55}
\colorlet{color5}{color1!40}
\colorlet{color6}{color1!25}
\colorlet{contrast}{mpgreen}

\colorlet{mygray}{mpgray!20}
\colorlet{outlinecolor}{tured}

\else\if\colorpalette t
\colorlet{color1}{tublue}

\colorlet{color2}{color1!85}
\colorlet{color3}{color1!70}
\colorlet{color4}{color1!55}
\colorlet{color5}{color1!40}
\colorlet{color6}{color1!25}
\colorlet{contrast}{tured}

\colorlet{mygray}{tugray!20}
\colorlet{outlinecolor}{tuyellow}

\else\if\colorpalette r
\definecolor{strong1}{HTML}{DB001F}

\colorlet{color1}{strong1!50}
\definecolor{strong2}{HTML}{F68C21}
\colorlet{color2}{strong2!50}
\definecolor{strong3}{HTML}{FFED00}
\colorlet{color3}{strong3!50}
\definecolor{strong4}{HTML}{238224}
\colorlet{color4}{strong4!50}
\definecolor{strong5}{HTML}{3840FE}
\colorlet{color5}{strong5!50}
\definecolor{strong6}{HTML}{730087}
\colorlet{color6}{strong6!50}
\definecolor{contrastcolor}{HTML}{95540e}
\colorlet{contrast}{contrastcolor!50}

\colorlet{mygray}{gray}
\colorlet{outlinecolor}{magenta}

\fi\fi\fi

\makeatletter
\def\NAT@spacechar{~}
\makeatother

\graphicspath{{figures/}}

\begin{document}

\maketitle
\begin{abstract}
We study kill-and-restart and preemptive strategies for the fundamental scheduling problem of minimizing the sum of weighted completion times on a single machine in the non-clairvoyant setting.
First, we show a lower bound of~$3$ for any deterministic non-clairvoyant kill-and-restart strategy. Then, we give for any $b > 1$ a tight analysis for the natural $b$-scaling kill-and-restart strategy as well as for a randomized variant of it. In particular, we show a competitive ratio of $(1+3\sqrt{3})\approx 6.197$ for the deterministic and of $\approx 3.032$ for the randomized strategy, by making use of the largest eigenvalue of a Toeplitz matrix. In addition, we show that the preemptive Weighted Shortest Elapsed Time First (WSETF) rule is $2$-competitive when jobs are released online, matching the lower bound for the unit weight case with trivial release dates for any non-clairvoyant algorithm.
Using this result as well as the competitiveness of round-robin for multiple machines, we prove performance guarantees smaller than $10$ for adaptions of the $b$-scaling strategy to online release dates and unweighted jobs on identical parallel machines.

\end{abstract}


\section{Introduction}\label{sec:intro}

Minimizing the total weighted completion time on a single processor is one of the most fundamental problems in the field of machine scheduling. The input consists of $n$ jobs with processing times $p_1,\ldots,p_n$ and weights~$w_1,\dotsc,w_n$, and the task is to sequence them in such a way that the sum of weighted completion times~$\sum_{j=1}^n w_j C_j$ is minimized. We denote this problem as \threefield{1}{}{w_jC_j}. \Citet{Smi56} showed in the 50's that the optimal schedule is obtained by the Weighted Shortest Processing Time first (\wspt{}) rule, i.e., jobs are sequenced in non-decreasing order of the ratio of their processing time and their weight.

Reality does not always provide all information beforehand. Around 30 years ago, the non-clairvoyant model, in which the processing time of any job becomes known only upon its completion, was introduced for several scheduling problems~\cite{FST94,SWW95,MPT94}.
It is easy to see that no non-preemptive non-clairvoyant algorithm can be constant-competitive for the unweighted variant~\threefield{1}{}{C_j}. In their seminal work, \citet*{MPT94} proved for this problem that allowing preemption breaks the non-constant barrier. Specifically, they showed that the natural round-robin algorithm is $2$-competitive, matching a lower bound for all non-clairvoyant algorithms. This opened up a new research direction, leading to constant-competitive preemptive non-clairvoyant algorithms in much more general settings, like weighted jobs~\cite{KC03}, multiple machines~\cite{BBEM12,IKM17,IKMP14}, precedence constraints~\cite{GGKS19}, and non-trivial release dates. When jobs are released over time, they are assumed to be unknown before their arrivals (online scheduling). No lower bound better than $2$ is known for this case, whereas the best known upper bound before this work was $3$, see e.g.~\cite{LM22}.

But there is a downside of the preemptive paradigm as it uses an unlimited number of interruptions at no cost and has a huge memory requirement to maintain the ability to resume all interrupted jobs.
Therefore, we continue by studying the natural class of \emph{kill-and-restart strategies} that---inspired by computer processes---can abort the execution of a job (kill), but when processed again later, the job has to be re-executed from the beginning (restart). It can be considered as an intermediate category of algorithms between preemptive and non-preemptive ones, as on one hand jobs may be interrupted, and on the other hand when jobs are completed, they have been processed as a whole. Hence, by removing all aborted executions one obtains a non-preemptive schedule. Although this class of algorithms has already been investigated since the 90's~\cite{SWW95}, to the best of our knowledge, 
the competitive ratio of non-clairvoyant kill-and-restart strategies for the total completion time objective has never been studied.

\paragraph{Our Contribution.}

We start by strengthening the preemptive lower bound of $2$ for the kill-and-restart model.
\begin{restatable*}{theorem}{DetLB}\label{thm:detLB}
	For \threefield{1}{}{C_j}, no deterministic non-clairvoyant kill-and-restart strategy can achieve a competitive ratio smaller than $3-\frac{2}{n+1}$ on instances with $n\geq 3$ jobs, even if every job~$j$ has processing time $p_j \ge 1$.
\end{restatable*}

The main part of this work is devoted to the $b$-scaling strategy~$\alg$ that repeatedly probes each unfinished job for the time of an integer power of $b>1$ multiplied by its weight. For \threefield{1}{}{w_jC_j} it is easy to see that $\alg[2]$ is $8$-competitive by comparing its schedule to the weighted round-robin schedule for a modified instance and using the $2$-competitiveness due to \citet{KC03}. Using a novel and involved analysis we determine the exact competitive ratio of $\alg$.

\begin{restatable*}{theorem}{DetALGUB}\label{thm:upperbound}
	For $b>1$, $\alg$ is $\bigl( 1+\frac{2b^{\nicefrac{3}{2}}}{b-1} \bigr)$-competitive for \threefield{1}{}{w_j C_j}. This ratio is minimized for $b=3$, yielding a performance guarantee of $1+3\sqrt{3}\approx 6.196$.
\end{restatable*} 

\begin{restatable*}{theorem}{DetALGLB}\label{thm:ALG_LB}
	For every $b > 1$, there exists a sequence of instances $(\vec{p}_L)_{L \in \N}$ for \threefield{1}{}{C_j} such that 
	\[
	\lim_{L \to \infty} \frac{\alg (\vec{p}_L)}{\opt(\vec{p}_L)} = 1+\frac{2b^{\frac{3}{2}}}{b-1}
	.
	\]
\end{restatable*}

Our main technique is to reduce the problem of finding the competitive ratio of $\alg$ to the computation of the largest eigenvalue of a tridiagonal Toeplitz matrix. Subsequently, we obtain a significantly better exact competitive ratio for a randomized version of the $b$-scaling strategy, denoted by $\ralg$, that permutes the jobs uniformly at random and chooses a random offset drawn from a log-uniform distribution.

\begin{restatable*}{theorem}{RandALGUB}\label{thm:Rand_ALG_UB}
	For every $b>1$, $\ralg$ is $\frac{2b+\sqrt{b}-1}{\sqrt{b} \ln b}$-competitive for \threefield{1}{}{w_j C_j}. This ratio is minimized for $b\approx 8.16$, yielding a performance guarantee smaller than $3.032$.
\end{restatable*}

\begin{restatable*}{theorem}{TheoLBRand}
For all $b>1$, there exists a sequence of instances $(\vec{p}_L)_{L\in\mathbb{N}}$ for \threefield{1}{}{C_j} such that \[\lim_{L \to \infty} \frac{\ralg(\vec{p}_L)}{\opt(\vec{p}_L)} = \frac{\sqrt{b}+2b-1}{\sqrt{b}\ln b}.\]
\end{restatable*}

The analysis basically mimics that of the deterministic strategy, but it is necessary to group the jobs whose Smith ratio falls in the $i$th interval of the form $(b^{i/K},b^{(i+1)/K}]$, where $K$ is a large natural number. This approach leads to the computation of the largest eigenvalue of a banded symmetric Toeplitz matrix of bandwidth~$2K-1$, and the result is obtained by letting $K\to\infty$.

We then study more general scheduling environments. For the online problem, in which jobs are released over time, denoted by \threefield{1}{r_j,\,\mathrm{pmtn}}{w_j C_j}, we close the gap for the best competitive ratio of preemptive algorithms by analyzing the Weighted Shortest Elapsed Time First rule, short \wsetf. This policy runs at every point in time the job(s) with minimum ratio of the processing time experienced so far (elapsed time) over the weight.

\begin{restatable*}{theorem}{WSETF}\label{thm:WSETF}
 \wsetf{} is $2$-competitive for \threefield{1}{r_j,\,\mathrm{pmtn}}{w_j C_j}.
\end{restatable*} 

\Cref{thm:WSETF} generalizes the known $2$-competitiveness for trivial release dates shown by \citet{KC03}. It also matches the performance guarantee of the best known stochastic online scheduling policy \fgipp~\cite{MV14}, a generalization of the Gittins index priority policy~\cite{Sev74,Wei95}, for the stochastic variant of our problem where the probability distributions of the processing times are given at the release dates and the expected objective value is to be minimized. Our improvement upon the analysis of this policy, applied to a single machine, is threefold: First, our strategy does not require any information about the distributions of the processing times, second, we compare to the clairvoyant optimum, while \fgipp{} is compared to the optimal non-anticipatory policy, and third, \wsetf{} is more intuitive and easier to implement in applications than the \fgipp{} policy.

Using \cref{thm:WSETF}, we then give an upper bound on the competitive ratio of a generalized version of $\alg$ that also handles jobs arriving over time by never interrupting a probing.

\begin{restatable*}{theorem}{TrivialBoundReleaseDates}\label{thm:trivial_bound_release_dates}
	$\alg$ is $\frac{2b^4}{2b^2-3b+1}$-competitive for \threefield{1}{r_j}{w_jC_j}. This ratio is minimized for $b=\frac{9+\sqrt{17}}{8}$, yielding a performance guarantee of $\frac{107+51\sqrt{17}}{32} \approx 9.915$.
\end{restatable*} 

Finally, we also analyze the unweighted problem~\threefield{\mathrm P}{}{C_j} on multiple identical parallel machines.

\begin{restatable*}{theorem}{TrivialBoundParallel}\label{thm:trivial_bound_parallel}
	$\alg$ is $\frac{3b^2-b}{b-1}$-competitive for \threefield{\mathrm P}{}{C_j}. This ratio is minimized for $b=\frac{3+\sqrt{6}}{3}$, yielding a performance guarantee of $5+2\sqrt{6} \approx 9.899$.
\end{restatable*}

\paragraph{Related Work.}
\subparagraph{Non-preemptive scheduling.}
The beginnings of the field of machine scheduling date back to the work of \citet{Smi56}, who investigated the problem of non-preemptively minimizing the sum of weighted completion times on a single machine. Its optimal schedule is obtained by sequencing the jobs in non-decreasing order of their processing time to weight ratio $\nicefrac{p_j}{w_j}$ (Smith's rule). When all jobs have unit weights, one obtains the Shortest Processing Time first (SPT) rule. This can be generalized to the identical parallel machine setting, where list scheduling~\cite{Gra66} according to SPT is optimal~\cite{CMM67} for unit-weight jobs.
However, the problem of scheduling jobs released over time on a single machine is strongly NP-hard~\cite{Rin76} (even for unit weights), and Chekuri and Khanna developed a polynomial-time approximation scheme (PTAS) for it~\cite{ABC+99}. When jobs arrive online, no deterministic algorithm can be better than $2$-competitive~\cite{HV96}, and this ratio is achieved by a delayed variant of Smith's rule~\cite{AP04}.

In the non-clairvoyant setting it is well known that no (randomized) non-preemptive algorithm is constant-competitive (see \cref{prop:non-clairvoyant-non-constant}). A less pessimistic model is the stochastic model, where the distributions of the random processing times~$P_j$ are given and one is interested in non-anticipatory (measurable) policies~\cite{MRW84}. For some classes of distributions, this information allows obtaining constant expected competitive ratios~\cite{SSS06} for parallel identical machines. However, policies minimizing the expected competitive ratio do not need to minimize the expected objective value---the classic measure in stochastic optimization. For this criterion, \citet{Rot66} showed that for the single machine case the optimality of Smith's rule can be transferred to the Weighted Shortest Expected Processing Time rule, in which jobs are sorted in non-decreasing order of $\nicefrac{\E[P_j]}{w_j}$. In order to deal with the stochastic counterparts of the NP-hard problems mentioned above, \citet{MSU99} introduced approximative scheduling policies, whose expected objective value is compared to the expected objective value of an optimal non-anticipatory policy. While there are constant-competitive policies for stochastic online scheduling on a single machine~\cite{Jaeg21}, the performance guarantees of all known approximative policies for multiple machines depend on either the maximum coefficient of variation~\cite{JS18} or the number of jobs and machines~\cite{IMP15}, even for unit-weight jobs released at time~$0$.

\subparagraph{Preemptive scheduling.}
For the clairvoyant offline model, allowing preemption only helps in the presence of non-trivial release dates~\cite{McN59}. In this case, the optimal preemptive schedule may be a factor of $\mathrm e/(\mathrm e - 1)$ better than the best non-preemptive one~\cite{EL16}. Finding an optimal preemptive schedule is still strongly NP-hard~\cite{LLLRK84}, and there is a PTAS adapted to this problem~\cite{ABC+99}. For jobs arriving online \citet{Sit10} developed a $1.566$-competitive deterministic algorithm, and \citet{EvS03} proved a lower bound of $1.073$.

When the job lengths are uncertain, allowing preemption becomes much more crucial. \Citet{MPT94} showed that the simple (non-clairvoyant) round-robin procedure has a competitive ratio of $2$ for minimizing the total completion time on identical machines. This gives the same share of machine time to each job in rounds of infinitesimally small time slices. For weighted jobs, the Weighted Round-Robin (\wrr) rule (also known as generalized processor sharing (GPS)~\cite{parekh1992,chandra2000}), which always distributes the available machine capacity to the jobs proportionally to their weights, was shown to be $2$-competitive on a single machine by \citet{KC03}, and the same competitive ratio is achieved by a generalization for multiple identical machines~\cite{BBEM12}. Similar time sharing algorithms were also developed in the context of non-clairvoyant online scheduling, where jobs arrive over time and are not known before their release dates. Here one can distinguish between minimizing the total (weighted) completion time and the total (weighted) flow time.
The \wrr{} rule can be generalized in two natural way in this setting: Either the machine capacity is still allocated based only on the weights or based on the weighted elapsed times, resulting in the \wsetf{} rule, mentioned above. It is easy to see that both are $3$-competitive, see e.g.~\cite{LM22}. On the other hand, there exist examples showing that the first option is not $2$-competitive for total weighted completion time. For the total weighted flow time objective constant competitiveness is unattainable~\cite{MPT94}. Apart from work on non-constant competitive ratios~\cite{BL04}, the problem has been primarily studied in the resource augmentation model~\cite{KP00}, where the machine used by the algorithm runs $1+\varepsilon$ times faster.
\Citeauthor{KC03} and \citet{BD07} independently proved that \wsetf{} is $(1+\varepsilon)$-speed $(1+\nicefrac 1 \varepsilon)$-competitive for weighted flow time on a single machine. By running this algorithm on the original-speed machine, the completion times increase by a factor of $1+\nicefrac 1 \varepsilon$, so that one obtains a $(1+\varepsilon)(1+\nicefrac 1 \varepsilon)$-competitive algorithm for the total weighted completion time~\cite{BP04}, which yields a ratio of $4$ for $\varepsilon = 1$. The proofs of \citeauthor{KC03} and \citeauthor{BD07} both proceed by showing that at any time~$t \ge 0$ the total weight of unfinished jobs in the \wsetf{} schedule is at most a factor of $(1+\nicefrac 1 \varepsilon)$ larger than the unfinished weight in the optimal schedule. The lower-bound example of \citeauthor{MPT94} (many equal small jobs released at time~$0$) demonstrates that with such an approach no better bound than $4$ is achievable.
Consequently, a completely different technique is needed to prove \cref{thm:WSETF}. For the much more general setting of unrelated machines \citet*{IKMP14} established a $(1+\varepsilon)$-speed $\mathcal{O}(\nicefrac{1}{\varepsilon^2})$-competitive algorithm.
\Citeauthor{MPT94} also considered the model in which the number of allowed preemptions is limited, for which they devised algorithms that resemble the kill-and-restart algorithms presented in this paper. As mentioned above, for the stochastic model for minimizing the expected total weighted completion time, the Gittins index policy is optimal for single-machine with trivial release dates~\cite{Sev74,Wei95}, and \citet{MV14} established a $2$-competitive online policy for multiple machines and arbitrary release dates.

\subparagraph{Kill-and-restart scheduling.}
The kill-and-restart model was introduced by \citet{SWW95} in the context of makespan minimization. For the total completion time objective we are not aware of any work on kill-and-restart strategies in the non-clairvoyant model. However, in the clairvoyant online model, kill-and-restart algorithms have been considered by \citet{Ves97} and \citet{EvS03}, who gave lower bounds that are larger than the lower bounds for preemptive algorithms but much smaller than the known lower bounds for non-preemptive online algorithms, suggesting that allowing restarts may help in the online model. The proof of this fact was given several years later by \citet{vSLP05}, who achieved a deterministic competitive ratio of $3/2$ for minimizing the total completion time on a single machine, beating even the lower bound of $\mathrm e/(\mathrm e - 1) \approx 1.582$ for any randomized non-preemptive online algorithm~\cite{CMNS01}. In the non-clairvoyant setting, considered in this work, we observe a much larger benefit from allowing restarts, reducing the competitive ratio from $\Omega(n)$ to a constant.

\subparagraph{Further related work.}
In the end, all aborted probings served only the purpose of collecting information about the unknown processing times of the jobs. Kill-and-restart strategies can thus be regarded as online algorithms for non-preemptive scheduling with the possibility to invest time in order to obtain some information. In that sense, the considered model resembles that of explorable uncertainty~\cite{DEMM20,AE20,GL21}. In order to allow for any reasonable competitiveness results, it must be ensured in both models that testing/probing provides some benefit to the algorithm other than information gain. In the explorable uncertainty model, this is achieved by the assumption that testing can shorten the actual processing times, while in our model the probing time replaces the processing time if the probing was long enough.

Scheduling on a single machine under the kill-and-restart model shares many similarities with \emph{optimal search problems}, in which a number of agents are placed in some environment and must either find some target or meet each other as quickly as possible. A problem that received a lot of attention is the so-called $w$-lanes \emph{cow-path} problem, in which
an agent (the cow) is initially placed at the crossing of $w$ roads, and must find
a goal (a grazing field) located at some unknown distance on one the $w$ roads.
For the case $w=2$, deterministic and randomized search strategies 
were given that
achieve the optimal competitive
ratio of $9$~\cite{Baeza-YatesCR93} and approximately $4.5911$~\cite{KRT96}, respectively. This work has been extended by
\citet{KMSY98}, who give optimal deterministic and randomized algorithms for all $w\in\mathbb{N}$. The single-machine scheduling problem with kill-and-restart strategies can in fact be viewed in this framework: There are now $n=w$ goals, and the $j$th goal is located at some unknown distance $p_j$ on the $j$th road. The agent can move at unit speed on any of the roads, and has the ability to \emph{teleport back to the origin} at any point in time, which represents the action of aborting a job. The objective is to minimize the sum of times at which each goal is found.

\section{Preliminaries}\label{sec:preliminaries}

We consider the machine scheduling problem of minimizing the weighted sum of completion times on a single machine (\threefield{1}{}{w_j C_j}).
Formally, we consider instances $I = (\vec{p},\vec{w})$ consisting of a vector of processing times $\vec{p} = (p_j)_{j=1}^n$ and
a vector of weights $\vec{w} = (w_j)_{j=1}^n$.

If the jobs are in WSPT order, i.e., jobs are ordered increasingly by their Smith ratios $p_j/w_j$, then it is easy to see that sequencing the jobs in this ordering yields an optimal schedule.
We denote this (clairvoyant) schedule by $\opt(I)$. By slight abuse of notation, we also denote the objective value of an optimal schedule by $\opt(I)$. In particular, its cost is
$
	\opt(I) = \sum_{j=1}^n w_j\sum_{k=1}^j p_{k} = \sum_{j=1}^n p_{j}\sum_{k=j}^n w_j
$. 

The focus of our work lies in the analysis of non-clairvoyant strategies.
We call a strategy \emph{non-clairvoyant} if it does not use information on the processing time $p_j$ of a job~$j$ before $j$ has been completed.
A deterministic strategy $\alg*$ is said to be \emph{$c$-competitive} if, for all instances $I = (\vec{p}, \vec{w})$, $\alg* (I) \leq c \cdot \opt(I)$, where $\alg* (I)$ denotes the cost of the strategy for instance $I$. The \emph{competitive ratio} of $\alg*$ is defined as the infimum over all~$c$ such that $\alg*$ is $c$-competitive. 
For a randomized strategy \ralg*, the cost for instance $I$ is a random variable $X_I\colon \Omega\to\R_{\geq 0}$
that associates an outcome $\omega$ of the strategy's sample space to the realized cost, and we denote by $\ralg*(I)\coloneqq \E[X_I]$ the expected cost of the randomized strategy for instance $I$.
We say that $\ralg*$ is
\emph{$c$-competitive} if for all instances $I = (\vec{p}, \vec{w})$, $\ralg*(I)\leq c \cdot \opt(I)$.
It is well known that for our problem no non-preemptive strategy can achieve a constant competitive ratio. 

\begin{proposition} \label{prop:non-clairvoyant-non-constant}
	No randomized non-preemptive non-clairvoyant strategy has a constant competitive ratio for \threefield{1}{}{C_j}.
\end{proposition}
\begin{proof}
	By Yao's principle~\cite{Yao77} it suffices to construct a randomized instance for which any deterministic strategy has non-constant competitive ratio. To this end, we consider the instance with $n$ jobs where $p_1=\cdots=p_{n-1}=1$ and $p_n=n^2$ and randomize uniformly over all permutations of the jobs. Clearly, an optimal clairvoyant strategy sequences the jobs in any realization in SPT order and hence, we have
	$
	\opt=\sum_{j=1}^n (n-j+1)p_j=\frac{1}{2}n(n-1)+n-1+n^2=\mathcal{O}(n^2)
	.
	$

	The schedule of any deterministic strategy can be represented as a permutation of the jobs as idling only increases the objective value. Hence, for any permutation we obtain the expected cost
	\[
	\sum_{\sigma}\frac{1}{n!}\sum_{j=1}^{n}(n-j+1)p_{\sigma(j)}=\sum_{k=1}^n\frac{1}{n}\left(\sum_{j=1}^{k-1}j +(n-k+1)n^2 + \sum_{j=k+1}^{n}j \right)\geq \frac{1}{n}\cdot\frac{(n-1)n^3}{2}=\Omega(n^3),
	\]
	where we used the fact that in a uniformly distributed permutation, the probability that the long job appears in each position $k\in[n]$ is $\frac{1}{n}$.
\end{proof}

\newcommand{\statespace}{\mathcal{S}}
\newcommand{\actionspace}{\mathcal{A}}
\newcommand{\actionspaceS}[1][s]{\mathcal{A}(#1)}

\paragraph{Kill-and-Restart Strategies.} 
Due to this negative result, we study non-clairvoyant \emph{kill-and-restart strategies} for \threefield{1}{}{w_j C_j} that may abort the processing of a job, but when it is processed again later, it has to be executed from the beginning. 
In order to define such strategies, we first introduce a state and action space as well as a transition function modeling the kill-and-restart setting. Then, we can describe kill-and-restart strategies as functions mapping states to actions.

Formally, we consider the \emph{state space} $\statespace \coloneqq \R \times 2^{[n]} \times \R^n$. 
A state~$(\theta, U, \vec \mu) \in \statespace$ consists of the current time $\theta$, the set of unfinished jobs $U$ at~$\theta$, and a vector $\vec\mu$ of lower bounds on the processing times learned from past probings, such that $p_j\geq \mu_j$ for all jobs~$j$.
For every state $s = (\theta, U, \vec \mu) \in \statespace$, there is a set of possible kill-and-restart actions $\actionspaceS$, where an action $a = \bigl((t_i,j_i,\tau_i)\bigr)_{i \in \mathcal{I}} \in \actionspaceS$ is a family of probings $(t_i,j_i,\tau_i)$ such that the intervals $ (t_i, t_i+\tau_i)$, $i \in \mathcal{I}$, are disjoint and contained in $\R_{> \theta}$ and $j_i \in U$ for all $i \in \mathcal I$.
We denote by $\actionspace = \bigcup_{s \in \statespace} \actionspaceS$ the set of all actions in all states. 
Additionally, we define a \emph{transition function} $T_I \colon \statespace \times \actionspace \to  \statespace$ depending on the instance~$I$. This function transforms any state $s = (\theta, U, \vec \mu)$ and action $a = \big((t_i, j_i, \tau_i)\big)_{i \in \mathcal{I}} \in \mathcal A(s)$  into a new state $s' = (\theta', U', \vec \mu')$ as follows.  First, we identify the probing indexed with $i^* \coloneqq \operatorname{argmin} \big\{ t_i + p_{j_i} \mid i \in \mathcal{I} \text{ with } \tau_i \geq p_{j_i} \big\}$, which corresponds to the first probing in $a$ that leads to the completion of some job.
Then, the lower bounds $\vec{\mu}'$ of the new state~$s'$ are defined by $\mu_j' \coloneqq \max \big\{ \mu_j , \max \{ \min \{ \tau_i, p_j \} \mid i \in \mathcal{I} : t_i \leq t_{i^*}, j_i = j \} \big\}$, i.e., the lower bounds are set to the maximum probing time a job received so far or, if a job is completed, to its processing time. Further, the job completing in probing $i^*$ is removed from the set of unfinished jobs by setting $U'\coloneqq U\setminus \{j_{i^*}\}$, and the time is updated to $\theta' \coloneqq t_{i^*} + p_{j_{i^*}}$.
Finally, we define a \emph{kill-and-restart strategy} as a function $\Pi \colon \statespace \to \actionspace$ with $\Pi(s) \in \actionspaceS$ for all $s \in \statespace$. 
Note that a kill-and-restart strategy is non-clairvoyant by definition as it only has access to the lower bounds on the processing times, while the actual processing time is only revealed to the strategy upon completion of a job.

However, observe that such strategies may not be implementable, e.g., on a Turing machine,
as the above definition allows for an infinite number of probings in a bounded time range. 
On the other hand, a deterministic kill-and-restart strategy without infinitesimal probing cannot be constant-competitive.
To see this, consider an arbitrary algorithm $ALG$, and assume without loss of generality that the first job it probes is the first job presented in the input. Denote by $t>0$ the first probing time, and consider the instance $I_\epsilon=\big((t, \epsilon t,\ldots,\epsilon t),(1,\ldots,1)\big)$ with $n$ unit weight jobs. By construction, $ALG$ processes the first job without aborting it, so $ALG\geq nt + n(n-1)/2\cdot  t \epsilon = nt+\mathcal{O}(\epsilon)$. On the other hand $\opt$ schedules the job in SPT order, yielding $\opt=n(n-1)/2\cdot t \epsilon + (n-1)t\epsilon+t=t+\mathcal{O}(\epsilon)$. Hence, $\alg/\opt$ approaches $n$ as $\epsilon\to 0$.
This subtlety is in fact inherent to all scheduling problems with unknown processing times or 
search problems with unknown distances.

We discuss in Section~\ref{sec:upperbound} that no deterministic kill-and-restart strategy can be constant-competitive without infinitesimal probing, as there is no lower bound on the processing times at time $0$.
On the other hand, infinitesimal probing can be avoided 
if we know a lower bound on the $p_j$'s, thus turning the strategies analyzed in this paper into implementable ones.

We denote by $Y_j^{\alg*}(I, t)$ the total time for which the machine has been busy processing job~$j$ until time~$t$ in the schedule constructed by the strategy~$\alg*$ on the instance~$I$.

\section{Lower Bound for Deterministic Strategies}\label{sec:lowerbound}

\DetLB

\begin{proof}
	Let $\epsilon \in \bigl(\frac{2}{n+1},1\bigr]$ and define $T \coloneqq \frac{(2-\epsilon)(n^2+n)}{\epsilon(n+1)-2}$.
	Consider an arbitrary deterministic kill-and-restart strategy~$\alg*$ with the initially chosen family of probings $(t_i,j_i,\tau_i)_{i\in\mathcal{I}}$. Let $Y_j(\theta)\coloneqq \sum_{i\in\mathcal{I}: t_i< \theta, j_{i}=j } \min\{\tau_i,\theta-t_i\}$ be the total probing time assigned by $\alg*$ to job~$j$ up to time $\theta.$ We define an instance $I\coloneqq(\vec{p},\vec{1})$ by distinguishing two cases on the first job~$j_0$ planned to be probed at or after time~$T$. Note that such a job exists, as otherwise $\alg*$ does not complete all jobs if processing times are long enough.
	
	If $j_0$ is probed for a finite amount of time, we denote by $t\geq T$ the end of its probing time. Then, define $p_j:=1+Y_j(t)$ for all $j\in[n]$. Clearly, no job finishes before $t$ when $\alg*$ runs the instance $I$, hence
	$\alg*(I) \geq nt + \opt(I).$ 
	On the other hand, it is well known that $\opt(I) \le \frac{n+1}{2} \cdot \sum_{j=1}^n p_j$, which is the expected objective value when scheduling non-preemptively in a random order, thus,
	$\opt \leq \frac{n+1}{2} (t+n)$.
	Therefore, we have
	$
	\frac{\alg*(I)}{\opt(I)} \geq 1+\frac{2nt}{(t+n)(n+1)}
	\geq 1+\frac{2nT}{(T+n)(n+1)} =3-\epsilon.
	$
	
	If $j_0$ is probed for $\tau=\infty$, i.e., it is processed non-preemptively until its completion, then for each job $j\neq j_0$ we set $p_j \coloneqq 1+Y_j(T)$. Denote by $\opt'$ the optimal $\mathrm{SPT}$ cost for  jobs $[n] \setminus \{ j_{0} \}$, and set $p_{j_0}:= 10\cdot \opt'$. As $j_0$ is the first job to complete in $I$, we clearly have
	$\alg*(I) \geq n\cdot p_{j_0}=  10n \cdot \opt'$. On the other hand, $\opt$ processes $j_0$ last, so $\opt(I)=\opt' + \sum_{j\neq j_0} p_j + p_{j_0} \leq (1+1+10) \cdot \opt'$. This implies $\frac{\alg*(I)}{\opt(I)}\geq \frac{10n}{12}\geq 3-\frac{2}{n+1}$,
	where the last inequality holds for all $n\geq 3$. 
\end{proof}

\section{\texorpdfstring{The $b$-Scaling Strategy}{The b-Scaling Strategy}}\label{sec:upperbound}

Let us now introduce the $b$-scaling algorithm $\alg$, which is the basis for most results in this paper.
The idea of this algorithm is simple and quite natural: 
it proceeds by rounds $\jobrank \in \Z$. In round~$\jobrank$ every non-completed job is probed (once) for $w_j b^\jobrank$ in some prescribed order, where $b>1$ is a constant. To execute $\alg$, we can store for each job its \emph{rank} at time $t$, i.e., the largest $\jobrank$ such that it was probed for $w_jb^{\jobrank -1}$ until $t$. At any end of a probing, $\alg$ schedules the job~$j$ with minimum rank and minimum index for time $w_jb^q$. 

We also introduce a randomized variant of the algorithm. Randomization occurs in two places: First the jobs are reordered according to a random permutation $\Sigma$ at the beginning of the algorithm. Second, we replace the probing time $w_j b^\jobrank$ of the $\jobrank$th round with $w_j b^{\jobrank+\Xi}$ for some random offset $\Xi\in[0,1]$.
\Cref{alg:ALGb_pi_xi} gives the pseudo-code
of this strategy when it starts from round $\jobrank_{0}\in\mathbb{Z}$, in which case it is denoted by $\alg^{\sigma,\xi,\jobrank_{0}}$. The kill-and-restart strategy $\alg^{\sigma,\xi}$ studied in this paper
can actually be seen as the limit of $\alg^{\sigma,\xi,\jobrank_{0}}$ when $\jobrank_{0}\to-\infty$, and is described formally below.
The deterministic $b$-scaling algorithm~$\alg$ is obtained by setting $\sigma=\operatorname{id}$ (the identity permutation) and $\xi=0$, while the randomized variant $\ralg$ is obtained for a permutation $\Sigma$ drawn uniformly at random from $\mathcal{S}_n$ and a random uniform offset $\Xi\sim \mathcal{U}([0,1])$, i.e., 
\[
\alg \coloneqq \alg^{\operatorname{id},0}\qquad \text{ and } \qquad
\ralg \coloneqq \alg^{\Sigma,\Xi}.
\]
As for $\opt$, by slight abuse of notation, we denote by $\alg(I)$ and $\ralg(I)$ the schedule for instance~$I$ computed by $\alg$ and $\ralg$, respectively, as well as its cost. We drop the dependence on~$I$ whenever the instance is clear from context.
\begin{algorithm}[t] 
    \caption{$\alg^{\sigma,\xi,\jobrank_{0}}$: $b$-Scaling algorithm with permutation $\sigma\in\mathcal{S}_n$ and offset $\xi\in[0,1]$, when starting from round $\jobrank_{0}\in\mathbb{Z}$}
	\label{alg:ALGb_pi_xi}
	\begin{algorithmic}[1]
		\Require{$I=(\vec{p},\vec{w})$}
		\Ensure{kill-and-restart schedule}
		\State Set $q\gets \jobrank_{0}$ 
		\State Initialize the list of unfinished jobs permuted according to $\sigma$:
		\mbox{$U\gets \left[\sigma^{-1}(1),\ldots,\sigma^{-1}(n)\right]$}
		\While{$U\neq\emptyset$}
		\For{$j$ in $U$}
		\State probe $j$ for $w_j b^{q+\xi}$
		\If{$j$ is completed} \Comment{{\scriptsize this happens if $p_j \leq w_j b^{q+\xi}$}}
            \State $U\gets U\setminus\{j\}$
		\EndIf \Comment{{\scriptsize Otherwise the probing fails and the job is killed}}
		\EndFor
        \State $q \gets q+1$
		\EndWhile
	\end{algorithmic}
\end{algorithm}
While $\alg^{\text{id},0,\jobrank_{0}}$ can easily be implemented, it is not possible to implement the limit strategy~$\alg$, for example, on a Turing machine, since at a time arbitrarily close to $0$ it has probed each jobs an infinite number of times.

Let us now formally define the strategy $\alg^{\sigma,\xi}$, by describing the action
$a(s)=\big((t_i,j_i,\tau_i)\big)_{i\in\mathcal{I}}\in\mathcal{A}(s)$ it takes in any state $s=(\theta, U, \vec \mu)$, in accordance with the kill-and-restart framework described in Section~\ref{sec:preliminaries}.
Recall that an action is a family of planned probings that the strategy is committed to execute until a job completes and a new action is determined. Moreover, the state~$s$ specifies lower bounds $\mu_j\leq p_j$ for every job~$j$, the set $U$ of unfinished jobs, and the current time $\theta$.

In the initial state $s_0=(0, [n], \vec 0)$, we plan to probe all jobs~$j\in[n]$ in rounds, where in each round the jobs are probed in the order given by $\sigma$ for $w_jb^{q+\xi}$ for some $q\in\Z$ and then $q$ incremented by $1$. Hence, $\sum_{\widehat q=-\infty}^{q-1}\sum_{k\in[n]}w_kb^{\widehat q+\xi}=\frac{b^{q+ \xi}}{b-1} \sum_{k\in[n]} w_k$ is the point in time at which the first job~$j =\sigma^{-1}(1)$ is probed for $w_jb^{q+\xi}$.
We define the action of $\alg^{\sigma,\xi}$ for state $s_0$ by
\[
	a(s_0)\coloneqq \Bigg(\bigg(
	\frac{b^{q+ \xi}}{b-1} \sum_{k\in[n]} w_k +   \sum_{\mathclap{\substack{k \in [n] \\ \sigma(k) < \sigma(j)}}} w_kb^{q+ \xi} ,
	\;
	j,
	\;
	w_j b^{q+ \xi}
	\bigg)\Bigg)_{(q, j) \in  \Z \times [n]}.
\]

In a state $s=(\theta, U, \vec \mu)$ with $\theta > 0$
occurring at the completion of a job, there exists $q^*\in\mathbb{Z}$ by construction such that $\mu_j\in\{w_j b^{q^*+\xi-1}, w_j b^{q^*+\xi}\}$, for all $j\in U$. The set of jobs $J^*\coloneqq \{j\in U \mid \mu_j=w_jb^{q^*+\xi-1}\}$ are those jobs that have not been probed yet for $w_jb^{q^*+\xi}$. Hence, these jobs must be probed first before the new round $q^*+1$ can start. For $q>q^*$ we define
\begin{align*}
\theta_q \coloneqq \theta +  \sum_{k \in J^*} w_kb^{q^* + \xi} +  \sum_{\widehat q = q^{*}+1}^{q-1}\sum_{k\in U} w_k b^{\widehat q}
=\theta + b^{q^* + \xi} \sum_{k \in J^*} w_k + \frac{b^{q+\xi}-b^{q^*+1+\xi}}{b-1} \sum_{k\in U} w_k
\end{align*}
as the point in time when round $q$ starts.
We define the actions of $\alg^{\sigma,\xi}$ for state~$s$ by
\begin{align*}
	a(s) \coloneqq \Bigg(\bigg(
	\theta +  \sum_{\substack{k \in J^*: \\ \sigma(k) < \sigma(j)}} w_k b^{q^* + \xi},
	\;
	j,
	\;
	w_j b^{q^* + \xi}
	\bigg)\Bigg)_{j \in J^*}
	\cup\quad
	\Bigg(
	\bigg(
	\theta_q
	+ \sum_{\substack{k \in U: \\ \sigma(k) < \sigma(j)}} w_k b^{q + \xi} 
	,
	\;
	j,
	\;
	w_j b^{q + \xi}
	\bigg)\Bigg)_{(q, j) \in \Z_{> q^*} \times U}.
\end{align*}

\subsection{\texorpdfstring{Tight Analysis of the Deterministic $b$-Scaling Strategy}{Tight Analysis of the Deterministic b-Scaling Strategy}} \label{subsec:Analysis_Alg_det}
In this \lcnamecref{subsec:Analysis_Alg_det}, we compute tight bounds for the competitive ratio of $\alg$ for \threefield{1}{}{w_jC_j}. 
For the analysis, we need some additional definitions.
We denote by $s_j\coloneqq \frac{p_j}{w_j}$ the \emph{Smith ratio} of job~$j\in[n]$.
Further, we define $D_{jk} \coloneqq Y_j^{\alg}(C_k^{\alg})$ as the amount of time spent probing job $j$ before the completion of job $k$. For all $j,k\in[n]$ we
define the weighted mutual delay $\Delta_{jk}$
by
$\Delta_{jk}\coloneqq w_j D_{jj}$ if $j=k$ and 
$\Delta_{jk}\coloneqq w_k D_{jk} + w_j D_{kj}$ if $j\neq k$. 
Thus, it holds
\[
 \alg(\vec p, \vec w) = 
 \sum_{j=1}^n w_j C_j = \sum_{j=1}^n w_j \sum_{k=1}^n D_{kj}=
 \sum_{1\leq j\leq k\leq n} \Delta_{jk}.
\]
\Cref{lem:DeltaF,lem:Smith_intger_power_b,lem:ratio_of_quads} constitute preparations for \cref{thm:upperbound}, establishing the upper bound on the competitive ratio. Afterwards, the tightness is proven in \cref{thm:ALG_LB}. The first step towards the upper bound is to provide an overestimator of $\Delta_{jk}$ that is piecewise linear in $(s_j, s_k)$.

\begin{lemma}\label{lem:DeltaF}
 Define the function $F\colon \{ (s,s')\in\R_{>0}^2 \mid s\leq s' \}\to \R$ by
 \[
  F(s,s') \coloneqq
 \begin{cases}
  \frac{2}{b-1}\cdot b^{\lfloor \log_b s\rfloor+1} + s' & \text{ if }  \lfloor \log_b(s)\rfloor = \lfloor \log_b(s')\rfloor, \\
  \bigl(\frac{2}{b-1}+1\bigr)\cdot b^{\lfloor \log_b s\rfloor+1} + s & \text{ if } \lfloor \log_b(s)\rfloor < \lfloor \log_b(s')\rfloor.
 \end{cases}
 \]
 Then $F$ is non-decreasing in both arguments. Moreover, for all $j,k\in[n]$ such that $s_j\leq s_k$, it holds that
 $\Delta_{jk}\leq w_j w_k\ F(s_j, s_k)$.
\end{lemma}

\begin{proof}
 Let first $s' > 0$ be fixed, and let $r' \coloneqq \lfloor \log_b(s') \rfloor$.
 Then the function $F(\cdot, s') \colon (0,s'] \to \R$ is obviously non-decreasing on $(0,b^{r'})$ and on $[b^{r'}, s']$.
 To see that it is also non-decreasing around the breakpoint~$b^{r'}$, we take the limit
 \begin{multline*}
  \lim_{t \nnearrow b^{r'}} F(t, s') 
  = \lim_{t \nnearrow b^{r'}} \Bigl(\frac{2}{b-1} + 1 \Bigr) \cdot b^{\lfloor \log_b t \rfloor + 1} + t 
  = \Bigl(\frac{2}{b-1} + 1 \Bigr) \cdot b^{r'} + b^{r'} = \frac{2b}{b-1} \cdot b^{r'} \\
  = \frac{2}{b-1} \cdot b^{\lfloor \log_b b^{r'} \rfloor + 1} < F(b^{r'}, s').
 \end{multline*}

 Now let $s > 0$ be fixed, and let $r \coloneqq \lfloor \log_b(s) \rfloor$. Then the function $F(s, \cdot) \colon [s, \infty) \to \R$ is clearly non-decreasing on $[s, b^{r+1})$ and on $[b^{r+1}, \infty)$. At the breakpoint we have
	\[
  \lim_{t' \nnearrow b^{r+1}} F(s, t') = \lim_{t' \nnearrow b^{r+1}} \frac{2}{b-1} \cdot b^{r+1} + t' = \frac{2}{b-1} \cdot b^{r+1} + b^{r+1} = \Bigl(\frac{2}{b-1} + 1\Bigr) \cdot b^{r+1} < F(s, b^{r+1}),
	\]
 so that it is globally non-decreasing.

For all $j\in[n]$, 
let $q_j \coloneqq  \lceil \log_b(s_j) \rceil$, so that $b^{q_j-1} < s_j \leq b^{q_j}$. 
We have $D_{jj}=\sum_{i=-\infty}^{q_j-1} w_j b^i + p_j 
= w_j(\frac{b^{q_j}}{b-1}+s_j)$,
so it holds $\Delta_{jj} = w_j D_{jj} = w_j^2(\frac{b^{q_j}}{b-1}+s_j) \leq w_j^2 F(s_j,s_j)$,
where we have used the fact
that $q_j=\lceil \log_b(s_j) \rceil\leq \lfloor \log_b(s_j) \rfloor+1$.

Now, let $j\neq k$ such that $s_j \leq s_k$. We first assume that jobs~$j$ and~$k$ complete in the same round, i.e., $\lceil\log_b(s_j)\rceil=\lceil\log_b(s_k)\rceil$. If job $k$ is executed first in this round, then we have
$D_{jk}=\sum_{i=-\infty}^{q_j-1} w_j b^i = w_j\frac{b^{q_j}}{b-1}$ and 
$D_{kj}=\sum_{i=-\infty}^{q_j-1} w_k b^i+p_k = w_k \frac{b^{q_j}}{b-1}+p_k$, which gives
\begin{equation}
 \Delta_{jk}=2w_j w_k \frac{b^{q_j}}{b-1} + w_j p_k
 = w_j w_k \Bigl(\frac{2 b^{q_j}}{b-1} + s_k\Bigr)
 \le w_j w_k  \Bigl(\frac{2 b^{\lfloor \log_b(s_j) \rfloor + 1}}{b-1} + s_k \Bigr)
.\label{bound_Delta_1}
\end{equation}
Similarly, if job $j$ is completed first, we have $D_{jk}=w_j \frac{b^{q_j}}{b-1}+p_j$ and $D_{kj}=w_k\frac{b^{q_j}}{b-1}$, so we obtain $\Delta_{jk} = w_j w_k(\frac{2b^{q_j}}{b-1} + s_j)
\leq w_j w_k(\frac{2b^{\lfloor \log_b(s_j)\rfloor + 1}}{b-1} + s_k)$, i.e., the bound~\eqref{bound_Delta_1} is still valid.
If $\lfloor \log_b(s_j) \rfloor = \lfloor \log_b(s_k) \rfloor$, then the right-hand side equals $w_j w_k F(s_j, s_k)$.
Otherwise, $s_k = b^{\lfloor \log_b(s_j) \rfloor + 1}$, so that
\[w_j w_k \Bigl(\frac{2 b^{\lfloor \log_b(s_j) \rfloor + 1}}{b-1} + s_k\Bigr)
= w_j w_k \cdot \lim_{t' \nnearrow s_k} \Bigl(\frac{2 b^{\lfloor \log_b(s_j) \rfloor + 1}}{b-1} + t'\Bigr)
= w_j w_k \cdot \lim_{t' \nnearrow s_k} F(s_j, t') \le F(s_j, s_k).\]

Now, assume that job~$k$ is completed in a later round than job~$j$, i.e.,
$\lceil\log_b(s_j)\rceil<\lceil\log_b(s_k)\rceil$. Then,
$D_{jk}=w_j \frac{b^{q_j}}{b-1}+p_j$ and 
$D_{kj}\leq w_k \frac{b^{q_j}}{b-1}+w_k b^{q_j}$, where the inequality is tight whenever job~$k$ is probed before job~$j$ in the round where~$j$ is completed. Thus,
\[
\Delta_{jk} \leq 2w_jw_k\frac{b^{q_j}}{b-1}+w_k p_j + w_j w_k b^{q_j} = w_j w_k \Bigl(\frac{2b^{q_j}}{b-1}+b^{q_j}+s_j \Bigr) = w_j w_k \left(\Bigl(\frac{2}{b-1} + 1\Bigr) b^{q_j} + s_j\right).
\]
If $\lfloor \log_b(s_j) \rfloor < \lfloor \log_b(s_k) \rfloor$, then the right-hand side be bounded by $w_j w_k F(s_j, s_k)$, using that $q_j \le \lfloor \log_b(s_j) \rfloor + 1$. Otherwise, $s_j = b^{q_j} < s_k < b^{q_j+1}$, so that
\[w_j w_k \left(\Bigl(\frac{2}{b-1}+1\Bigr) b^{q_j} + s_j \right)
= w_j w_k \cdot \lim_{t \nnearrow s_j} \left(\Bigl(\frac{2}{b-1}+1\Bigr) b^{\lfloor \log_b(t) \rfloor + 1} + t \right)
= w_j w_k \cdot \lim_{t \nnearrow s_j} F(t, s_k) \le w_j w_j F(s_j, s_k).\]
\end{proof}

Summing the bounds of the previous lemma yields
\begin{equation}\label{boundU}
 \alg(\vec{p},\vec{w}) \leq
 \sum_{1\leq j \leq k \leq n} w_j w_k\ F\bigl(\min(s_j, s_k),\max(s_j,s_k)\bigr) \eqqcolon U(\vec{p},\vec{w}).
\end{equation}
We next prove a lemma showing that for bounding the ratio~$U/\opt$ we can restrict to instances in which all Smith ratios are integer powers of $b$.

\begin{lemma}\label{lem:Smith_intger_power_b}
For any instance $(\vec{p},\vec{w})$, there exists
another instance $(\vec{p}',\vec{w})$ with $p_j'=w_j b^{q_j}$ for some $q_j\in\N_0$, for all $j\in [n]$, such that
\[
 \frac{U(\vec{p},\vec{w})}{\opt(\vec{p},\vec{w})}\leq \frac{U(\vec{p}',\vec{w})}{\opt(\vec{p}',\vec{w})}.
\]
\end{lemma}
\begin{proof}
Let $\rho_{\min} = \min_{j \in [n]} \lfloor \log_b(s_j) \rfloor$. Then in the instance $(b^{-\rho_{\min}} \vec p, \vec w)$ all jobs have Smith ratio $\ge 1$, and it holds that $U(b^{-\rho_{\min}} \vec p, \vec w) = b^{-\rho_{\min}} \cdot U(\vec p, \vec w)$. Clearly, we also have $\opt(b^{-\rho_{\min}} \vec p, \vec w) = b^{-\rho_{\min}} \cdot \opt(\vec p, \vec w)$. Therefore, without loss of generality, we assume that $s_j \ge 1$ for all $j \in [n]$. Let $S \coloneqq \bigl\{ \log_b(s_j) \mid  j\in[n]\bigr\}\setminus \mathbb{Z}$. If $S=\emptyset$, we are done. Otherwise,
let $q \coloneqq \min(S) > 0$ and $I\coloneqq \{ j\in [n] \mid s_j = b^q \}$. We will either decrease
the Smith ratio of each job $j\in I$
to $b^{\lfloor q \rfloor}$ or increase them to $b^{q'}$, where $q' \coloneqq \min\bigl(S \cup \{\lceil q \rceil\} \setminus \{q\}\bigr)$,
so that the cardinality of $S$ is decreased by $1$, and repeat this operation until each Smith ratio is an integer power of $b$.
For $\delta\in\R$, define $p_j(\delta) \coloneqq p_j+ \delta \cdot w_j \mathds{1}_I(j)$,
so the Smith ratio of job~$j$ in the instance $(\vec{p}(\delta),\vec{w})$ is $s_j(\delta)=s_j + \delta$ if $j\in I$ and 
$s_j(\delta)=s_j$ otherwise.%
\newcommand{\bunderline}[1]{\underline{#1\mkern-3mu}\mkern3mu }
Let $\bunderline{\delta} \coloneqq b^{\lfloor q \rfloor}-b^q<0$ and $\bar{\delta}\coloneqq b^{q'}-b^q>0$.
Since an optimal schedule follows the \wspt{} rule, it is easy to see that the function
$\delta\mapsto \opt(\vec{p}(\delta),\vec{w})$ is linear in the interval $[\bunderline{\delta},\bar{\delta})$, as the order of the Smith ratios remains unchanged for all $\delta$ in this interval.
For the same reason,
and because for all $j,k$ the Smith ratios $s_j(\delta),s_k(\delta)$ remain in the same piece of the piecewise linear function $(s_j,s_k) \mapsto F(\min(s_j,s_k), \max(s_j,s_k))$ for all $\delta\in[\bunderline{\delta},\bar{\delta})$,
the function $\delta\mapsto U(\vec{p}(\delta),\vec{w})$ is also linear over $[\bunderline{\delta},\bar{\delta})$.
As a result, the function
\[h \colon \delta \mapsto \frac{U(\vec{p}(\delta),\vec{w})}{\opt(\vec{p}(\delta),\vec{w})}\]
is a quotient of linear functions and thus monotone over 
$[\bunderline{\delta},\bar{\delta})$. Indeed, $\opt(\vec{p}(\delta),\vec{w})>0$ for all $\delta \geq \bunderline{\delta}$, so $h$ has no pole in this interval. We can thus distinguish two cases: if the function $h$ is non-increasing, 
we let $\delta'=\bunderline{\delta}$, so
we have $h(\delta')\geq h(0)$, which means that we can decrease the Smith ratio of each job $j\in I$ to $s_j(\delta')=s_j+\bunderline{\delta}=b^q+b^{\lfloor q\rfloor}-b^q=b^{\lfloor q\rfloor} \ge 1$ without decreasing the bound $U/\opt$ on the competitive ratio. Otherwise, the function $h$ is non-decreasing, 
hence $h(0)\leq \lim_{\delta\nnearrow\bar{\delta}} h(\delta) \leq h(\bar{\delta})$,
where the last inequality comes from the fact that $F$ is non-decreasing.
So in this case we set $\delta'=\bar{\delta}$
and we can increase the Smith ratio of
each job $j\in I$ to $b^{q'}$ without decreasing the bound on the competitive ratio; if 
$q'=\lceil q \rceil$, it means that we round up these Smith ratios to the next integer power of $b$, otherwise it is $q'=\min(S\setminus\{q\})$, so we 
cluster together a larger group of jobs with a Smith ratio $s_j(\bar{\delta})=b^{q'}$ that is not an integer power of $b$. 

In all cases, the number of distinct non-integer values of $\log_b(p_j(\delta')/w_j)$ is decremented by one compared to the original instance, while the bound on the competitive ratio is only larger:
\[\frac{U(\vec{p}(\delta'),\vec{w})}{\opt(\vec{p}(\delta'),\vec{w})}
\geq 
\frac{U(\vec{p}(0),\vec{w})}{\opt(\vec{p}(0),\vec{w})}.\]
Repeating this construction until all Smith ratios are integer powers of $b$ yields the desired result.
\end{proof}

The next lemma gives a handy upper bound for the competitive ratio of $\alg$ relying on the ratio of two quadratic forms. For $L\in \mathbb{N}_0$ define the symmetric $((L+1) \times (L+1))$-matrices
$\vec{A}_L
\coloneqq \big( \frac{1}{2} b^{\min(\ell,m)} \mathds{1}_{\{\ell \neq m\}}\big)_{0\leq \ell,m \leq L}$
and $\vec{B}_L \coloneqq \big( \frac{1}{2} b^{\min(\ell,m)}\big)_{0\leq \ell,m \leq L}$.

\begin{lemma} \label{lem:ratio_of_quads}
For any instance $(\vec{p},\vec{w})$ there exists an integer $L$ and a vector $\vec{x}\in\R^{\{0,\dotsc,L\}}$
such that
\[\frac{\alg(\vec p, \vec w)}{\opt(\vec p, \vec w)} \leq \frac{2b}{b-1} + 1 + b\cdot \frac{\vec{x}^\top \vec{A}_L \vec{x}}{\vec{x}^\top \vec{B}_L \vec{x}}.
\]
\end{lemma}

\begin{proof}
Consider an arbitrary instance $(\vec{p},\vec{w})$.
It follows from~\eqref{boundU} that
$\frac{\alg(\vec{p},\vec{w})}{\opt(\vec{p},\vec{w})}\leq \frac{U(\vec{p},\vec{w})}{\opt(\vec{p},\vec{w})}
$. By Lemma~\ref{lem:Smith_intger_power_b}, we construct an instance $(\vec{p}',\vec{w})$ in which each Smith ratio is a non-negative integer power of $b$, and such that
$\frac{U(\vec{p},\vec{w})}{\opt(\vec{p},\vec{w})}\leq \frac{U(\vec{p}',\vec{w})}{\opt(\vec{p}',\vec{w})}
$ holds. 
In the remainder of this proof, we relabel the jobs so that
$\frac{p_1'}{w_1}\leq \cdots \leq \frac{p_n'}{w_n}$.
We define $L \coloneqq \max_{j \in [n]} \log_b(s'_j)$.
For all $\ell=0,\dotsc,L$, we denote by 
$J_\ell \coloneqq \{ j\in[n] \mid  p_j'=w_jb^\ell \}$ the subset of jobs with Smith ratio equal to $b^\ell$, so by construction we have $[n]=J_0\cup\dotsb\cup J_L$.
We also define $x_\ell \coloneqq \sum_{j\in J_\ell} w_j$ and 
$y_\ell \coloneqq \sum_{j\in J_\ell} w_j^2$, for all $\ell=0,\dotsc,L$.

We first get a handy expression for $\opt(\vec{p}',\vec{w})$ relying on the vectors $\vec{x},\vec{y}\in\R^{\{0,\dotsc,L\}}$. 
By optimality of the \wspt{} rule,
\begin{equation}\label{eq:OPT_intg_Smith}\begin{split}
 \opt(\vec{p}',\vec{w}) = \sum_{k=1}^n w_{k} \sum_{j=1}^k p'_j
 &=\sum_{j=1}^n p_j' \sum_{k=j}^n w_{k} \\
 &=\sum_{\ell=0}^L\sum_{j\in J_\ell} p_j' \sum_{k=j}^n w_k \\
 &=\sum_{\ell=0}^L b^\ell \sum_{j\in J_\ell} w_j \sum_{k=j}^n w_{k} \\
 &=\sum_{\ell=0}^L b^\ell \Bigg(\sum_{\substack{j,k\in J_\ell\\ j\leq k}} w_j w_{k} + \sum_{m=\ell+1}^L\sum_{j\in J_\ell}w_j\sum_{k\in J_m}
w_{k} \Bigg) \\
&=\sum_{\ell=0}^L b^{\ell}
\left(\frac12 y_\ell + \frac12 x_\ell^2 +\sum_{m=\ell+1}^L x_m x_\ell
\right)
=\sum_{\ell=0}^L \frac12 b^{\ell} y_\ell
+ \vec{x}^\top \vec{B}_L \vec{x}.
\end{split}\end{equation}
On the other hand, using the fact that $F(b^\ell,b^m)= b^\ell F(1, b^{m-\ell}) = b^\ell F(1,b^{\min(m-\ell,1)})$ for all integers $\ell\leq m$, we obtain
\begin{equation}\label{eq:U_intg_Smith}\begin{split}
 U(\vec{p}',\vec{w}) &= \sum_{\ell=0}^L \sum_{j\in J_\ell} \left(\sum_{\substack{k\in J_\ell\\ k\geq j}} w_j w_{k} F(b^\ell,b^\ell)
 +\sum_{m=\ell+1}^L\sum_{k\in J_m} w_j w_{k} F(b^\ell,b^m)\right)\\
 &=\sum_{\ell=0}^L \Biggl(
 b^{\ell} F(1,1) 
 \sum_{\substack{j,k\in J_\ell\\ j\leq k}} w_j w_{k}
 +b^{\ell} F(1,b) 
 \sum_{m=\ell+1}^L\sum_{j\in J_\ell} w_j
 \sum_{k\in J_m} w_{k} 
 \Biggr)\\
 &= \sum_{\ell=0}^L \left( b^\ell  
 F(1,1) \cdot 
 \Bigl(\frac12 y_\ell + \frac12 x_\ell^2\Bigr)
 +b^{\ell} F(1,b) \sum_{m=\ell+1}^L x_mx_\ell\right)\\
 &=
 F(1,1) \cdot \bigg(\sum_{\ell=0}^L \frac12 b^\ell y_\ell
+ \vec{x}^\top \vec{B}_L \vec{x}\bigg) + [F(1,b)-F(1,1)]
\cdot \sum_{\ell=0}^L \sum_{m=\ell+1}^L b^\ell x_m x_\ell\\
&=F(1,1)\cdot \opt(\vec{p}',\vec{w})
+[F(1,b)-F(1,1)] \cdot \vec{x}^\top \vec{A}_L \vec{x}.
\end{split}\end{equation}
Substituting $F(1,1)=\frac{2b}{b-1}+1$ and $F(1,b)=\frac{2b}{b-1}+b+1$, we get
\[
 \frac{U(\vec{p}',\vec{w})}{\opt(\vec{p}',\vec{w})}
 = \frac{2b}{b-1}+1 + b \cdot\frac{\vec{x}^\top \vec{A_L}\vec{x}}{\sum_{\ell=0}^L \frac12 b^\ell y_\ell
+ \vec{x}^\top \vec{B}_L \vec{x}}
\leq \frac{2b}{b-1}+1 + b \cdot\frac{\vec{x}^\top \vec{A_L}\vec{x}}{\vec{x}^\top \vec{B}_L \vec{x}}.\qedhere
\]
\end{proof}

In order to determine an upper bound for the competitive ratio of $\alg$, we need to bound the last term in the expression from \cref{lem:ratio_of_quads}. The latter is the ratio of two quadratic forms, and an upper bound for this term can be derived by computing the maximum eigenvalue of the matrix $\vec{Z}_L \coloneqq \vec{Y}_L^{-\top} \vec{A}_L \vec{Y}_L^{-1}$, where $\vec{B}_L = \vec{Y}_L^\top \vec{Y}_L$ is the Cholesky decomposition of the matrix~$\vec{B}_L$. 
An explicit computation of the matrix $\vec{Z}_L$ reveals that it is a tridiagonal matrix whose principal submatrix---obtained by deleting the first row and first column---is a (tridiagonal) Toeplitz matrix that we refer to as $\vec{T}_{L}$.
Finding an upper bound for the largest eigenvalue of $\vec{Z}_L$ is the main ingredient of the proof of \cref{thm:upperbound}, while the eigenvector corresponding to this eigenvalue can be used to construct instances that prove the tightness of the bound (\cref{thm:ALG_LB}).

\DetALGUB

\begin{proof}[Proof of \cref{thm:upperbound}]
By \Cref{lem:ratio_of_quads} we have
\begin{equation}\label{bound_CRdet_rho}
 \sup_{\vec p, \vec w \in \R_{\ge 0}^n} \frac{\alg(\vec p, \vec w)}{\opt(\vec p, \vec w)} \leq 1 + \frac{2b}{b-1}+b\cdot \sup_{L\in\mathbb{N}} \sup_{\vec{x}\in\R^{\{0,\dotsc,L\}}} 
\frac{\vec{x}^{\top} \vec{A}_L\vec{x}}{\vec{x}^{\top} \vec{B}_L\vec{x}}
\end{equation}
As described above, for every $L \in \N$, 
\[
	\sup_{\vec{x}\in\R^{\{0,\dotsc,L\}}} 
	\frac{\vec{x}^{\top} \vec{A}_L\vec{x}}{\vec{x}^{\top} \vec{B}_L\vec{x}} = \lambda_{\max}(\vec Z_L)
\]
is the maximum eigenvalue of the matrix $\vec Z_L = \vec Y_L^{-\top} \vec A_L \vec Y_L^{-1}$. For $\alpha, \beta \in \R$ let $\vec T_L(\alpha, \beta) \in \R^{L \times L}$ denote the symmetric tridiagonal Toeplitz matrix with $\alpha$ on the main diagonal and $\beta$ on both adjacent diagonals. The explicit representation of $\vec Z_L$ can be derived by applying \cref{supremum ratio} in the appendix with $a_1 = \cdots = a_L = 1$. As many terms in the general form cancel out, this yields the tridiagonal matrix
		\[
		\vec{Z}_{L} = 
		\begin{pmatrix}
			0 & \frac{1}{\sqrt{b-1}} &           & & & &\\
			\frac{1}{\sqrt{b-1}}& -\frac{2}{b-1} & \frac{\sqrt{b}}{b-1}  & & & &\\
			& \frac{\sqrt{b}}{b-1} & -\frac{2}{b-1} & \frac{\sqrt{b}}{b-1} & & &\\
			&          & \ddots & \ddots & \ddots &\\
			&          &        & \frac{\sqrt{b}}{b-1} & -\frac{2}{b-1} & \frac{\sqrt{b}}{b-1}\\                  
			&          &        &          & \frac{\sqrt{b}}{b-1} & -\frac{2}{b-1}    
		\end{pmatrix}=\begin{pmatrix} 0 & \frac{1}{\sqrt{b-1}} \vec e_1^\top \\ \frac{1}{\sqrt{b-1}}\vec e_1 & \vec{T}_{L}\bigl(-\frac{2}{b-1}, \frac{\sqrt b}{b-1}\bigr)\end{pmatrix}.
		\]
	 We now want to show that $\lambda_{\max}(\vec Z_L) \le \frac{2 (\sqrt b - 1)}{b-1}$. 
  This is equivalent to the matrix $\vec H \coloneqq \frac{2 (\sqrt b - 1)}{b-1} \vec{I}_{L+1} - \vec{Z}_{L}$ being positive semidefinite, where $\vec I_{L+1}$ denotes the identity matrix with indices $\{0,\dotsc,L\}$. We compute
  \[\vec H = \begin{pmatrix} \frac{2 (\sqrt b - 1)}{b-1} & -\frac{1}{\sqrt{b -1}} \vec e_1^\top \\ -\frac{1}{\sqrt{b -1}} \vec e_1 & \vec{T}_{L}\bigl(\frac{2\sqrt b}{b-1},-\frac{\sqrt b}{b-1}\bigr)\end{pmatrix} = \frac{\sqrt b}{b-1} \begin{pmatrix} 2 - \frac{2}{\sqrt b} & -\sqrt{1-\frac 1 b}\vec e_1^\top \\ -\sqrt{1-\frac 1 b}\vec e_1 & \vec T_L(2,-1) \end{pmatrix}.
   \]
 This matrix has the form required in \cref{lem:generalized_Toeplitz_pos_semidefinite} with $k=1$, $\alpha = 2-\frac{2}{\sqrt b}$, and $\vec v = -\sqrt{1-\frac 1 b} \in \R^1$. Since $\alpha - \Vert \vec v \Vert^2 = 2 - \frac{2}{\sqrt b} - \bigl(1 - \frac 1 b\bigr) = \frac{(\sqrt b - 1)^2}{b} \ge 0$, the \lcnamecref{lem:generalized_Toeplitz_pos_semidefinite} implies that $\vec H$ is positive semidefinite, so that $\lambda_{\max}(\vec Z_L) \le \frac{2(\sqrt b -1)}{b-1}$. Since this holds for every $L \in \N$, we obtain with inequality~\eqref{bound_CRdet_rho}
		\begin{align*}
			\frac{\alg}{\opt} 
			&\leq 1 + \frac{2b}{b-1} + \frac{2b(\sqrt{b}-1)}{b-1}
			= 1+\frac{2b^{\frac{3}{2}}}{b-1}.
		\end{align*}
		The latter is minimized for $b=3$ yielding the performance guarantee of $1 + 3 \sqrt{3} $.
\end{proof}

Next, we show that our analysis of $\alg$ is asymptotically tight.

\DetALGLB

\begin{proof}
For $L \geq 1$ let $\vec{Y}_L$, and $ \vec{Z}_{L}$ be the matrices defined above, and let $\vec T_L \coloneqq \vec T_L(-\frac{2}{b-1}, \frac{\sqrt b}{b-1})$ be the principal submatrix of $\vec Z_L$ and $\vec{z}_L = (z^{(L)}_\ell)_{0 \le \ell \le L}$ with $z^{(L)}_{\ell} \coloneqq \sqrt{\frac{2}{L+1}} \cdot \sin\bigl(\frac{\ell \pi}{L + 1}\bigr)$.
By \cite[Theorem 2.4]{BG05}, $\tilde{\vec{z}}_{L} \coloneqq (z^{(L)}_{\ell})_{1 \leq \ell \leq L} $ is the eigenvector of the matrix $\vec{T}_{L}$ corresponding to the largest eigenvalue $\lambda_{\max} (\vec{T}_{L}) = - \frac{2}{b-1} + 2\frac{\sqrt{b}}{b-1} \, \cos \bigl(\frac{\pi}{L+1} \bigr)$, and we have $\Vert \vec{z}_L \Vert^2 = \Vert \tilde{\vec{z}}_{L} \Vert^2 = \frac{2}{L+1} \sum_{\ell=1}^{L} \sin^2 \big(\frac{\ell \pi}{L+1}\big) = 1$.
Define $\vec{x}_L = (x_\ell^{(L)})_{0 \le \ell \le L} \coloneqq \vec{Y}_L^{-1} \vec{z}_L$. 
By construction, it holds 
\begin{equation}\label{limit_quotient_quadratic_forms}
\frac{\vec{x}_L^\top \vec{A}_L \vec{x}_L}{\vec{x}_L^\top \vec{B}_L \vec{x}_L}
=
\frac{\vec{z}_L^\top \vec{Z}_L \vec{z}_L}{\|\vec{z}_L\|^2}
=
\tilde{\vec{z}}_{L}^\top \vec{T}_{L} \tilde{\vec{z}}_{L}
=\lambda_{\max} (\vec{T}_{L}) = - \frac{2}{b-1} + 2\frac{\sqrt{b}}{b-1} \, \cos \Bigl(\frac{\pi}{L+1} \Bigr)
\xrightarrow{L \to \infty} \frac{2(\sqrt{b}-1)}{b-1}.
\end{equation}
The idea is to define for every $L \in \N$ an instance~$\vec p_L$ via a non-negative integer vector~$\vec{n}_L \in \N_0^{L}$ that is similar to $\vec{x}_L$, which contains $n_{\ell}$ jobs with processing time~$b^{\ell} + \epsilon$ for all $\ell\in [L]$ and for some $\epsilon>0$. There is an $\ell^* \in \N_{>0}$ such that $x^{(L)}_{\ell} \geq 0$ holds for all $L \ge \ell \geq \ell^*$. This follows from \cref{lem:ellstar} in the appendix because the matrix~$\vec Y$ has exactly the required form, as shown in \cref{lem:Cholesky}. Therefore, for $L \ge \ell^*$ the vector $\vec{n}_L = (n^{(L)}_{\ell})_{0 \leq \ell \leq L}$ with $n_\ell^{(L)} \coloneqq 0$ for $\ell < \ell^*$ and $n^{(L)}_\ell \coloneqq \lfloor b^L x^{(L)}_\ell \rfloor$ for $\ell \ge \ell^*$ is a non-negative integer vector, so that for every $\varepsilon \ge 0$ the instance~$\vec p_L(\varepsilon) = (p_j^{(L)}(\varepsilon))$ consisting of $n_\ell^{(L)}$ jobs with processing time $b^\ell + \varepsilon$ for $\ell = 0,\dotsc,L$, ordered non-increasingly by processing times, is well-defined. Let $n^{(L)} \coloneqq \sum_{\ell=0}^L n^{(L)}_\ell$ be the number of jobs in $\vec p_L(\varepsilon)$.
 
 Let now $\varepsilon > 0$.
 Clearly, we have
	\begin{equation}\label{eq:forgetful:opt_cost_LB_instance}
		\opt(\vec p_L(\varepsilon)) = \opt(\vec p_L(0)) + \sum_{j=1}^{n^{(L)}}(n^{(L)}-j+1) \cdot \varepsilon
		= \opt(\vec p_L^*) + \frac{n^{(L)}(n^{(L)}+1)}{2}\cdot \varepsilon.
	\end{equation}
 For every job~$j$ let $q_j \coloneqq \lceil \log_b(p_j^{(L)}) \rceil$, i.e., $q_j = \ell + 1$ for the $n_\ell$ jobs with processing time~$b^\ell + \varepsilon$. Proceeding similarly as in the proof of \cref{lem:DeltaF}, we obtain
$\Delta_{jj}=\frac{b^{q_j}}{b-1}+p_j^{(L)}=b^{q_j-1} (\frac{b}{b-1}+1) + \varepsilon$,
and for $j \neq k$ with $q_j \le q_k$ we have
$\Delta_{jk}=\frac{2b^{q_j}}{b-1}+p_j^{(L)}=b^{q_j-1}(\frac{2b}{b-1}+1)+\varepsilon$ if $q_j=q_k$ or
$\Delta_{jk}=\frac{2b^{q_j}}{b-1}+b^{q_j}+p_j^{(L)}=b^{q_j-1}(\frac{2b}{b-1}+b+1)+\varepsilon$ otherwise.
This can be rewritten as 
\[
 \Delta_{jk} =
\begin{cases}
 F(p_j^{(L)}(0),p_j^{(L)}(0)) + \varepsilon -\frac{b}{b-1} p_j^{(L)}(0) & \text{ if } j=k,\\
 F(p_j^{(L)}(0),p_k^{(L)}(0)) + \varepsilon & \text{ if } j \neq k \text{ and } q_j\leq q_k.
\end{cases}
\]
Summing over all pairs of jobs, we thus obtain
\[
 \alg(\vec{p}_L(\varepsilon))=U(\vec{p}_L(0))+\frac{n^{(L)}(n^{(L)}+1)}{2} \cdot \varepsilon -\frac{b}{b-1} \sum_{j=1}^{n^{(L)}} p_j^{(L)}(0)
\]

As the processing times of $\vec p_L(0)$ are integer powers of $b$, we can use \cref{eq:U_intg_Smith,eq:OPT_intg_Smith} with $\vec{y}=\vec{x}=\vec{n}_L$, resulting in
\begin{align*} \lim_{\varepsilon \searrow 0} \frac{\alg(\vec p_{L}(\varepsilon))}{\opt(\vec p_L(\varepsilon))} = \frac{U(\vec p_L(0)) - \frac{b}{b-1} \sum_{j=1}^{n^{(L)}} p_j^{(L)}(0)}{\opt(\vec p_L(0))} &= \frac{2b}{b-1} + 1 + \frac{b \cdot \vec n_L^\top \vec A_L \vec n_L - \frac{b}{b-1} \sum_{\ell=0}^{L} n_\ell^{(L)} b^\ell}{\sum_{\ell=0}^L \frac 1 2 n_\ell^{(L)} b^\ell + \vec n_L^\top \vec B_L \vec n_L} \\
&= \frac{2b}{b-1} + 1 + \frac{b^{1-2L} \vec n_L^\top \vec A_L \vec n_L - \frac{b^{1-2L}}{b-1} \sum_{\ell=0}^L n_\ell^{(L)} b^\ell}{\frac 1 2 b^{-2L} \sum_{\ell=0}^L n_\ell^{(L)} b^\ell + b^{-2L} \vec n_L^\top \vec B_L \vec n_L}.\end{align*}
In the following we compute the limit for $L \to \infty$ by computing the limits of the occurring terms separately. By \cref{lem:ellstar}, $\sum_{\ell=\ell^*}^L x_\ell^{(L)} \xrightarrow{L\to \infty} 0$.
Therefore, we have
	\[
	\frac{\sum_{\ell = 0}^L n^{(L)}_\ell b^{\ell}}{ b^{2L}}
	\leq \frac{n^{(L)} b^{L}}{b^{2L}} = \frac{n^{(L)}}{b^{L}} = \frac{\sum_{\ell=\ell^*}^L \lfloor b^L x^{(L)}_\ell \rfloor}{b^L} \le \sum_{\ell=\ell^*}^L x^{(L)}_\ell \xrightarrow{L \to \infty} 0
	.
	\]
For $L \ge \ell^*$ let $\vec{x}_1^{(L)} = (x_\ell^{(L)})_{0 \le \ell \le \ell^*-1}$, $\vec{x}^{(L)}_2 = (x_\ell^{(L)})_{\ell^* \le \ell \le L}$, $\vec{A}_{11}^{(L)} = (A_{\ell m}^{(L)})_{0 \le \ell, m \le \ell^* - 1}$, $\vec{A}_{12}^{(L)} = (A_{\ell m}^{(L)})_{\substack{0 \le \ell \le \ell^*-1\\ \ell^* \le m \le L}}$, and $\vec{A}_{22}^{(L)} = (A_{\ell m}^{(L)})_{\ell^* \le \ell,m \le L}$, so that 
	\[
	\vec{x}_L = 
	\begin{pmatrix}
		\vec{x}^{(L)}_1 \\ \vec{x}^{(L)}_2
	\end{pmatrix}
	\quad\text{and}\quad
	\vec{A}_L = 
	\begin{pmatrix}
		\vec{A}_{11}^{(L)} & \vec{A}_{12}^{(L)} \\
		\bigl( \vec{A}_{12}^{(L)} \bigr)^{\top} & \vec{A}_{22}^{(L)}
	\end{pmatrix}
	,
	\]
 and let $\vec n_2^{(L)} = (n_\ell^{(L)})_{\ell^* \le \ell \le L}$.
	With the definition of $\vec{n}_L$, we compute
	\[
		\vec{x}_L^{\top} \vec{A}_L \vec{x}_L - b^{-2L} \vec{n}_L^{\top} \vec{A}_L \vec{n}_L
		= \big({\vec{x}}_1^{(L)} \big)^\top \vec{A}_{11}^{(L)} \vec{x}_1^{(L)} + 2 \bigl(\vec{x}_1^{(L)}\bigr)^\top \vec{A}_{12}^{(L)} \vec{x}^{(L)}_2 + \Bigl(\vec x_2^{(L)} - \frac{\vec n_2^{(L)}}{b^L} \Bigr)^\top \vec A_{22}^{(L)} \Bigl(\vec x_2^{(L)} - \frac{\vec n_2^{(L)}}{b^L} \Bigr).
	\]
 For the first summand we have
	\[
		\Bigl\vert\big(\vec x_1^{(L)}\big)^\top\vec A_{11}^{(L)} \vec x_1^{(L)}\Bigr\vert \leq\sum_{0 \leq \ell < \ell'\leq \ell^*-1}\vert x_\ell^{(L)}\vert\vert x_{\ell'}^{(L)}\vert b^\ell
  \le \frac{\ell^* (\ell^* + 1)}{2} \cdot \frac{4(\sqrt{b}+1)^2}{(L+1) (b-1)} \xrightarrow{L \to \infty} 0,
  	\]
	and for the second summand, by \cref{lem:ellstar}, we have
	\[
		2\bigl\vert \bigl(\vec x_1^{(L)}\bigr)^\top \vec A_{12}^{(L)} \vec x_2^{(L)} \bigr\vert = \Biggl\vert\bigg(\sum_{\ell=0}^{\ell^*-1} x_\ell^{(L)} b^\ell \bigg) \bigg(\sum_{\ell=\ell^*}^L x_\ell^{(L)} \bigg)\Biggr\vert
		\le \ell^* \cdot \frac{2(\sqrt{b}+1)b^{(\ell^*-1)/2}}{\sqrt{(L+1)(b-1)}} \cdot \bigg(\sum_{\ell=\ell^*}^L x_\ell^{(L)} \bigg) \xrightarrow{L \to \infty} 0.
	\]
 Finally, the third summand satisfies
 \begin{align*}
  \biggl\vert \Bigl(\vec x_2^{(L)} - \frac{\vec n_2^{(L)}}{b^L} \Bigr)^\top \vec A_{22}^{(L)} \Bigl(\vec x_2^{(L)} - \frac{\vec n_2^{(L)}}{b^L} \Bigr) \biggr\vert
  &\le \sum_{\ell^* \le \ell < \ell' \le L} \Bigl\vert x_\ell^{(L)} - \frac{n_\ell^{(L)}}{b^L} \Bigr\vert \Bigl\vert x_{\ell'}^{(L)} - \frac{n_{\ell'}^{(L)}}{b^L} \Bigr\vert b^\ell \\
  &\le \frac{(L-\ell^*+2)^2}{2} \cdot \Bigl\vert x_\ell^{(L)} - \frac{\lfloor b^L x_\ell^{(L)} \rfloor}{b^L} \Bigr\vert \Bigl\vert x_{\ell'}^{(L)} - \frac{\lfloor b^L x_{\ell'}^{(L)} \rfloor}{b^L} \Bigr\vert b^L \\
  &\le \frac{(L-\ell^* + 2)^2}{2b^L} \xrightarrow{L \to \infty} 0.
 \end{align*}
Similarly, we have
	\[
		\vec{x}_L^{\top} \vec{B}_L \vec{x}_L - b^{-2L} \vec{n}_L^{\top} \vec{B}_L \vec{n}_L \xrightarrow{L \to \infty} 0.
	\]
As $\vec{x}_L^\top \vec{B}_L \vec{x}_L = \Vert \vec z_L \Vert = 1 \neq 0$ for all $L$, we have thus shown that
	\[
  \lim_{L \to \infty} \lim_{\varepsilon \to 0}\frac{\alg(\vec p_L(\varepsilon))}{\opt(\vec p_L(\varepsilon))} = \lim_{L \to \infty} \frac{2b}{b-1} + 1 + b \cdot \frac{\vec x_L^\top \vec A_L \vec x_L}{\vec x_L^\top \vec B_L \vec x_L}
  \stackrel{\eqref{limit_quotient_quadratic_forms}}= \frac{2b}{b-1} + 1 + \frac{2b (\sqrt b - 1)}{b-1} = 1 + \frac{2b^{\frac 3 2}}{b-1}.
 \]
 By \cref{lem:double limit} there is a sequence of problem instances for which the competitive ratio converges to the right hand side.
\end{proof}

\subsection{\texorpdfstring{Tight Analysis of the Randomized $b$-Scaling Strategy}{Tight Analysis of the Randomized b-Scaling Strategy}}\label{sec:random}
\newcommand{\indicator}[1]{\mathds{1}_{\{#1\}}}

\newcommand{\xiLargerEta}{\indicator{\Xi > \eta}}

We now consider $\ralg = \alg^{\Sigma, \Xi}$, where $\Sigma$ is a permutation drawn uniformly at random from $\mathcal{S}_n$ and $\Xi$ is uniformly distributed on the interval $[0,1]$.

As in the analysis of the deterministic algorithm, we start with a lemma giving an overestimator of $\E[\Delta_{jk}]$ for jobs $j$ and $k$ such that $s_j\leq s_k$. This time, our overestimator is not piecewise linear in $s_j$ and $s_k$ anymore. Instead, it depends on a concave function applied to the ratio $\frac{s_k}{s_j}\geq 1$.  

\begin{lemma}\label{lem:bound_f}
  Let $f \colon [1,b] \to \R$ be defined by
 \[
 	f(\alpha)\coloneqq \frac{1+\alpha}{2}+\frac{2}{\ln b}+\frac{\alpha-1}{2\ln b} \cdot 
	(1-\ln\alpha)
	.
\]
Then $f$ is positive and increasing, and for all $j\neq k$ such that $s_j\leq s_k$ it holds
 \[
  \E[\Delta_{jj}]= w_j^2\ s_j\cdot \Big(1+\frac{1}{\ln b}\Big)\leq w_j^2\ s_j\cdot f(1)
  \quad\text{and}\quad
  \E[\Delta_{jk}]= w_j w_k s_j \cdot f\left(\min\Big(b,\frac{s_k}{s_j}\Big)\right).
   \]
 \end{lemma}
\begin{proof}
 By straightforward calculus, we obtain
 \[
  f'(\alpha)=\frac{1-\alpha \ln(\nicefrac{\alpha}{b})}{2\alpha\ln b}
  \quad \textrm{and}\quad 
  f''(\alpha) = -\frac{\alpha+1}{2\alpha^2\ln b}.
 \]
 Since $f''(\alpha) < 0$, the function $\alpha\mapsto f'(\alpha)$ is decreasing over $[1,b]$. Hence, for all $\alpha\in[1,b]$ we have $f'(\alpha)\geq f'(b)=\frac{1}{2b\ln b} > 0$, which proves that $f$ is increasing. Therefore, $f(\alpha) \ge f(1) = 1 + \frac{2}{\ln b} > 0$, concluding the proof of the first part of the \lcnamecref{lem:bound_f}.

 Let now $j, k \in [n]$, $j \neq k$ with $s_j \le s_k$ be fixed, and let $\ell \in \mathbb{Z}$ and $u \in (0,1]$ be such that
 $s_j=b^{\ell+u}$. Moreover, let $q_j(\Xi) \coloneqq \lceil \log_b(s_j) - \Xi \rceil = \ell + \mathds{1}_{\{\Xi < u\}}$ be such that $j$ completes in the round, where jobs are probed for $w_j b^{q_j(\Xi) + \Xi}$. We thus have
 \[
  D_{jj} = \sum_{i=-\infty}^{q_j(\Xi)-1} w_j b^{i+\Xi} + p_j
  =w_j\Bigl(\frac{b^{q_j(\Xi) + \Xi}}{b-1} + s_j\Bigr) = w_j \Bigl(\frac{b^{\ell + \mathds 1_{\{\Xi < u\}} + \Xi}}{b-1} + s_j\Bigr).
 \]
Therefore, we get
\begin{align*}
 \E[\Delta_{jj}] = w_j^2 \left(s_j + \frac{b^{\ell}}{b-1}\int_{0}^1 b^{\xi+\mathds{1}_{\{\xi<u\}}}\, \mathrm d \xi\right)
 &=w_j^2 \left(s_j+
 \frac{b^{\ell}}{b-1}
 \left( 
 \int_{0}^{u} b^{\xi+1}\, \mathrm d \xi
 + \int_{u}^1 b^{\xi}\, \mathrm d \xi
 \right)\right)\\
 &=w_j^2 \left(s_j + \frac{b^{\ell}}{(b-1)\ln b}
 \bigl[b(b^{u}-1)+(b-b^{u})\bigr]\right)\\
 &=w_j^2 \left(s_j + \frac{b^{\ell+u}}{\ln b}\right)
  = w_j^2\ s_j\Bigl(1+\frac{1}{\ln b}\Bigr).
\end{align*}
Now, we fix a realization $\xi$ of $\Xi$ and compute
$\E[\Delta_{jk}|\Xi=\xi]$ for another job $k$ such that $s_j\leq s_k$.
If $q_j(\xi)=q_k(\xi)$, then
$\E[D_{jk}|\Xi=\xi]=w_j \frac{b^{q_j(\xi)+\xi}}{b-1}+\frac{1}{2}\cdot p_j$ and 
$\E[D_{kj}|\Xi=\xi]=w_k \frac{b^{q_j(\xi)+\xi}}{b-1}+\frac{1}{2}\cdot p_k$,
where the factors $\frac{1}{2}$ in front of $p_j$ and $p_k$ come from the fact that job $j$ is completed with probability $\frac{1}{2}$ before job $k$ due to the random permutation of the jobs. Otherwise, it is $q_j(\xi)<q_k(\xi)$ and we have 
$\E[D_{jk}|\Xi=\xi]=w_j \frac{b^{q_j(\xi)+\xi}}{b-1}+p_j$,
$\E[D_{kj}|\Xi=\xi]=w_k \frac{b^{q_j(\xi)+\xi}}{b-1}+w_k \frac{b^{q_j(\xi)+\xi}}{2}$. Putting all together,
\[
 \E[\Delta_{jk}|\Xi=\xi]=w_j w_k \cdot \left( 2\frac{b^{\ell+\mathds{1}_{\{\xi<u\}}+\xi}}{b-1} + \left\{
 \begin{array}{ll}
  \frac{s_j+s_k}{2} & \text{ if } q_j(\xi)=q_k(\xi)\\
  s_j+\frac{b^{\ell+\mathds{1}_{\{\xi<u\}}+\xi}}{2} & \text{ if } q_j(\xi)<q_k(\xi)
 \end{array}
 \right. \right).
\]
Let us first consider the case $s_k\geq b \cdot s_j$, as in this case $q_j(\xi)<q_k(\xi)$ for all $\xi \in (0,1)$. Thus,
\begin{align*}
 \E[\Delta_{jk}] &= w_j w_k \left(s_j + b^{\ell}\cdot \Big(\frac{2}{b-1}+\frac{1}{2}\Big)\cdot \int_{0}^1 b^{\xi + \mathds{1}_{\{\xi<u\}}}\, \mathrm d \xi\right)\\
 &= w_j w_k \left(s_j + b^{\ell}\cdot \Big(\frac{2}{b-1}+\frac{1}{2}\Big)\cdot \frac{b^{u} (b-1)}{\ln b}\right)\\
 &= w_j w_k s_j \biggl(1+\frac{1}{\ln b}\Bigl(2 + \frac{b-1}{2}\Bigr)\biggr) = w_j w_k s_j\cdot f(b)=w_j w_k s_j\cdot f\Bigl(\min\Bigl(b,\frac{s_k}{s_j}\Bigr)\Bigr).
\end{align*}
It remains to handle the case $s_j \leq s_k < b \cdot s_j$, in which it can occur that jobs $j$ and $k$ are completed in the same round. Let $\delta\in[0,1)$ such that
$s_k=s_j\cdot b^{\delta}=b^{\ell+u+\delta}$. To compute $\E[\Delta_{jk}]$, we have to distinguish between the cases $u+\delta \leq 1$ and $u+\delta > 1$. We only handle the former case, as the latter can be handled similarly and yields the same formula, so in the remainder of this proof we assume $u+\delta \leq 1$. Then,
it holds $q_k(\xi)=\ell+\mathds{1}_{\{\xi<u+\delta\}}$,
so that $(q_j(\xi),q_k(\xi))=(\ell+1,\ell+1)$ if $\xi\in [0,u)$,
$(q_j(\xi),q_k(\xi))=(\ell,\ell+1)$ if $\xi\in[u,u+\delta)$, and
$(q_j(\xi),q_k(\xi))=(\ell,\ell)$ if $\xi\in[u+\delta,1]$. Thus
we can write
\begin{align*}
 \E[\Delta_{jk}] &= w_jw_k \left(\int_{0}^u 2\frac{b^{\ell+1+\xi}}{b-1} +\frac{s_j+s_k}{2}\, \mathrm d \xi
 +\int_{u}^{u+\delta} b^{\ell+\xi}\Bigl(\frac{2}{b-1}+\frac{1}{2}\Bigr) +s_j\, \mathrm d \xi
 +\int_{u+\delta}^1 2\frac{b^{\ell+\xi}}{b-1} +\frac{s_j+s_k}{2}\, \mathrm d \xi\right)\\
 &= w_j w_k \left(\frac{s_j}{2}(1+\delta) + \frac{s_k}{2} (1-\delta)
 +\frac{b^\ell}{\log b}\Bigl[\textstyle\frac{2}{b-1}(b^{u+1}-b)+\Big(\frac{2}{b-1}+\frac{1}{2}\Big)(b^{u+\delta}-b^u)+\frac{2}{b-1}(b-b^{u+\delta})\Bigr]\right)\\
 &= w_j w_k \left(\frac{s_j}{2}(1+\delta) + \frac{s_k}{2} (1-\delta)
 +\frac{b^\ell}{\log b}\bigg[b^u \Big(2+\frac{b^{\delta}-1}{2}\Big)\bigg]\right)\\
 &= w_j w_k s_j \bigg(\frac{1+\delta}{2}+b^\delta\frac{1-\delta}{2}+\frac{1}{\log b} \Big(2+\frac{b^{\delta}-1}{2}\Big)\bigg)\\
 &= w_j w_k s_j \cdot f(b^\delta) =w_j w_k s_j \cdot f\Bigl(\frac{s_k}{s_j}\Bigr) =w_j w_k s_j\cdot f\Bigl(\min\Bigl(b,\frac{s_k}{s_j}\Bigr)\Bigr).\qedhere
\end{align*}
\end{proof}

The expressions of $\E[\Delta_{jk}]$ derived in the previous lemma show that $\ralg(\alpha \vec{p},\vec{w})=\alpha \ralg(\vec{p},\vec{w})$
 holds for all $\alpha>0$.
 Since the same trivially holds for the optimal solution, i.e., 
 $\opt(\alpha \vec{p},\vec{w})=\alpha \opt(\vec{p},\vec{w})$
 we can assume without loss of generality that the instance has been rescaled, so that $\min_{j\in[n]} s_j=1$. Moreover, we relabel the jobs so that $s_1 \le \cdots \le s_n$. Then, summing the bounds from the previous \lcnamecref{lem:bound_f} yields
 \begin{equation}\label{eq:def_U_rand}
  \ralg(\vec p, \vec w) \le \sum_{1 \le j \le k \le n} w_j w_k s_j \cdot f\biggl(\min\Bigl(b, \frac{s_k}{s_j}\Bigr)\biggr) \eqqcolon U(\vec p, \vec w).
 \end{equation}

In order to obtain an upper bound, we use a similar technique as in \cref{subsec:Analysis_Alg_det}. However, the proof is more involved because the bound from \cref{lem:bound_f} is not piecewise linear in $s_j$ and $s_k$, so we cannot construct a worse instance $(\vec{p}',\vec{w})$ in which all Smith ratios are integer powers of $b$. Instead, we are going to subdivide the Smith ratios in intervals of the form $[b^{i/K},b^{(i+1)/K})$ for some integer $K$, and we will get a bound by grouping all jobs in an interval. As in \cref{thm:upperbound}, this bound involves a ratio of two quadratic forms (\cref{lem:ratio_of_quads_rand}), but this time the maximization of this fraction amounts to finding the maximum eigenvalue of a banded Toeplitz matrix of bandwidth $2K-1$.

Let $K \in \N_{>0}$ and $\beta \coloneqq b^{\frac{1}{K}}$. 
For $L \in \N$ define the symmetric matrices $\vec A_L \coloneqq (\frac 1 2 a_{|m-\ell|} \beta^{\min(\ell,m)})_{0 \le \ell, m \le L}$, where $a_i \coloneqq f(\beta^{\min(K,i+1)}) - f(\beta)$ for $i \in \{0,\dotsc,L\}$, and let $\vec B_L \coloneqq (\frac 1 2 \beta^{\min(\ell,m)})_{0 \le \ell, m \le L}$.

\begin{lemma} \label{lem:ratio_of_quads_rand}
 For any instance ($\vec p, \vec w)$ there exists $L \in \N$ and a vector $\vec x \in \R^{\{0,\dotsc,L\}}$ such that
 \[\frac{\ralg(\vec p,\vec w)}{\opt(\vec p, \vec w)} \le \beta \biggl(f(\beta) + \frac{\vec x^\top \vec A_L \vec x}{\vec x^\top \vec B_L \vec x}\biggr).\]
\end{lemma}

\begin{proof}
 Let $L$ be a multiple of $K!$ such that $s_j \le \beta^L = b^{L/K}$ for all $j \in [n]$. For all $\ell \in \{0,\dotsc,L\}$ we define
 $J_\ell \coloneqq \{j \in [n] \mid \beta^{\ell}\leq s_j < \beta^{\ell+1}\}$, $x_\ell \coloneqq \sum_{j\in J_\ell} w_j$,
 and $y_\ell \coloneqq \sum_{j\in J_\ell} w_j^2$. Then $[n] = J_0 \cup \dotsb \cup J_L$. We obtain as a lower bound on the optimal cost
 \begin{align*}
  \opt(\vec p, \vec w) = \sum_{j=1}^n p_j \sum_{k=j}^n w_{k} &= \sum_{\ell=0}^L \sum_{j \in J_\ell} w_j s_j \biggl(\sum_{\substack{k \in J_\ell \\ k \ge j}} w_{k} + \sum_{m=\ell+1}^L \sum_{k \in J_m} w_k\biggr) \\
  &\ge \sum_{\ell=0}^L \beta^\ell \biggl(\sum_{\substack{j, k \in J_\ell \\ j \le k} } w_j w_{k} + \sum_{m=\ell+1}^L \sum_{j \in J_\ell} w_j \sum_{k \in J_m} w_{k}\biggr) \\&= \sum_{\ell=0}^L \beta^\ell \biggl(\frac 1 2 y_\ell + \frac 1 2 x_\ell^2 + \sum_{m=\ell+1}^L x_\ell x_m \biggr) = \sum_{\ell=0}^L \frac 1 2 \beta^\ell y_\ell + \vec x^\top \vec B_L \vec x.
 \end{align*}
 
 On the other hand, using \eqref{eq:def_U_rand}, we compute
 \begin{align*}
		\ralg(\vec p, \vec w)&\le U(\vec p, \vec w) = \ \sum_{1 \le j \le k \le n} w_j w_k s_j f\biggl(\min\Bigl(b, \frac{s_k}{s_j}\Bigr)\biggr) \\
		&\leq\ \sum_{\ell=0}^L \sum_{j\in J_{\ell}}\biggl(w_j^2 s_j \underbrace{f(1)}_{\leq f(\beta)} + \sum_{\substack{k\in J_{\ell}:\\ k> j}}w_jw_k s_j f\biggl(\min\Bigl( b,\underbrace{\frac{s_k}{s_j}}_{\mathclap{\leq \beta}}\Bigr)\biggr) \\
		& \ \quad \quad \quad \quad \quad + \sum_{m=\ell+1}^L \sum_{k\in J_m} w_jw_ks_j f\biggl(\min\Bigl( b,\frac{s_k}{s_j}\Bigr)\biggr) \biggr) \\
		&\leq\  \sum_{\ell=0}^L \sum_{j\in J_{\ell}}\bigg(\sum_{\substack{k\in J_{\ell}:\\ k \ge j}}w_jw_k\beta^{\ell+1} f(\beta) + \sum_{m=\ell+1}^L \sum_{k\in J_m} w_jw_k\beta^{\ell+1} f\bigl(\beta^{\min(K,m-\ell+1)}\bigr) \bigg) \\
		& =\ \beta \sum_{\ell=0}^L \biggl( f(\beta) \beta^\ell \Bigl(\frac 1 2 y_{\ell} + \frac{1}{2} x_{\ell}^2\Bigr) + \sum_{m=\ell+1}^L x_{\ell}x_m\beta^{\ell}f\big(\beta^{\min(K,m-\ell+1)}\big)\biggr) \\
		&\le\ \beta \biggl( f(\beta)\cdot \opt(I) + \sum_{\ell=0}^L \sum_{m=\ell+1}^L x_{\ell}x_m\beta^{\ell}\bigl(f\bigl(\beta^{\min(K,m-\ell+1)}\bigr)-f(\beta)\bigr)\biggr) \\
	 & =\ \beta \bigl( f(\beta)\cdot\opt(I) + \vec x^\top \vec A_L \vec x \bigr) 
		.
	\end{align*}
 Consequently,
\[\frac{\ralg(\vec p, \vec w)}{\opt(\vec p, \vec w)} \le \beta \biggl(f(\beta) + \frac{\vec x^\top \vec A_L \vec x}{\vec x^\top \vec B_L \vec x + \frac 1 2 \sum_{\ell=0}^L \beta^\ell y_\ell}\biggr) \le \beta \biggl(f(\beta) + \frac{\vec x^\top \vec A_L \vec x}{\vec x^\top \vec B_L \vec x}\biggr). \qedhere\]
\end{proof}

We next prove the main result of this section.

\RandALGUB

\begin{proof}
 By \cref{lem:ratio_of_quads_rand} we have for every $K \in \N_{>1}$ and $\beta = b^{\frac 1 K}$ that
 \begin{equation} \label{ineq:upper_bound_rand}
  \sup_{\vec p, \vec w \in \R_{\ge 0}^n} \frac{\ralg(\vec p, \vec w)}{\opt(\vec p, \vec w)} \le \beta \biggl(f(\beta) + \sup_{L \in \N} \sup_{\vec x \in \R^{\{0,\dotsc,L\}}} \frac{\vec x^\top \vec A_L \vec x}{\vec x^\top \vec B_L \vec x}\biggr).
 \end{equation}
 For now let $K \in \N_{>1}$ be fixed. Similarly to the proof of \cref{thm:upperbound}, \cref{supremum ratio} yields that the inner supremum is
 \[\sup_{\vec x \in \R^{\{0,\dotsc,L\}}} \frac{\vec x^\top \vec A_L \vec x}{\vec x^\top \vec B_L \vec x} = \lambda_{\max}(\vec Z_L),\]
for the matrix
\[
 \vec{Z}_L = \left(\begin{array}{c|c}
            0   & \vec{u}^\top\ \  \vec{0}^\top\\\hline
            \begin{matrix}\vec{u}\\\vec{0}\end{matrix} & \vec{T}_L
           \end{array}\right),
\]
where $\vec{u}\in\mathbb{R}^{K-1}$ has coordinates 
$u_i=\frac{f(\beta^{i+1})-f(\beta^i)}{\sqrt{\beta^i-\beta^{i-1}}}$, $i=1,\ldots,K-1$, and $\vec{T}_L$ is the $L\times L$ banded symmetric Toeplitz matrix of bandwidth $2(K-1)+1$ with elements
\begin{equation}\label{tk_f}
 t_k \coloneqq \left\{\begin{array}{ll}
            \frac{-2}{\beta-1} [f(\beta^2)-f(\beta)] & \text{ if }k=0,\\
            \frac{1}{\beta^{k/2}(\beta-1)} [(\beta+1)f(\beta^{k+1})-f(\beta^{k+2})-\beta f(\beta^k)] & \text{ if }1\leq k\leq K-2,\\
            \frac{1}{\beta^{(K-3)/2}(\beta-1)}[f(\beta^K)-f(\beta^{K-1})] & \text{ if }k= K-1.
               \end{array}
\right.
\end{equation}
on its $k$th and $-k$th superdiagonals (where for $k\in\mathbb{Z}$, the $k$th superdiagonal is the set of coordinates $(\vec T_L)_{\ell,m}$ such that $m-\ell=k$; in particular, the $0$th superdiagonal corresponds to the main diagonal).
Substitution of $f(\beta^k)$ with its value $\frac 1 2 \bigl(1+\frac k K + \beta^k \bigl(1 - \frac k K\bigr) + \frac{\beta^k+3}{\ln b}\bigr)$ yields the following simplified expression for $t_k$:
\begin{equation}\label{tk_simpl}
 t_k = \left\{\begin{array}{ll}
            \frac{1+2\beta}{K}-\beta(1+\frac{1}{\ln b}) & \text{ if }k=0,\\
            \frac{1+\beta^{1+k}}{2K\beta^{k/2}} & \text{ if }1\leq k\leq K-2,\\
            \frac{1}{2\beta^{(K-1)/2}} (\frac{b}{\ln b}-\frac{b-\beta}{K(\beta-1)}) & \text{ if }k= K-1.
               \end{array}
\right.
\end{equation}

We next show that $\lambda_{\max}(\vec{Z}_L)\leq t_0 + 2\sum_{k=1}^{K-1} t_k$. To this end, we form the matrix 
\[\vec{H} \coloneqq \biggl(t_0 + 2\sum_{k=1}^{K-1} t_k \biggr) \vec{I} - \vec{Z}_L\] and prove that this matrix is positive semidefinite.
For every $i \in [K-1]$ the sum~$\sum_{k=i}^{K-1} \frac{t_k}{\beta^{k/2}}$ is a telescoping sum, which sums up to $\frac{f(\beta^{i+1}) - f(\beta^i)}{\beta^i - \beta^{i-1}} = \frac{u_i}{\sqrt{\beta^i-\beta^{i-1}}}$. In particular, $\sum_{k=1}^{K-1} \frac{t_k}{\beta^{k/2}} = \frac{f(\beta^2)-f(\beta)}{\beta-1}$, so that $t_0 = -2 \sum_{k=1}^{K-1} \frac{t_k}{\beta^{k/2}}$. Therefore, we can rewrite $\vec{H}$ as a linear combination of $t_1,\ldots,t_{K-1}$:
\[
 \vec{H} = \sum_{k=1}^{K-1} t_k \vec{H}_k,\qquad \text{where}\quad
 \vec{H}_k\coloneqq \left(\begin{array}{c|c}
            2\bigl(1-\frac{1}{\beta^{k/2}}\bigr)   & \vec{v}_k^\top\ \  \vec{0}^\top\\\hline
            \begin{matrix}\vec{v}_k\\\vec{0}\end{matrix} & \vec{T}_{k,L}
           \end{array}\right),
\]
$\vec{v}_k\in \R^k$ is a vector with coordinates $(\vec{v}_k)_i= -\frac{\sqrt{\beta-1}}{\beta^{(k-i+1)/2}}$, ($i=1,\ldots,k$),
 and $\vec{T}_{k,L}$ is the sparse symmetric Toeplitz matrix of size $L\times L$ whose only non-zero elements are $2$ on the main diagonal and $-1$ on the $k$th and $-k$th superdiagonals, i.e., $(\vec{T}_{k,L})_{ij}= 2\cdot \mathds{1}_{\{i=j\}} - \mathds{1}_{\{|i-j|=k\}}$. To show that $\vec{H}$ is positive semidefinite, it suffices to show that $t_k\geq 0$ and $\vec{H}_k$ is positive semidefinite for all $k\in[K-1]$.

Let $k\in[K-1]$. If $k=K-1$, then $t_{k}\geq 0$ follows from~\eqref{tk_f} and the fact that $f$ is non-decreasing over $[1,b]$. Otherwise,
this inequality follows from \cref{tk_simpl}. Next, we use \cref{lem:generalized_Toeplitz_pos_semidefinite} to show that $\vec{H}_k$ is positive semidefinite. This is possible because $k\mid L$, as $L$ is a multiple of $K!$, and
\[
 2\left(1-\frac{1}{\beta^{k/2}}\right) - \|\vec{v}_k\|^2
 =  2\left(1-\frac{1}{\beta^{k/2}}\right) 
 - \sum_{i=1}^k \frac{\beta-1}{\beta^{k-i+1}} 
 =  2\left(1-\frac{1}{\beta^{k/2}}\right) 
 - 1+\frac{1}{\beta^k}=(\beta^{k/2}-1)^2 \beta^{-k} \geq 0.
\]
This concludes the proof that all $\vec{H}_k$ are positive semidefinite, hence $\lambda_{\max}(\vec{Z})\leq t_0+2\sum_{k=1}^{K-1} t_k$. Together, we have shown that
\[\frac{\ralg}{\opt} \le \beta \biggl(f(\beta) + t_0 + 2 \sum_{k=1}^{K-1} t_k\biggr)\]
for every $K \in \N$. 

In the final part of the proof we show that the right-hand side converges to $\frac{2b + \sqrt b- 1}{\sqrt b \ln b}$ for $K \to \infty$.
To this end, let us now compute the sum 
\begin{align*}
\sum_{k=1}^{K-2} t_k=\frac{1}{2K}\bigg(\sum_{k=1}^{K-2} \beta^{-k/2}+\beta\sum_{k=1}^{K-2} \beta^{k/2}\bigg)
&=\frac{1}{2K}\left(\frac{\beta^{-(K-1)/2}-\beta^{-1/2}}{\beta^{-1/2}-1}+\beta \frac{\beta^{(K-1)/2}-\beta^{1/2}}{\beta^{1/2}-1}\right) \\
&=\frac{1}{2K}\left(\frac{b^{-\frac{K-1}{2K}}-b^{-\frac{1}{2K}}}{b^{-\frac{1}{2K}}-1}+ \frac{b^{\frac{K+1}{2K}}-b^{\frac{3}{2K}}}{b^{\frac{1}{2K}}-1}\right) \\
&\xrightarrow{K\to\infty} \frac{b^{-\frac12}-1}{-\ln b}+
\frac{b^{\frac12}-1}{\ln b}=\frac{b-1}{\sqrt{b}\ln b}.
\end{align*}
Moreover, we have $t_0\xrightarrow{K\to\infty}-(1+\frac{1}{\ln b})$ and $t_{K-1}=\frac12 b^{-\frac{K-1}{2K}}\bigl(\frac{b}{\ln b}-\frac{b-b^{\frac1K}}{K(b^{\frac1K}-1)}\bigr)
\xrightarrow{K\to\infty} \frac{1}{2\sqrt{b}}(\frac{b}{\ln b}-\frac{b-1}{\ln b})=\frac{1}{2\sqrt{b}\ln b}$.
By using $\beta\xrightarrow{K\to\infty} 1$
and $f(\beta)\xrightarrow{K\to\infty} f(1)=1+\frac{2}{\ln b}$, we obtain the final bound by taking the limit when $K\to\infty$:
\begin{align*}
 \frac{\ralg}{\opt}\leq \lim_{K\to\infty} \beta\biggl(f(\beta)+t_0+2\sum_{k=1}^{K-1} t_k\biggr)
 &=1+\frac{2}{\ln{b}}
 -\Bigl(1+\frac{1}{\ln{b}}\Bigr)+\frac{2(b-1)}{\sqrt{b}\ln b}+\frac{1}{\sqrt{b}\ln b}\\
 &= \frac{\sqrt{b}+2b-1}{\sqrt{b}\ln b}.
\end{align*}
A numerical minimization yields an optimal value of $b \approx 8.16$ with a performance guarantee smaller than $3.032$.
\end{proof}

We next show that our analysis is tight.

\TheoLBRand

\begin{proof}
 For $K \in \N$ let $\beta = b^{1/K}$, and for $L \in \N$ let $\vec A_L' \coloneqq (\frac 1 2 a_{|m-\ell|}' \beta^{\min(\ell,m)})_{0 \le \ell, m \le L}$, where $a_i' \coloneqq f(\beta^{\min(K,i)}) - f(1)$ for $i \in \{0,\dotsc,L\}$. Moreover let $\vec B_L = \vec Y_L^\top \vec Y_L$ be the Cholesky decomposition of $\vec B_L$. By \cref{supremum ratio}, we have
 \[\vec Y_L^{-\top} \vec A'_L \vec Y_L^{-1} = \left(\begin{array}{c|c}
  0   & (\vec{u}')^\top\ \  \vec{0}^\top\\\hline
  \begin{matrix}\vec{u}'\\\vec{0}\end{matrix} & \vec{T}_L'
 \end{array}\right) \eqqcolon \vec Z_L'
\]
for some $\vec{u}'\in\mathbb{R}^{K}$, where $\vec{T}'_L$ is the $L\times L$ banded symmetric Toeplitz matrix of bandwidth $2K+1$ with elements
\begin{equation}
 t_k' \coloneqq \left\{\begin{array}{lll}
  \frac{-2}{\beta-1} [f(\beta)-f(1)] 
  &=\frac1K-(1+\frac{1}{\ln b})
  & \text{ if }k=0,\\
  \frac{1}{\beta^{k/2}(\beta-1)} [(\beta+1)f(\beta^{k})-\beta f(\beta^{k-1})-f(\beta^{k+1})]
  &=\frac{1+\beta^k}{2K\beta^{k/2}}
  & \text{ if }1\leq k\leq K-1,\\
  \frac{1}{\beta^{K/2-1}(\beta-1)}[f(\beta^K)-f(\beta^{K-1})] 
  & =\frac{1}{2\beta^{K/2}} (\frac{b}{\ln b}-\frac{b-\beta}{K(\beta-1)})
  & \text{ if }k= K.
     \end{array}
\right.
\end{equation}
on the $k$th and $-k$th superdiagonals.
As in the proof of
 \cref{thm:ALG_LB} for the deterministic version of the strategy, define 
 $\vec{z}_L = (z_\ell^L)_{0 \le \ell \le L}$ with $z_{\ell}^{(L)} \coloneqq \sqrt{\frac{2}{L+1}} \cdot \sin \big( \frac{\ell \pi}{L+1} \big)$ for $\ell=0,\dotsc,L$,
 and 
 $\vec{x}_L \coloneqq \vec{Y}_L^{-1} \vec{z}_L \in \R^{\{0,\dotsc,L\}}$. By construction,
 \begin{equation}
  \vec{x}_L^\top \vec{A}'_L \vec{x}_L
  =
  \vec{z}_L^\top \vec{Z}'_L \vec{z}_L
  =
  \tilde{\vec{z}}^\top \vec{T}_L' \tilde{\vec{z}},
  \label{eq:qudratic_Aprime}
 \end{equation}
 where $\tilde{\vec{z}}\coloneqq [z_1,z_2,\ldots,z_{L}]^\top\in\R^L$ and we have used the fact that $z_0=0$ for the last equality. Furthermore, 
 \begin{equation}
  \vec x_L^\top \vec B_L \vec x_L = \lVert \vec z_L \rVert = 1. \label{eq:quadratic_B}
 \end{equation}
 Unlike the proof of \cref{thm:ALG_LB} however, it is not true anymore that $\tilde{\vec{z}}$ is the
 eigenvector corresponding to the largest eigenvalue of $\vec{T}'_L$ because $\vec{T}'_L$ is not tridiagonal.

 \Cref{lem:Cholesky,lem:ellstar} imply that there is an $\ell^* \in \N_{>0}$ such that $x_\ell^{(L)} \ge 0$ for all $L \ge \ell \ge \ell^*$. Therefore, for $L \ge \ell^*$ the vector $\vec{n}_L = (n^{(L)}_{\ell})_{\ell=0}^L$ with $n_\ell^{(L)} \coloneqq 0$ for $\ell < \ell^*$ and $n^{(L)}_\ell \coloneqq \lfloor \beta^L x^{(L)}_\ell \rfloor$ for $\ell \ge \ell^*$ is a non-negative integer vector, defining the instance~$\vec p_L$ that consists of $n_\ell^{(L)}$ jobs with processing time $\beta^\ell$ for $\ell = 0,\dotsc,L$.

 For such an instance, we have
 \begin{equation}
  \opt(\vec{p}_L)= \sum_{\ell=0}^L \beta^\ell \biggl(
  \frac{n_\ell^{(L)} (n_\ell^{(L)}+1)}{2}  +
  \sum_{m=\ell+1}^L n_\ell^{(L)} n_m^{(L)} \biggr)
  =\sum_{\ell=0}^L \frac{n_\ell^{(L)}}{2} \beta^\ell + \vec{n}_L^\top \vec{B}_L \vec{n}_L \label{eq:OPT_randomized}
 \end{equation}
 and
 \[
  \ralg(\vec{p}_L)=\sum_{\ell=0}^L \Bigl(\frac{n_\ell^{(L)} (n_\ell^{(L)}+1)}{2} \beta^\ell f(1) -\frac{n_\ell^{(L)} \beta^\ell}{\ln b}\Bigr) + \sum_{\ell=0}^L \sum_{m=\ell+1}^L n_\ell^{(L)} n_m^{(L)} \beta^\ell f(\beta^{\min(m-\ell,K)}).
 \]
 Then, proceeding similarly as in the proof of \cref{lem:ratio_of_quads}, we obtain
 \begin{equation}\begin{split}
 \ralg(\vec{p}_L)
 &= f(1)\cdot \opt(\vec{p}_L)
 - \sum_{\ell=0}^L \frac{n_\ell^{(L)} \beta^\ell}{\ln b} 
 + \sum_{\ell=0}^L \sum_{m=\ell+1}^L n_\ell^{(L)} n_m^{(L)} \beta^\ell \bigl(f(\beta^{\min(K, m-\ell)}) - f(1)\bigr) \\
 &= f(1) \cdot \opt(\vec{p}_L) 
 -\sum_{\ell=0}^L \frac{n_\ell^{(L)} \beta^\ell}{\ln b} + \vec{n}_L^\top \vec{A}'_L \vec{n}_L.
 \end{split}\label{eq:ralg}\end{equation}
 Therefore, 
 \begin{align*}
  \frac{\ralg(\vec p_L)}{\opt(\vec p_L)} &\;\stackrel{\eqref{eq:ralg}}= \; f(1) + \frac{\vec n_L^\top \vec A_L' \vec n_L - \sum_{\ell=0}^L \frac{n_\ell^{(L)}}{\ln b} \beta^\ell}{\opt(\vec p_L)} \\ 
  &\;\stackrel{\eqref{eq:OPT_randomized}}= \; f(1) + \frac{\vec n_L^\top \vec A_L' \vec n_L - \sum_{\ell=0}^L \frac{n_\ell^{(L)}}{\ln b} \beta^\ell}{\vec n_L^\top \vec B_L \vec n_L + \sum_{\ell=0}^L \frac{n_\ell^{(L)}}{2} \beta^\ell} 
  = f(1) + \frac{\beta^{-2L} \vec n_L^\top \vec A_L' \vec n_L - \sum_{\ell=0}^L \frac{n_\ell^{(L)}}{\ln b} \beta^{\ell-2L}}{\beta^{-2L} \vec n_L^\top \vec B_L \vec n_L + \sum_{\ell=0}^L \frac{n_\ell^{(L)}}{2} \beta^{\ell-2L}} \\ 
  &\;\stackrel{\mathclap{\eqref{eq:qudratic_Aprime},}\hphantom{\eqref{eq:qudratic_Aprime}}\mathclap{\eqref{eq:quadratic_B}}}= \; f(1) + \frac{\tilde{\vec z}_L^\top \vec T'_L \tilde{\vec z}_L + \bigl(\frac{\vec n_L}{\beta^L} - \vec x_L\bigr)^\top \vec A_L' \bigl(\frac{\vec n_L}{\beta^L} - \vec x_L\bigr) - \sum_{\ell=0}^L \frac{n_\ell^{(L)}}{\ln b} \beta^{\ell-2L}}{1 + \bigl(\frac{\vec n_L}{\beta^L} - \vec x_L\bigr)^\top \vec B_L \bigl(\frac{\vec n_L}{\beta^\ell} - \vec x_L\bigr) + \sum_{\ell=0}^L \frac{n_\ell^{(L)}}{2} \beta^{\ell-2L}}.
 \end{align*}

 As in the proof of \cref{thm:ALG_LB}, we compute the limits of the occurring terms for $L \to \infty$.
 \[0 \le \sum_{\ell=0}^L n_\ell^{(L)} \beta^{\ell-2L} \le \sum_{\ell=0}^L n_\ell^{(L)} \beta^{-L} = \frac{\sum_{\ell=\ell^*}^L \lfloor \beta^L x_\ell^{(L)} \rfloor}{\beta^L} \le \sum_{\ell=\ell^*}^L x_\ell^{(L)} \xrightarrow{L \to \infty} 0,\]
 where the convergence follows from \cref{lem:ellstar}. Thus, we see that the last summands of numerator and denominator go towards zero. Next, using that for $\ell \neq m$ the absolute value of the entry of $\vec A_L'$ indexed by $\ell$ and $m$ is bounded by $\frac 1 2 f(b) \beta^{\min(\ell, m)}$, exactly the same calculation as in the proof of \cref{thm:ALG_LB} shows that
 \[\Bigl(\frac{\vec n_L}{\beta^L} - \vec x_L\Bigr)^\top \vec A_L' \Bigl(\frac{\vec n_L}{\beta^L} - \vec x_L\Bigr) \xrightarrow{L \to \infty} 0 \quad\text{ and, analogously, }\quad \Bigl(\frac{\vec n_L}{\beta^L} - \vec x_L\Bigr)^\top \vec B_L \Bigl(\frac{\vec n_L}{\beta^L} - \vec x_L\Bigr) \xrightarrow{L \to \infty} 0.\]
 We are going to show that 
 $\tilde{\vec{z}}_L^\top \vec{T}'_L \tilde{\vec{z}}_L$
 converges to $t_0'+2\sum_{k=1}^K t_k'$ as the dimension $L$ grows to $\infty$. For this, define the function $\Phi \colon \theta \mapsto t_0' + 2\sum_{k=1}^K t_k'\cos(kx)$, which is the Fourier Series associated with the Toeplitz matrix $\vec{T}'_L$. We claim that for all $\ell\in[L]$, it holds
 \[
 (\vec{T}'_L\tilde{\vec{z}}_L)_\ell \geq \tilde{z}_\ell^{(L)} \cdot \Phi\Bigl(\frac{\pi}{L+1}\Bigr). 
 \]
 To see this, we first extend the definition of 
 $\tilde{z}_{\ell}^{(L)} =
 \sqrt{\frac{2}{L+1}} \sin \big( \frac{\ell \pi}{L+1} \big)$ to all $\ell\in\{1-K,\ldots,L+K\}$, and we observe that 
 for all $\ell\in[L]$ and $k\in[K]$ we have
 \begin{align*}
  \tilde{z}_{\ell-k}^{(L)} + \tilde{z}_{\ell+k}^{(L)} &= 
  \sqrt{\frac{2}{L+1}} \left[ \sin \Bigl( \frac{(\ell-k) \pi}{L+1} \Bigr) + \sin \Bigl( \frac{(\ell+k) \pi}{L+1} \Bigr)\right] \\
  &= \sqrt{\frac{2}{L+1}} \cdot 
  2\cdot \sin \Bigl( \frac{\ell \pi}{L+1} \Bigr)
  \cdot
  \cos \Bigl( \frac{k \pi}{L+1} \Bigr)
  =2\tilde{z}_\ell^{(L)} \cdot \cos \Bigl( \frac{k \pi}{L+1} \Bigr). 
 \end{align*}
 Then, we use the fact that $\tilde{z}_\ell\leq 0$ for all $\ell \in\{1-K,\ldots,0\} \cup \{L+1,\ldots,L+K\}$, so we have 
 \begin{align*}
  (\vec{T}'_L \tilde{\vec{z}}_L)_\ell &= t_0' \tilde{z}_\ell^{(L)} + \sum_{k=1}^K t_k' (\tilde{z}_{\ell-k}^{(L)} \mathds{1}_{\{\ell-k\geq 1\}}+\tilde{z}_{\ell+k}^{(L)} \mathds{1}_{\{\ell+k\leq L\}})\\
  &\geq 
  t_0' \tilde{z}_\ell^{(L)} + \sum_{k=1}^K t_k' (\tilde{z}_{\ell-k}^{(L)} +\tilde{z}_{\ell+k}^{(L)})\\
  &=\tilde{z}_\ell^{(L)} \cdot \biggl(t_0'+2\sum_{k=1}^K t_k'\cos\Bigl(\frac{k\pi}{L+1}\Bigr)\biggr)=\tilde{z}_\ell^{(L)} \cdot \Phi\Bigl(\frac{\pi}{L+1}\Bigr).
 \end{align*}
 Consequently, using the fact that $\tilde{z}_\ell^{(L)}$ is positive for all $\ell\in [L]$ we obtain 
 \[
  \tilde{\vec{z}}_L^\top \vec{T}'_L \tilde{\vec{z}}_L \geq \|\tilde{\vec{z}}_L\|^2 \cdot \Phi\Bigl(\frac{\pi}{L+1}\Bigr)
  =\Phi\Bigl(\frac{\pi}{L+1}\Bigr) \xrightarrow{L\to\infty} \Phi(0)=t_0'+2\sum_{k=1}^K t_k'.
 \]

 We have thus constructed for every $K \in \N$ a sequence~$\vec p_L$ for which the competitive ratio converges to $t_0' + \sum_{k=1}^K t_k'$. Further, similar calculations as in the proof of \cref{thm:Rand_ALG_UB} show that
 $t_0'+2\sum_{k=1}^K t_k'\xrightarrow{K\to\infty} -1-\frac{1}{\ln b}+\frac{2b-1}{\sqrt{b}\ln b}$. By \cref{lem:double limit} there is a sequence of problem instances 
 for which the competitive ratio of $\ralg$ converges to the desired value of \[f(1) -1-\frac{1}{\ln b}+\frac{2b-1}{\sqrt{b}\ln b} = \frac{\sqrt{b}+2b-1}{\sqrt{b}\ln b}. \qedhere\]
\end{proof}


\section{Weighted Shortest Elapsed Time First} \label{sec:WSETF}

In this \lcnamecref{sec:WSETF} we consider the online time model, where each job~$j$ arrives at its release date~$r_j$ and is not known before that time. Thus, an instance for our problem is now given by a triple $I = (\vec p, \vec w, \vec r)$ of processing times, weights, and release dates of all jobs. We consider the classical Weighted Shortest Elapsed Time First (\wsetf) rule for this model. Intuitively, \wsetf{} is the limit for $\varepsilon \to 0$ of the algorithm that divides the time into time slices of length~$\varepsilon$ and in each time slice processes a job with minimum ratio of elapsed processing time over weight. To formalize this limit process we allow fractional schedules~$\mathrm S$ that, at every point in time~$t$, assign each job~$j$ a rate~$y_j^{\mathrm S}(t) \in [0,1]$ such that $\sum_{j=1}^n y_j^{\mathrm S}(t) \le 1$ for all $t \in \R_{\ge 0}$ and $y_j^{\mathrm S}(t) = 0$ if $t < r_j$ or $t > C_j^{\mathrm S}(I)$, where $C_j^{\mathrm S}(I)$ is the smallest $t$ such that $Y_j^{\mathrm S}(I, t) \coloneqq \int_0^t y_j^{\mathrm S}(s)\,\mathrm d s \ge p_j$ (this requires $y_j^{\mathrm S}$ to be measurable).
At any time~$t$ let $J(t)$ be the set of all released and unfinished jobs, and let $A(t)$ be the set of all jobs from $J(t)$ that currently have minimum ratio of elapsed time over weight. Then \wsetf{} sets the rate for all jobs~$j \in A(t)$ to
\[y_j^{\wsetf}(t) \coloneqq \begin{cases*} w_j/\bigl(\sum_{k \in A(t)} w_k\bigr) &if $j \in A(t)$,\\ 0 &else. \end{cases*}\]
In other words, \wsetf{} always distributes the available processor rate among the jobs in $J(t)$ so as to maximize $\min_{j \in J(t)} Y_j^{\wsetf}(I, t)/w_j$. An example is given in \cref{fig:sec5_example_release_dates}.

The following \lcnamecref{thm:WSETF} gives the tight competitive ratio of \wsetf{} for non-clairvoyant online scheduling on a single machine.

\WSETF

We start by collecting some simple properties of the schedule created by \wsetf. 

\begin{lemma} \label{Order WSETF}
 Consider an instance $I = (\vec p, \vec w, \vec r)$, and let $j, k$ be two jobs with $r_k < C_j^{\wsetf}(I)$ and $p_k/w_k \le p_j/w_j$. Then $C_j^{\wsetf}(I) \ge C_k^{\wsetf}(I)$.
\end{lemma}
\begin{proof}
 Suppose that $C_j^{\wsetf}(I) < C_k^{\wsetf}(I)$. The job~$j$ must be processed at a positive rate during some interval $(t, C_j^{\wsetf}(I))$ by the \wsetf{} schedule, meaning that $j \in A(C_j^{\wsetf}(I))$. Hence, we get the contradiction \[\frac{p_j}{w_j} = \frac{Y_j^{\wsetf}(I, C_j^{\wsetf}(I))}{w_j} \le \frac{Y_k^{\wsetf}(I, C_j^{\wsetf}(I))}{w_k} < \frac{p_k}{w_k},\] where the last inequality holds because $k$ completes after $C_j^{\wsetf}(I)$.
\end{proof}

The \lcnamecref{Order WSETF} implies that in an instance with trivial release dates, for which \wsetf{} coincides with the Weighted Round-Robin algorithm, analyzed by \citet{KC03}, the jobs~$j$ are completed in the order of their Smith ratios~$p_j/w_j$. In this case the weighted delay of each job in the \wsetf{} schedule compared to the optimal \wspt{} schedule is exactly its processing time multiplied with the total weight of jobs with larger index.

\begin{lemma} \label{lemma:WSETF_no_release}
 Let $I_0 = (\vec{p},\vec{w},\vec{0})$ be an instance with trivial release dates and $p_1/w_1 \le \cdots \le p_n/w_n$. For every job~$j \in [n]$ we have 
 \[
 w_j \cdot C_j^{\wsetf}(I_0) = w_j \cdot C_j^{\wspt}(I_0) + {}\hypertarget{weightedDelay}{\underbrace{\sum_{k=j+1}^n w_k \cdot p_j}_{(\text{\textasteriskcentered})}}. \label{eq:WSETF_no_release}
 \]
\end{lemma}
\begin{proof}
 This will be shown by induction on $n$. Clearly, the statement is true if there is only a single job. So in the following let $n > 1$. We have
 \[
 w_1 \cdot C_1^{\wsetf}(I_0) = w_1 \cdot \frac{p_1}{w_1/\sum_{j=1}^n w_j} = p_1 \cdot \sum_{j=1}^n w_j = w_1 \cdot C_1^{\wspt}(I_0) + \sum_{k=2}^n w_k \cdot p_1,
 \] so the statement holds for the first job. In order to show the statement for all other jobs, we consider the problem instance~$I_0'$ with job set~$J' \coloneqq \{2,\dotsc,n\}$. For every $j \in J'$ it holds that $C_j^{\wspt}(I_0) = p_1 + C_j^{\wspt}(I'_0)$.
 In the \wsetf{} schedule for $I_0$ every $j \in J'$ is processed at a rate of $w_j/\sum_{k=1}^n w_k$
 until time $C_1^{\wsetf}(I_0)$, while in the \wsetf{} schedule for $I_0'$ it is first processed at a rate of $w_j/\sum_{k=2}^n w_k$. Since
 \begin{multline*}
  Y_j^{\wsetf} (I_0, C_1^{\wsetf}(I_0)) = 
  \frac{w_j}{\sum_{k=1}^n w_k} \cdot C_1^{\wsetf}(I_0) = \frac{w_j}{\sum_{k=1}^n w_k} \cdot \frac{p_1}{w_1} \cdot \sum_{k=1}^n w_k = w_j \cdot \frac{p_1}{w_1} \\ = \frac{w_j}{\sum_{k=2}^n w_k} \cdot \frac{p_1}{w_1} \cdot \sum_{k=2}^n  w_k = \frac{w_j}{\sum_{k=2}^n w_k} \bigl(C_1^{\wsetf}(I_0) - p_1\bigr)
  =   Y_j^{\wsetf} (I'_0, C_1^{\wsetf}(I_0) - p_1) , 
 \end{multline*}
 every job~$j$ has received the same amount of processing in the \wsetf{} schedule for $I_0$ at time $C_1^{\wsetf}(I_0)$ as in the \wsetf{} schedule for $I_0'$ at time~$C_1^{\wsetf}(I_0) - p_1$. Thus, for any time~$t > C_1^{\wsetf}(I_0)$ the \wsetf{} schedule for $I_0$ at time~$t$ coincides with the \wsetf{} schedule for $I_0'$ at time $t-p_1$. Therefore, $C_j^{\wsetf}(I_0) = p_1 + C_j^{\wsetf}(I_0')$ for all jobs~$j \in J'$. Putting things together, we obtain for all $j\in J'$
 \begin{align*}
 	w_j \cdot C_j^{\wsetf}(I_0) &= w_j \cdot p_1 + w_j\cdot C_j^{\wsetf}(I_0') \\
   &= w_j \cdot p_1 + w_j\cdot C_j^{\wspt}(I_0')+ \sum_{k=j+1}^n w_k \cdot p_j
   = w_j\cdot C_j^{\wspt}(I_0) + \sum_{k=j+1}^n w_k \cdot p_j,
 \end{align*}
 where we applied the induction hypothesis to the instance $I_0'$ with $n-1$ jobs.
\end{proof}

To bound the optimum objective value from below, we consider the \emph{mean busy times}
\[M_j^{\mathrm S} \coloneqq \int_0^\infty t \cdot y_j^{\mathrm S}(t) \,\mathrm d t\] of the jobs~$j$ in an arbitrary schedule~$\mathrm S$. Since the mean busy time of each job is smaller than its completion time, the sum of weighted mean busy times is a lower bound on the sum of weighted completion times. It is well known~\cite{Goe96,Goe97} that the former is minimized by the Preemptive WSPT (\pwspt) rule, which always processes an available job with smallest index (i.e.\ with smallest Smith ratio~$p_j/w_j$). Thus, the sum of weighted mean busy times in the \pwspt{} schedule is a lower bound on $\opt(I)$, and it suffices to show that \[\sum_{j=1}^n w_j \cdot C_j^{\wsetf}(I) \le 2 \cdot \sum_{j=1}^n w_j \cdot M_j^{\pwspt}(I).\] The \pwspt{} rule is illustrated in \cref{fig:sec5_example_release_dates}. 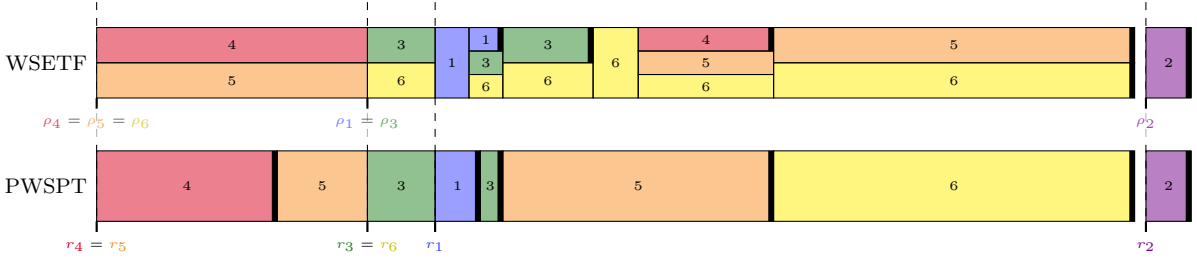
\begin{figure}
\centering

\begin{tikzpicture}[xscale=0.15, yscale=0.9375]

 \begin{scope}[yshift=-1.75cm]
  \coordinate (scheduleNW) at (0,0);
  \coordinate (scheduleW) at (0,-.5);
  \coordinate (scheduleSW) at (0,-1);
  \coordinate (scheduleNE) at (93.0,0);
  \coordinate (scheduleE) at (93.0,-.5);
  \coordinate (scheduleSE) at (93.0,-1);


  \fill[color1] (0.0, 0) rectangle ++(16.0, -1);
  \fill[color2] (16.0, 0) rectangle ++(8.0, -1);
  \fill[color4] (24.0, 0) rectangle ++(6.0, -1);
  \fill[color5] (30.0, 0) rectangle ++(4.0, -1);
  \fill[color4] (34.0, 0) rectangle ++(2.0, -1);
  \fill[color2] (36.0, 0) rectangle ++(24.0, -1);
  \fill[color3] (60.0, 0) rectangle ++(32.0, -1);
  \fill[color6] (93.0, 0) rectangle ++(4.0, -1);


  \draw (0.0, 0) rectangle node {\tiny 4} ++(16.0, -1);
  \draw (16.0, 0) rectangle node {\tiny 5} ++(8.0, -1);
  \draw (24.0, 0) rectangle node {\tiny 3} ++(6.0, -1);
  \draw (30.0, 0) rectangle node {\tiny 1} ++(4.0, -1);
  \draw (34.0, 0) rectangle node {\tiny\!3} ++(2.0, -1);
  \draw (36.0, 0) rectangle node {\tiny 5} ++(24.0, -1);
  \draw (60.0, 0) rectangle node {\tiny 6} ++(32.0, -1);
  \draw (93.0, 0) rectangle node {\tiny 2} ++(4.0, -1);


  \draw[thick] (30.0, -1) -- ++(0, -1ex) node[below] (r1) {\tiny $\textcolor{strong5}{r_1}$};
  \draw[thick] (93.0, -1) -- ++(0, -1ex) node[below] (r2) {\tiny $\textcolor{strong6}{r_2}$};
  \draw[black, thick] (24.0, -1) -- ++(0, -1ex) node[below] (r3) {\tiny $\textcolor{strong4}{r_3} = \textcolor{strong3!80!black}{r_6}$};
  \draw[black, thick] (0.0, -1) -- ++(0, -1ex) node[below] (r4) {\tiny $\textcolor{strong1}{r_4} = \textcolor{strong2}{r_5}$};


  \draw[line width=0.75mm] (16.0, 0) ++(-2.25mm, 0) -- ++(0, -1);
  \draw[line width=0.75mm] (34.0, 0) ++(-2.25mm, 0) -- ++(0, -1);
  \draw[line width=0.75mm] (36.0, 0) ++(-2.25mm, 0) -- ++(0, -1);
  \draw[line width=0.75mm] (60.0, 0) ++(-2.25mm, 0) -- ++(0, -1);
  \draw[line width=0.75mm] (92.0, 0) ++(-2.25mm, 0) -- ++(0, -1);
  \draw[line width=0.75mm] (97.0, 0) ++(-2.25mm, 0) -- ++(0, -1);
 \end{scope}

 \foreach \j in {1,...,4}
  \draw[dashed] (r\j) -- +(0, 3.5);
 \node[anchor=east] at (scheduleW) {\footnotesize $\pwspt$};

 \begin{scope}
  \coordinate (scheduleNW) at (0,0);
  \coordinate (scheduleW) at (0,-.5);
  \coordinate (scheduleSW) at (0,-1);
  \coordinate (scheduleNE) at (93.0,0);
  \coordinate (scheduleE) at (93.0,-.5);
  \coordinate (scheduleSE) at (93.0,-1);


  \fill[color1] (0, 0) rectangle ++(24.0, -0.5);
  \fill[color2] (0, -0.5) rectangle ++(24.0, -0.5);
  \fill[color4] (24.0, 0) rectangle ++(6.0, -0.5);
  \fill[color3] (24.0, -0.5) rectangle ++(6.0, -0.5);
  \fill[color5] (30.0, 0) rectangle ++(3.0, -1.0);
  \fill[color5] (33.0, 0) rectangle ++(3.0, -0.3333333333333333);
  \fill[color4] (33.0, -0.3333333333333333) rectangle ++(3.0, -0.3333333333333333);
  \fill[color3] (33.0, -0.6666666666666666) rectangle ++(3.0, -0.3333333333333333);
  \fill[color4] (36.0, 0) rectangle ++(8.0, -0.5);
  \fill[color3] (36.0, -0.5) rectangle ++(8.0, -0.5);
  \fill[color3] (44.0, 0) rectangle ++(4.0, -1.0);
  \fill[color1] (48.0, 0) rectangle ++(12.0, -0.3333333333333333);
  \fill[color2] (48.0, -0.3333333333333333) rectangle ++(12.0, -0.3333333333333333);
  \fill[color3] (48.0, -0.6666666666666666) rectangle ++(12.0, -0.3333333333333333);
  \fill[color2] (60.0, 0) rectangle ++(32.0, -0.5);
  \fill[color3] (60.0, -0.5) rectangle ++(32.0, -0.5);
  \fill[color6] (93.0, 0) rectangle ++(4.0, -1.0);


  \draw (0, 0) rectangle node {\tiny 4} ++(24.0, -0.5);
  \draw (0, -0.5) rectangle node {\tiny 5} ++(24.0, -0.5);
  \draw (24.0, 0) rectangle node {\tiny 3} ++(6.0, -0.5);
  \draw (24.0, -0.5) rectangle node {\tiny 6} ++(6.0, -0.5);
  \draw (30.0, 0) rectangle node {\tiny 1} ++(3.0, -1.0);
  \draw (33.0, 0) rectangle node {\tiny 1} ++(3.0, -0.3333333333333333);
  \draw (33.0, -0.3333333333333333) rectangle node {\tiny 3} ++(3.0, -0.3333333333333333);
  \draw (33.0, -0.6666666666666666) rectangle node {\tiny 6} ++(3.0, -0.3333333333333333);
  \draw (36.0, 0) rectangle node {\tiny 3} ++(8.0, -0.5);
  \draw (36.0, -0.5) rectangle node {\tiny 6} ++(8.0, -0.5);
  \draw (44.0, 0) rectangle node {\tiny 6} ++(4.0, -1.0);
  \draw (48.0, 0) rectangle node {\tiny 4} ++(12.0, -0.3333333333333333);
  \draw (48.0, -0.3333333333333333) rectangle node {\tiny 5} ++(12.0, -0.3333333333333333);
  \draw (48.0, -0.6666666666666666) rectangle node {\tiny 6} ++(12.0, -0.3333333333333333);
  \draw (60.0, 0) rectangle node {\tiny 5} ++(32.0, -0.5);
  \draw (60.0, -0.5) rectangle node {\tiny 6} ++(32.0, -0.5);
  \draw (93.0, 0) rectangle node {\tiny 2} ++(4.0, -1.0);


  \draw[thick] (93.0, -1) -- ++(0, -1ex) node[fill=white, fill opacity=.7, text opacity=1,below] {\tiny $\textcolor{strong6}{\rho_2}$};
  \draw[black, thick] (24.0, -1) -- ++(0, -1ex) node[fill=white, fill opacity=.7, text opacity=1,below] {\tiny $\textcolor{strong5}{\rho_1} = \textcolor{strong4}{\rho_3}$};
  \draw[black, thick] (0.0, -1) -- ++(0, -1ex) node[fill=white, fill opacity=.7, text opacity=1,below] {\tiny $\textcolor{strong1}{\rho_4} = \textcolor{strong2}{\rho_5} = \textcolor{strong3!80!black}{\rho_6}$};


  \draw[line width=0.75mm] (36.0, 0) ++(-2.25mm, 0) -- ++(0, -0.3333333333333333);
  \draw[line width=0.75mm] (44.0, 0) ++(-2.25mm, 0) -- ++(0, -0.5);
  \draw[line width=0.75mm] (60.0, 0) ++(-2.25mm, 0) -- ++(0, -0.3333333333333333);
  \draw[line width=0.75mm] (92.0, 0) ++(-2.25mm, 0) -- ++(0, -0.5);
  \draw[line width=0.75mm] (92.0, -0.5) ++(-2.25mm, 0) -- ++(0, -0.5);
  \draw[line width=0.75mm] (97.0, 0) ++(-2.25mm, 0) -- ++(0, -1.0);
 \end{scope}

 \node[anchor=east] at (scheduleW) {\footnotesize $\wsetf$};
\end{tikzpicture}

\caption{An example for the \wsetf{} schedule and the \pwspt{} schedule for the instance~$I$ with $\vec p = (4,4,8,16,32,32)^\top$\llap{,} $\vec r = (30,93,24,0,0,24)^\top$\llap{,} and unit weights $\vec w = 1$. Thick lines indicate completions of jobs.} \label{fig:sec5_example_release_dates}
\end{figure}%
Note that the inequality holds with equality for instances~$I_0$ with trivial release dates because for such instances we have $C_j^{\wspt} = M_j^{\pwspt} + p_j/2$, so that, by \cref{lemma:WSETF_no_release},
\[w_j \cdot C_j^{\wsetf}(I_0) = w_j \cdot M_j^{\pwspt}(I_0) + \frac{w_j p_j}{2} + \sum_{k=j+1}^n w_k p_j\]
for every $j \in [n]$. By summing over all jobs~$j$ we deduce
\begin{equation}\begin{split}
 &\sum_{j=1}^n w_j C_j^{\wsetf}(I_0) = \sum_{j=1}^n w_j \cdot M_j^{\pwspt}(I_0) + \sum_{j=1}^n \frac{w_jp_j}{2} + \sum_{j=1}^n \sum_{k=j+1}^n w_k p_j \\
 ={} &\sum_{j=1}^n w_j \cdot M_j^{\pwspt}(I_0) + \sum_{k=1}^n \frac{w_kp_k}{2} + \sum_{k=1}^n w_k \sum_{j=1}^{k-1} p_j = \sum_{j=1}^n w_j \cdot M_j^{\pwspt}(I_0) + \sum_{k=1}^n w_k \cdot M_k^{\pwspt}(I_0).
\end{split}\label{sum_weighted_completion_times_WRR}\end{equation}
This implies the $2$-competitiveness of the Weighted Round-Robin algorithm, proved by \citeauthor{KC03}. In the remainder of this \lcnamecref{sec:WSETF} this argument will be generalized to jobs released over time. We start by reducing the instance to a simpler case without changing the values to be compared.

\begin{lemma}
 For every instance~$I = (\vec p, \vec w, \vec r)$ there is an instance $I' = (\vec p', \vec w', \vec r')$ consisting of $n'$ jobs such that no job is preempted in the $\pwspt$ schedule for $I'$, $\sum_{j=1}^{n'} w_j' C_j^{\wsetf}(I') = \sum_{j=1}^n w_j C_j^{\wsetf}(I)$, and $\sum_{j=1}^{n'} w_j' M_j^{\pwspt}(I') = \sum_{j=1}^n w_j M_j^{\pwspt}(I)$.
\end{lemma}
\begin{proof}
 We split every job into subjobs corresponding to the parts processed without interruption in the $\pwspt$ schedule, i.e., we replace each job $j$ by jobs $(j,1),\dotsc,(j,\ell_j)$ such that $\sum_{i=1}^{\ell_j} p_{(j,i)}' = p_j$. Moreover, we set the weights to $w_{(j,i)} \coloneqq \frac{p_{(j,i)}}{p_j} \cdot w_j$, so that all parts have the same Smith ratio as the original job. Finally, the release dates are set to $r_{(j,i)} \coloneqq r_j$. This operation does not change the sum of weighted mean busy times in the $\pwspt$ schedule. Moreover, in the $\wsetf$ schedule, since all these jobs are released simultaneously and have the same Smith ratio, they will always be processed in a way so that their weighted elapsed times increase equally, so that they are all completed at the same time. Moreover, the total rate assigned to the jobs $(j,i)$ for a fixed $j$ equals the rate assigned to $j$ in the original schedule. Therefore, all these jobs finish exactly at the time when job~$j$ is completed in the original schedule.
\end{proof}

From now on we always consider an instance $I = (\vec p, \vec w, \vec r)$ so that no job is preempted in the $\pwspt$ schedule and $p_1/w_1 \le \dotsc \le p_n/w_n$. We omit the instance in the notation for completion and elapsed times. For every fixed job~$j$ let $\rho_j$ be the first point in time such that during $(\rho_j,C_j^{\wsetf}]$ the machine continuously processes jobs~$k$ with $Y_k^{\wsetf}(C_j^{\wsetf})/w_k \le p_j/w_j$, and let $R(j)$ be the set of jobs processed in this interval. The times~$\rho_j$ are shown in the example in \cref{fig:sec5_example_release_dates} and have the property that $r_k \ge \rho_j$ for all $k \in R(j)$ because otherwise the machine would only process jobs~$l$ with $Y_l^{\wsetf}(\rho_j)/w_l \le Y_k^{\wsetf}(\rho_j)/w_k \le p_j/w_j$ between $r_k$ and $\rho_j$, in contradiction to the minimality of $\rho_j$.
Let $I(j)$ be the instance with job set $R(j)$ where all jobs are released at time~$0$. In the next \lcnamecref{WSETF WRR I(j)} we compare the completion time of $j$ in the \wsetf{} schedule for the original instance~$I$ to the corresponding completion time in the instance~$I(j)$.

\begin{lemma} \label{WSETF WRR I(j)}
 $C_j^{\wsetf} = \rho_j + C_j^{\wsetf}(I(j))$ for every job~$j$.
\end{lemma}
\begin{proof}
 In the \wsetf{} schedule for $I$ every job~$k \in R(j)$ that is completed between the times~$\rho_j$ and $C_j^{\wsetf}(I)$ has $p_k/w_k = Y_k^{\wsetf}(C_j^{\wsetf})/w_k \le p_j/w_j$. Hence, by \cref{Order WSETF}, it also finishes before $j$ in the \wsetf{} schedule for $I(j)$. All other jobs~$k \in R(j)$ have received processing time $\frac{w_k}{w_j} \cdot p_j$ before the completion of $j$ in both schedules. Therefore, the total processing that jobs from $R(j)$ receive in the \wsetf{} schedule for $I$ between the times $\rho_j$ and $C_j^{\wsetf}$ equals their total processing before the completion of $j$ in the \wsetf{} schedule for $I(j)$. Since in both schedules the machine is continuously processing jobs from $R(j)$ in the considered time, this implies the \lcnamecref{WSETF WRR I(j)}.
\end{proof}

Analogously to \cref{WSETF WRR I(j)}, we compare in the next \lcnamecref{Sp WSPT I(j)} the preemptive WSPT schedules for $I$ and $I(j)$.

\begin{lemma} \label{Sp WSPT I(j)}
 For every job~$j$
 \[C_j^{\pwspt} = \rho_j + C_j^{\wspt}(I(j)) - \sum_{\substack{k \in R(j):k<j,\\C_k^{\pwspt} > C_j^{\pwspt}}} p_k + \sum_{\substack{k \in R(j):k>j,\\C_k^{\pwspt} < C_j^{\pwspt}}} p_k.\]
\end{lemma}
\begin{proof}
  By \cref{Order WSETF}, every job~$k \le j$ released before time~$C_j^{\wsetf}$ has $C_k^{\wsetf} \le C_j^{\wsetf}$. \pwspt{} always schedules some job $k \le j$ at a rate of $1$ if one is available. Thus, until every point in time, it cannot have spent less time on these jobs than \wsetf{}. Hence, at time~$C_j^{\wsetf}$, \pwspt{} must have finished all these jobs as well. In particular, it has completed $j$.
  Therefore, we have $C_j^{\pwspt} \le C_j^{\wsetf}$ for every job~$j$.

	Since the \pwspt{} and the \wsetf{} strategies both fully utilize the machine whenever some jobs are available, the resulting schedules have the same idle intervals. By definition of $\rho_j$, there is no idle time in $(\rho_j,C_j^{\wsetf}] \supseteq (\rho_j, C_j^{\pwspt}]$.
  Hence, we can partition $(\rho_j, C_j^{\pwspt}]$ according to the job being processed in the \pwspt{} schedule:
 \begin{align*}
  C_j^{\pwspt} = \rho_j + \sum_{\substack{k \in R(j):\\C_k^{\pwspt} \le C_j^{\pwspt}}} p_k &= \rho_j + \sum_{k \in R(j): k \le j} p_k - \sum_{\substack{k \in R(j): k \le j\\C_k^{\pwspt} > C_j^{\pwspt}}} p_k + \sum_{\substack{k \in R(j):k > j\\C_k^{\pwspt} \le C_j^{\pwspt}}} p_k\\
  &= \rho_j + C_j^{\wspt}(I(j)) - \sum_{\substack{k \in R(j): k < j\\C_k^{\pwspt} > C_j^{\pwspt}}} p_k + \sum_{\substack{k \in R(j):k > j\\C_k^{\pwspt} < C_j^{\pwspt}}} p_k.
  \qedhere
 \end{align*}
\end{proof}

Now we are ready to prove the \lcnamecref{thm:WSETF}.

\begin{proof}[Proof of \cref{thm:WSETF}]
 Using the fact that $\pwspt$ is non-preemptive, we have $C_j^{\pwspt} = M_j^{\pwspt} + \frac{p_j}{2}$ for all $j$, and thus
 \begin{equation}
  \sum_{j=1}^n w_j \cdot C_j^{\wsetf} = \sum_{j=1}^n w_j \cdot M_j^{\pwspt} + \sum_{j=1}^n w_j \cdot \frac{p_j}{2} + \sum_{j=1}^n w_j \cdot (C_j^{\wsetf} - C_j^{\pwspt}).\label{two summands}
 \end{equation}
 We will now bound the summands from the last sum by an expression generalizing (\hyperlink{weightedDelay}{\textasteriskcentered}) from \cref{lemma:WSETF_no_release}. For a fixed~$j \in [n]$ \cref{WSETF WRR I(j),Sp WSPT I(j)} yield
 \[
  C_j^{\wsetf} - C_j^{\pwspt} = C_j^{\wsetf}(I(j)) - C_j^{\wspt}(I(j)) + \sum_{\substack{k \in R(j): k < j \\ C_k^{\pwspt} > C_j^{\pwspt}}} p_k - \sum_{\substack{k \in R(j): k > j \\ C_k^{\pwspt} < C_j^{\pwspt}}} p_k.
 \]
 Since $I(j)$ is an instance with trivial release dates, by \cref{lemma:WSETF_no_release},
 \begin{align*}
  w_j \cdot (C_j^{\wsetf} - C_j^{\pwspt}) &= \sum_{k \in R(j): k > j} w_k p_j + \sum_{\substack{k \in R(j): k < j \\ C_k^{\pwspt} > C_j^{\pwspt}}} w_j p_k - \sum_{\substack{k \in R(j): k > j \\ C_k^{\pwspt} < C_j^{\pwspt}}} w_j p_k \\
  &\le \sum_{k \in R(j): k > j} w_k p_j + \sum_{\substack{k \in R(j): k < j \\ C_k^{\pwspt} > C_j^{\pwspt}}} w_k p_j - \sum_{\substack{k \in R(j): k > j \\ C_k^{\pwspt} < C_j^{\pwspt}}} w_k p_j \\
  &= \sum_{\substack{k \in R(j): k \neq j \\ C_k^{\pwspt} > C_j^{\pwspt}}} w_k p_j \le \sum_{\substack{k \neq j: \\ C_k^{\pwspt} > C_j^{\pwspt}}} w_k p_j,
 \end{align*}
 where the first inequality holds because $w_j p_k \le w_k p_j$ for $k < j$ and $w_j p_k \ge w_k p_j$ for $k > j$. Now we sum over all jobs~$j$, obtaining
 \[\sum_{j=1}^n w_j \cdot (C_j^{\wsetf} - C_j^{\pwspt}) \le \sum_{j=1}^n \sum_{\substack{k \neq j:\\ C_k^{\pwspt} > C_j^{\pwspt}}} w_k p_j = \sum_{k=1}^n w_k \sum_{\substack{j \neq k:\\ C_j^{\pwspt} < C_k^{\pwspt}}} p_j.\]
 Substituting this inequality into \eqref{two summands}, we can generalize the computation in \cref{sum_weighted_completion_times_WRR}:
 \begin{align*}
  \sum_{j=1}^n w_j \cdot C_j^{\wsetf}
  &\le \sum_{j=1}^n w_j \cdot M_j^{\pwspt} + \sum_{j=1}^n \frac{w_j p_j}{2} + \sum_{j=1}^n w_j \sum_{\substack{k \neq j:\\ C_k^{\pwspt} < C_j^{\pwspt}}} p_k \\
  &=\sum_{j=1}^n w_j \cdot M_j^{\pwspt} + \sum_{j=1}^n w_j \cdot \biggl(\frac{p_j}{2} + \sum_{\substack{k \neq j:\\ C_k^{\pwspt} < C_j^{\pwspt}}} p_k\biggr) \\
  &= 2 \cdot \sum_{j=1}^n w_j \cdot M_j^{\pwspt}. \qedhere
 \end{align*}
\end{proof}

\section{\texorpdfstring{Extensions of the $b$-Scaling Strategy for Release Dates and Parallel Machines}{Extensions of the b-Scaling Strategy for Release Dates and Parallel Machines}}

In this section, we present extensions of the $b$-scaling strategies for the settings \threefield{1}{r_j}{w_jC_j} and \threefield{\mathrm P}{}{C_j}. We denote by $(\vec{p},\vec{w},\vec{r},m)$ an instance on $m$ identical parallel machines in which each job~$j$ has processing time $p_j$, weight $w_j$, and release date $r_j$. In order to obtain bounds on the competitive ratios of these extensions, we compare the schedules of $\alg$ to schedules of algorithms that are constant-competitive. In particular, we use the $2$-competitiveness of \wsetf{} from the previous section for
\threefield{1}{r_j,\,\mathrm{pmtn}}{w_jC_j} and the $2$-competitiveness of the round-robin~(\rr) schedule for \threefield{\mathrm P}{\mathrm{pmtn}}{C_j}.
We begin with the relation of the optimal costs of two instances whose processing times and release dates differ only by a multiplicative factor.

\begin{lemma} \label{lem:b_factor_opt}
	Consider an instance $I=(\vec{p},\vec{w},\vec{r},m)$, and let $I'=(\vec{p}',\vec{w},\vec{r}',m)$, where $\vec{p}'\leq \alpha \vec{p}$ and $\vec{r}'\leq \alpha \vec{r}$. Then, we have
	\[
		\opt(I')\leq \alpha \cdot \opt(I).
	\]
\end{lemma}
\begin{proof}
	Let $S_j^\Pi(I)$ denote the starting time of job~$j$ in schedule $\Pi$ for instance $I$.
	We define a schedule $\Pi$ with $S_j^\Pi(I' )=\alpha\cdot C_j^{\opt}(I)- p_j'$. By definition $C_j^\Pi(I')=\alpha\cdot C_j^{\opt}(I)$. Clearly, $\Pi(I')=\alpha\cdot \opt(I)$.
	We claim that $\Pi$ is feasible for $I'$.
	Indeed, we have
	\[
		S_j^\Pi(I')= \alpha\cdot (S_j^{\opt}(I)+p_j)- p_j'\geq \alpha\cdot(r_j +p_j)- \alpha p_j\geq r_j'.
	\]
	Suppose now that two jobs~$j, k$ overlap in $\Pi$, i.e., $S_j^\Pi(I') < C_k^\Pi(I') \le C_j^\Pi(I')$. By definition of $\Pi$, this means that $\alpha \cdot C_j^{\opt}(I) - p_j' < \alpha \cdot C_k^{\opt}(I) \le \alpha \cdot C_j^{\opt}(I)$. Therefore, $S_j^{\opt}(I) \le S_j^{\opt}(I) + p_j - \frac{p_j'}{\alpha} < C_k^{\opt}(I) \le C_j^{\opt}(I)$, meaning that $j$ and $k$ also overlap in the optimal schedule for $I$.
	Since the optimal schedule for $I$ always schedules at most $m$ jobs in parallel, this is also the case for $\Pi$.
	Hence, $\Pi$ is a feasible schedule for $I'$ and $\opt(I')\leq\Pi(I')=\alpha\cdot \opt(I)$.
\end{proof}

This lemma will turn out to be useful for both investigated settings.

\subsection{Release Dates}

Let us first consider the single machine case where jobs arrive online. For instances of \threefield{1}{r_j}{w_jC_j}, we extend $\alg$ in the following way: The strategy keeps track of the rank $\jobrank_j(\theta)$ of each job $j$ at time $\theta$, where $\jobrank_j(\theta)$ is the largest integer~$q$ such that job $j$ has already been probed for $w_j b^{q-1}$ before $\theta$. At any end of a probing occurring at time $\theta$ it selects the job $j$ with minimum rank and index among all released and not completed jobs and probes it for $w_j b^{\jobrank_j(\theta)}$.
We only consider the limit strategy obtained when the rank of a job is set to $\jobrank_0\to-\infty$ at its release, so
that phases of infinitesimal probing occur after each release date.
Note that the strategy never interrupts a probing, i.e., each probing $(t_i, j_i, \tau_i)$ is executed until time $t_i + \min \{p_j, \tau_i\}$. In order to show an upper bound on the competitive ratio of $\alg$, we actually compare \alg{} against an optimal \emph{preemptive} offline algorithm. This means that we obtain an upper bound on the ``power of preemption'' in the online setting with restarts, complementing the unbounded ratio in the model without restarts.

\TrivialBoundReleaseDates

To prove this result, we need a bound on the end time of a probing with respect to the point in time at which the probing began.

\begin{lemma} \label{lem:probing_end_bound}
	For $I = (\vec{p}, \vec{w}, \vec{r}, m)$ denote by $S$ the schedule produced by $\alg$. If some job~$j$ is probed at time~$t$ for~$\tau = w_j b^{\jobrank}$ for some $\jobrank \in \mathbb{Z}$, then $t + \tau \leq b t$.
\end{lemma}
\begin{proof}
	At the start of the probing, $\alg$ has spent $\sum_{i=-\infty}^{\jobrank-1} w_j b^{i} = w_j \frac{b^{\jobrank}}{b-1}$ time probing the job~$j$. Since probings of the same job cannot run in parallel on multiple machines, we have that $t \geq w_j \frac{b^{\jobrank}}{b-1}$. Thus,
	\[
		t + \tau = t + w_j b^{\jobrank} \leq t + (b-1) t = b t
		.
		\qedhere
	\]
\end{proof}

The proof of \cref{thm:trivial_bound_release_dates} consists of two main steps that are carried out in the subsequent two lemmas. In the first step, we compare an arbitrary instance~$I$ to an instance~$I'$ with processing times such that the Smith ratios are rounded to integer powers of $b$ and release dates that are shifted to end times of probings. In the second step, we compare the performance of $\alg$ on the instance~$I'$ to the performance of $\wsetf$ on another instance~$I''$. \Cref{fig:alg-vs-wsetf} illustrates these two auxiliary instances for an example.

\begin{figure}

\centering
\begin{tikzpicture}[scale=1.5]
\begin{scope}[xscale=0.1, yscale=0.75,]
\coordinate (scheduleNW) at (0,0);
\coordinate (scheduleW) at (0,-.5);
\coordinate (scheduleSW) at (0,-1);
\coordinate (scheduleNE) at (95.0,0);
\coordinate (scheduleE) at (95.0,-.5);
\coordinate (scheduleSE) at (95.0,-1);


 \fill[gray] (0, 0) rectangle ++(2.0, -1);
\fill[color1] (2.0, 0) rectangle ++(1.0, -1);
\fill[color2] (3.0, 0) rectangle ++(1.0, -1);
\fill[color1] (4.0, 0) rectangle ++(2.0, -1);
\fill[color2] (6.0, 0) rectangle ++(2.0, -1);
\fill[color1] (8.0, 0) rectangle ++(4.0, -1);
\fill[color2] (12.0, 0) rectangle ++(4.0, -1);
\fill[color1] (16.0, 0) rectangle ++(8.0, -1);
\fill[gray] (24.0, 0) rectangle ++(2.0, -1);
\fill[color3] (26.0, 0) rectangle ++(1.0, -1);
\fill[color4] (27.0, 0) rectangle ++(1.0, -1);
\fill[color3] (28.0, 0) rectangle ++(2.0, -1);
\fill[gray] (30.0, 0) rectangle ++(1.0, -1);
\fill[color5] (31.0, 0) rectangle ++(1.0, -1);
\fill[color4] (32.0, 0) rectangle ++(2.0, -1);
\fill[color5] (34.0, 0) rectangle ++(2.0, -1);
\fill[color3] (36.0, 0) rectangle ++(4.0, -1);
\fill[color4] (40.0, 0) rectangle ++(3, -1);
\fill[color2] (43.0, 0) rectangle ++(8.0, -1);
\fill[color3] (51.0, 0) rectangle ++(8.0, -1);
\fill[color2] (59.0, 0) rectangle ++(9, -1);
\fill[color3] (68.0, 0) rectangle ++(9, -1);
\fill[gray] (93, 0) rectangle ++(1.0, -1);
\fill[color6] (94.0, 0) rectangle ++(1.0, -1);
\fill[color6] (95.0, 0) rectangle ++(2.0, -1);


 \draw (0, 0) rectangle ++(2.0, -1) ;
\draw (2.0, 0) rectangle ++(1.0, -1) ;
\draw (3.0, 0) rectangle ++(1.0, -1) ;
\draw (4.0, 0) rectangle ++(2.0, -1) ;
\draw (6.0, 0) rectangle ++(2.0, -1) ;
\draw (8.0, 0) rectangle ++(4.0, -1)  node[midway] {\tiny 1};
\draw (12.0, 0) rectangle ++(4.0, -1)  node[midway] {\tiny 2};
\draw (16.0, 0) rectangle ++(8.0, -1)  node[midway] {\tiny 1};
\draw (24.0, 0) rectangle ++(2.0, -1) ;
\draw (26.0, 0) rectangle ++(1.0, -1) ;
\draw (27.0, 0) rectangle ++(1.0, -1) ;
\draw (28.0, 0) rectangle ++(2.0, -1) ;
\draw (30.0, 0) rectangle ++(1.0, -1) ;
\draw (31.0, 0) rectangle ++(1.0, -1) ;
\draw (32.0, 0) rectangle ++(2.0, -1) ;
\draw (34.0, 0) rectangle ++(2.0, -1) ;
\draw (36.0, 0) rectangle ++(4.0, -1)  node[midway] {\tiny 3};
\draw (40.0, 0) rectangle ++(3, -1)  node[midway] {\tiny 4};
\draw (43.0, 0) rectangle ++(8.0, -1)  node[midway] {\tiny 2};
\draw (51.0, 0) rectangle ++(8.0, -1)  node[midway] {\tiny 3};
\draw (59.0, 0) rectangle ++(9, -1)  node[midway] {\tiny 2};
\draw (68.0, 0) rectangle ++(9, -1)  node[midway] {\tiny 3};
\draw (93, 0) rectangle ++(1.0, -1) ;
\draw (94.0, 0) rectangle ++(1.0, -1) ;
\draw (95.0, 0) rectangle ++(2.0, -1) ;


 \draw[thick] (0, -1) -- ++(0, -1ex) node[below] {\tiny $\textcolor{strong1}{r_1} = \textcolor{strong2}{r_2}$};
 \draw[thick] (18, -1) -- ++(0, -1ex) node[below] {\tiny $\textcolor{strong3!80!black}{r_3} = \textcolor{strong4}{r_4}$};
 \draw[thick] (30, -1) -- ++(0, -1ex) node[below] {\tiny \textcolor{strong5}{$r_5$}};
 \draw[thick] (93, -1) -- ++(0, -1ex) node[below] {\tiny \textcolor{strong6}{$r_6$}};


 \draw[line width=0.75mm] (24.0, 0) ++(-2.25mm, 0) -- ++(0, -1);
\draw[line width=0.75mm] (36.0, 0) ++(-2.25mm, 0) -- ++(0, -1);
\draw[line width=0.75mm] (43.0, 0) ++(-2.25mm, 0) -- ++(0, -1);
\draw[line width=0.75mm] (68.0, 0) ++(-2.25mm, 0) -- ++(0, -1);
\draw[line width=0.75mm] (77.0, 0) ++(-2.25mm, 0) -- ++(0, -1);
\draw[line width=0.75mm] (97.0, 0) ++(-2.25mm, 0) -- ++(0, -1);
\end{scope}

\node[anchor=east] at (scheduleW) {\footnotesize $\alg (I)$};

\begin{scope}[xscale=0.1, yscale=0.75,yshift=-1.75cm]
\coordinate (scheduleNW) at (0,0);
\coordinate (scheduleW) at (0,-.5);
\coordinate (scheduleSW) at (0,-1);
\coordinate (scheduleNE) at (95.0,0);
\coordinate (scheduleE) at (95.0,-.5);
\coordinate (scheduleSE) at (95.0,-1);


\fill[gray] (0, 0) rectangle ++(2.0, -1);
\fill[color1] (2.0, 0) rectangle ++(1.0, -1);
\fill[color2] (3.0, 0) rectangle ++(1.0, -1);
\fill[color1] (4.0, 0) rectangle ++(2.0, -1);
\fill[color2] (6.0, 0) rectangle ++(2.0, -1);
\fill[color1] (8.0, 0) rectangle ++(4.0, -1);
\fill[color2] (12.0, 0) rectangle ++(4.0, -1);
\fill[color1] (16.0, 0) rectangle ++(8.0, -1);
\fill[gray] (24.0, 0) rectangle ++(2.0, -1);
\fill[color3] (26.0, 0) rectangle ++(1.0, -1);
\fill[color4] (27.0, 0) rectangle ++(1.0, -1);
\fill[color3] (28.0, 0) rectangle ++(2.0, -1);
\fill[gray] (30.0, 0) rectangle ++(1.0, -1);
\fill[color5] (31.0, 0) rectangle ++(1.0, -1);
\fill[color4] (32.0, 0) rectangle ++(2.0, -1);
\fill[color5] (34.0, 0) rectangle ++(2.0, -1);
\fill[color3] (36.0, 0) rectangle ++(4.0, -1);
\fill[color4] (40.0, 0) rectangle ++(4.0, -1);
\fill[color2] (44.0, 0) rectangle ++(8.0, -1);
\fill[color3] (52.0, 0) rectangle ++(8.0, -1);
\fill[color2] (60.0, 0) rectangle ++(16.0, -1);
\fill[color3] (76.0, 0) rectangle ++(16.0, -1);
\fill[gray] (93.0, 0) rectangle ++(1.0, -1);
\fill[color6] (94.0, 0) rectangle ++(1.0, -1);
\fill[color6] (95.0, 0) rectangle ++(2.0, -1);


 \draw (0, 0) rectangle ++(2.0, -1) ;
\draw (2.0, 0) rectangle ++(1.0, -1) ;
\draw (3.0, 0) rectangle ++(1.0, -1) ;
\draw (4.0, 0) rectangle ++(2.0, -1) ;
\draw (6.0, 0) rectangle ++(2.0, -1) ;
\draw (8.0, 0) rectangle ++(4.0, -1)  node[midway] {\tiny 1};
\draw (12.0, 0) rectangle ++(4.0, -1)  node[midway] {\tiny 2};
\draw (16.0, 0) rectangle ++(8.0, -1)  node[midway] {\tiny 1};
\draw (24.0, 0) rectangle ++(2.0, -1) ;
\draw (26.0, 0) rectangle ++(1.0, -1) ;
\draw (27.0, 0) rectangle ++(1.0, -1) ;
\draw (28.0, 0) rectangle ++(2.0, -1) ;
\draw (30.0, 0) rectangle ++(1.0, -1) ;
\draw (31.0, 0) rectangle ++(1.0, -1) ;
\draw (32.0, 0) rectangle ++(2.0, -1) ;
\draw (34.0, 0) rectangle ++(2.0, -1) ;
\draw (36.0, 0) rectangle ++(4.0, -1)  node[midway] {\tiny 3};
\draw (40.0, 0) rectangle ++(4.0, -1)  node[midway] {\tiny 4};
\draw (44.0, 0) rectangle ++(8.0, -1)  node[midway] {\tiny 2};
\draw (52.0, 0) rectangle ++(8.0, -1)  node[midway] {\tiny 3};
\draw (60.0, 0) rectangle ++(16.0, -1)  node[midway] {\tiny 2};
\draw (76.0, 0) rectangle ++(16.0, -1)  node[midway] {\tiny 3};
\draw (93.0, 0) rectangle ++(1.0, -1) ;
\draw (94.0, 0) rectangle ++(1.0, -1) ;
\draw (95.0, 0) rectangle ++(2.0, -1) ;


 \draw[thick] (0.0, -1) -- ++(0, -1ex) node[below] {\tiny $\textcolor{strong1}{r_1'} = \textcolor{strong2}{r_2'}$};
 \draw[thick] (24.0, -1) -- ++(0, -1ex) node[below] {\tiny $\textcolor{strong3!80!black}{r_3'} = \textcolor{strong4}{r_4'}$};
 \draw[thick] (30.0, -1) -- ++(0, -1ex) node[below] {\tiny \textcolor{strong5}{$r_5'$}};
 \draw[thick] (93.0, -1) -- ++(0, -1ex) node[below] {\tiny \textcolor{strong6}{$r_6'$}};


 \draw[line width=0.75mm] (24.0, 0) ++(-2.25mm, 0) -- ++(0, -1);
\draw[line width=0.75mm] (36.0, 0) ++(-2.25mm, 0) -- ++(0, -1);
\draw[line width=0.75mm] (44.0, 0) ++(-2.25mm, 0) -- ++(0, -1);
\draw[line width=0.75mm] (76.0, 0) ++(-2.25mm, 0) -- ++(0, -1);
\draw[line width=0.75mm] (92.0, 0) ++(-2.25mm, 0) -- ++(0, -1);
\draw[line width=0.75mm] (97.0, 0) ++(-2.25mm, 0) -- ++(0, -1);
\end{scope}

\node[anchor=east] at (scheduleW) {\footnotesize $\alg (I')$};

\begin{scope}[xscale=0.1, yscale=0.75,yshift=-3.5cm]
\coordinate (scheduleNW) at (0,0);
\coordinate (scheduleW) at (0,-.5);
\coordinate (scheduleSW) at (0,-1);
\coordinate (scheduleNE) at (93.0,0);
\coordinate (scheduleE) at (93.0,-.5);
\coordinate (scheduleSE) at (93.0,-1);


 \fill[color1] (0, 0) rectangle ++(24.0, -0.5);
\fill[color2] (0, -0.5) rectangle ++(24.0, -0.5);
\fill[color3] (24.0, 0) rectangle ++(6.0, -0.5);
\fill[color4] (24.0, -0.5) rectangle ++(6.0, -0.5);
\fill[color5] (30.0, 0) rectangle ++(3.0, -1.0);
\fill[color3] (33.0, 0) rectangle ++(3.0, -0.3333333333333333);
\fill[color4] (33.0, -0.3333333333333333) rectangle ++(3.0, -0.3333333333333333);
\fill[color5] (33.0, -0.6666666666666666) rectangle ++(3.0, -0.3333333333333333);
\fill[color3] (36.0, 0) rectangle ++(8.0, -0.5);
\fill[color4] (36.0, -0.5) rectangle ++(8.0, -0.5);
\fill[color3] (44.0, 0) rectangle ++(4.0, -1.0);
\fill[color1] (48.0, 0) rectangle ++(12.0, -0.3333333333333333);
\fill[color2] (48.0, -0.3333333333333333) rectangle ++(12.0, -0.3333333333333333);
\fill[color3] (48.0, -0.6666666666666666) rectangle ++(12.0, -0.3333333333333333);
\fill[color2] (60.0, 0) rectangle ++(32.0, -0.5);
\fill[color3] (60.0, -0.5) rectangle ++(32.0, -0.5);
\fill[color6] (93.0, 0) rectangle ++(4.0, -1.0);


 \draw (0, 0) rectangle ++(24.0, -0.5)  node[midway] {\tiny 1};
\draw (0, -0.5) rectangle ++(24.0, -0.5)  node[midway] {\tiny 2};
\draw (24.0, 0) rectangle ++(6.0, -0.5)  node[midway] {\tiny 3};
\draw (24.0, -0.5) rectangle ++(6.0, -0.5)  node[midway] {\tiny 4};
\draw (30.0, 0) rectangle ++(3.0, -1.0)  node[midway] {\tiny 5};
\draw (33.0, 0) rectangle ++(3.0, -0.3333333333333333)  node[midway] {\tiny 3};
\draw (33.0, -0.3333333333333333) rectangle ++(3.0, -0.3333333333333333)  node[midway] {\tiny 4};
\draw (33.0, -0.6666666666666666) rectangle ++(3.0, -0.3333333333333333)  node[midway] {\tiny 5};
\draw (36.0, 0) rectangle ++(8.0, -0.5)  node[midway] {\tiny 3};
\draw (36.0, -0.5) rectangle ++(8.0, -0.5)  node[midway] {\tiny 4};
\draw (44.0, 0) rectangle ++(4.0, -1.0)  node[midway] {\tiny 3};
\draw (48.0, 0) rectangle ++(12.0, -0.3333333333333333)  node[midway] {\tiny 1};
\draw (48.0, -0.3333333333333333) rectangle ++(12.0, -0.3333333333333333)  node[midway] {\tiny 2};
\draw (48.0, -0.6666666666666666) rectangle ++(12.0, -0.3333333333333333)  node[midway] {\tiny 3};
\draw (60.0, 0) rectangle ++(32.0, -0.5)  node[midway] {\tiny 2};
\draw (60.0, -0.5) rectangle ++(32.0, -0.5)  node[midway] {\tiny 3};
\draw (93.0, 0) rectangle ++(4.0, -1.0)  node[midway] {\tiny 6};


 \draw[thick] (0.0, -1) -- ++(0, -1ex) node[below] {\tiny $\textcolor{strong1}{r_1''} = \textcolor{strong2}{r_2''}$};
 \draw[thick] (24.0, -1) -- ++(0, -1ex) node[below] {\tiny $\textcolor{strong3!80!black}{r_3''} = \textcolor{strong4}{r_4''}$};
 \draw[thick] (30.0, -1) -- ++(0, -1ex) node[below] {\tiny $\textcolor{strong5}{r_5''}$};
 \draw[thick] (93.0, -1) -- ++(0, -1ex) node[below] {\tiny $\textcolor{strong6}{r_6''}$};


 \draw[line width=0.75mm] (36.0, -0.6666666666666666) ++(-2.25mm, 0) -- ++(0, -0.3333333333333333);
\draw[line width=0.75mm] (44.0, -0.5) ++(-2.25mm, 0) -- ++(0, -0.5);
\draw[line width=0.75mm] (60.0, 0) ++(-2.25mm, 0) -- ++(0, -0.3333333333333333);
\draw[line width=0.75mm] (92.0, 0) ++(-2.25mm, 0) -- ++(0, -0.5);
\draw[line width=0.75mm] (92.0, -0.5) ++(-2.25mm, 0) -- ++(0, -0.5);
\draw[line width=0.75mm] (97.0, 0) ++(-2.25mm, 0) -- ++(0, -1.0);
\end{scope}

\node[anchor=east] at (scheduleW) {\footnotesize $\wsetf{}(I'')$};

\draw[dashed] (6.0, -1) -- (6.0, -3.75);
\draw[dashed] (9.200000000000001, -1) -- (9.200000000000001, -3.75);
\draw[dashed] (9.200000000000001, -1) -- (9.200000000000001, -3.75);
\draw[dashed] (4.4, -1) -- (4.4, -3.75);
\draw[dashed] (3.6, -1) -- (3.6, -3.75);
\draw[dashed] (9.700000000000001, -1) -- (9.700000000000001, -3.75);
\end{tikzpicture}

\caption{An example for the three schedules considered in the proof of \cref{thm:trivial_bound_release_dates} for the instance $I$ with $\vec{p} = (8,9,9,3,2,2)^{\top}$\llap, $\vec{r} = (0,0,18,18,30,93)^{\top}$\llap, and unit weights $\vec{w} = \vec{1}$, for $b = 2$.
Gray areas indicate infinitesimal probing; thick lines indicate the completion of a job.
\textbf{Top:} The schedule of $\alg*_2 (I)$ for the original instance~$I$.
\textbf{Middle:} The schedule of $\alg*_2 (I')$ for the modified instance~$I'$ with processing times rounded to the next integer power of $b=2$. The release dates $r_3'$ and $r_4'$ are shifted such that they coincide with the end of a probing.
\textbf{Bottom:} The schedule of $\wsetf$ for the instance~$I''$. The processing times (corresponding to the colored areas) correspond to the total probing times in the schedule of~$\alg*_2 (I')$.
The completion times in this schedule are higher or equal to the completion times in the schedule of~$\alg*_2 (I')$, as indicated by the vertical dashed lines.}
\label{fig:alg-vs-wsetf}
\end{figure}

\newcommand{\integerexponent}{\jobrank}

\begin{lemma}\label{lem:release_dates_I_I'}
	For an arbitrary instance $I = (\vec{p}, \vec{w}, \vec{r})$ of \threefield{1}{r_j}{w_jC_j} there exists another instance $I' = (\vec{p}', \vec{w}, \vec{r}')$ with $\vec{p}'\leq \frac{b^3}{2b-1} \vec{p}$ and $\vec{r}'\leq \frac{b^3}{2b-1} \vec{r}$ such that $\alg(I)\leq\alg(I')$. Moreover, the instance $I'$ has the property that for every job~$j$ there exists an integer $\jobrank_j\in\Z$ with $p_j'=w_j b^{\jobrank_j}$ and every release date either coincides with the end of a probing in the schedule $\alg(I')$ or is at some point in time at which $\alg$ idles in $I'$.
\end{lemma}

\begin{proof}
	Let $I$ be an arbitrary instance and consider the schedule $S$ produced by $\alg$ for~$I$. For a fixed job~$j$ with $r_j>0$, we denote by $\pi(j) = (t(j), k(j), \tau(j))$ the last probing that is started by $\alg$ before the release date of~$j$, hence, $\pi(j)$ satisfies
	$
	t (j) < r_j
	$.

	For every job~$j$, let
	\[
	p_j' \coloneqq  w_j b^{\jobrank_j},
	\]
	where $\jobrank_j \coloneqq \big\lceil \log_b \big( \frac{p_j}{w_j} \big)\big\rceil$.
	Observe that $p_j \leq p_j' \leq b  p_j \leq \frac{b^3}{2b-1}  p_j$.

	We define a new schedule $\bar{S}$ as follows.
		For every probing  $(t, k, \tau)$ in the schedule~$S$, there is a corresponding probing $(t', k', \tau')$ with $k'=k$, $\tau' = \tau$ and $t' \geq t$ in the schedule $\bar{S}$, where the times $t'$ are chosen such that no additional idle time exists in $\bar{S}$.
		In particular, in the schedule $\bar{S}$ the same jobs are probed for the same times in the same order as in the schedule~$S$.
		However, the actual duration of any probing $(t', k', \tau')$ in $\bar{S}$ depends on the processing time~$p_j'$ rather than $p_j$, i.e., the duration is $\min \{p_j', \tau\}$. In particular, all probings have the same duration except for those, where jobs complete. These probings last longer and might shift all subsequent probings.
	We define new release dates $\vec{r}'$ as follows.
		For every job~$j$ with $r_j > 0$, consider the probing $\pi'(j) = (t'(j), k'(j), \tau'(j))$ in the schedule~$\bar{S}$ that corresponds to the probing~$\pi(j)$, and set $r_j' \coloneqq \max\{ r_j, t' (j) + \tau'(j) \}$. For all jobs~$j$ with $r_j = 0$, we set $r_j' = 0$.
		Overall, we define a new instance

	\[
	I' \coloneqq (\vec{p}', \vec{w}, \vec{r}')
	.
	\]
	Then, the schedule produced by $\alg$ for $I'$ is exactly $\bar{S}$ by construction. In particular, we have $\alg (I) \leq \alg (I')$.

	We now argue that $\vec r' \le \frac{b^3}{2b-1} \vec r$. To this end, consider a fixed job~$j$. If $r'_j=r_j$ there is nothing to show. Hence, let $r'_j=t'+\tau'$. We compare the total time devoted to each job~$k$ before the start times of the probing performed when $j$ is released.
	If $k$ is not completed before $t$ in $\alg(I)$, the total probing time of $k$ before $t$ in $\alg(I)$ is the same as the total probing time of $k$ before $t'$ in $\alg(I')$.
	If $k$ is completed before $t$ in $\alg(I)$, the total time devoted to job~$k$ is $w_k\sum_{i=-\infty}^{\jobrank_k - 1} b^i + p_k = w_k \big( \frac{b^{\jobrank_k}}{b-1} + \frac{p_k}{w_k}\big)$. On the other hand, in $\alg(I')$ the total time devoted to~$k$ is $w_k\sum_{i=-\infty}^{\jobrank_k} b^i = w_k\frac{b^{\jobrank_k + 1}}{b-1}$.
	The ratio between these two elapsed times is
	\[
	\frac{b^{\jobrank_k + 1}}{b^{\jobrank_k} + (b-1) \frac{p_k}{w_k}} < \frac{b^{\jobrank_k + 1}}{b^{\jobrank_k} + (b-1) \cdot b^{\jobrank_k - 1}} = \frac{b^2}{2b-1}.
	\]
	Since this holds for every job~$k$, we have $t' \leq \frac{b^2}{2b-1} t$.

	Using \cref{lem:probing_end_bound} we obtain
	\[
	r_j' =
	t' + \tau' \leq b t'
	\leq
	\frac{b^3}{2b-1} t
	\leq \frac{b^3}{2b-1} r_j
	. \qedhere
	\]
\end{proof}

For every time $\theta > 0$ and job~$j$, we denote by
\[
 \jobrank_j (\theta) \coloneqq \max \bigl\{ q\in\mathbb{Z} \bigm\vert \exists \text{ probing } (t, j, w_j b^{q-1}) \text{ of } \alg(I') \text{ with } t + w_j b^{\jobrank-1} < \theta \bigr\}
\]
the rank of job~$j$ at time~$\theta$ in $\alg(I')$, where $(t, j, w_j b^{\jobrank-1})$ describes the probing of job~$j$ starting at time~$t$ for a probing time of $w_j b^{\jobrank-1}$ and $I'$ is defined as in~\cref{lem:release_dates_I_I'}.
In other words, $\jobrank_j(\theta)=q$ holds if and only if job~$j$ has already been probed for $w_j b^{q-1}$ before time $\theta$, and the next probing operation of job~$j$ will be a probing for $w_j b^q$ if not completed.
Note that the function $\jobrank_j (\theta)$ is piecewise constant and increases by $1$ whenever a probing has ended. Observe that $\jobrank_j (\infty)$ is exactly the value~$\jobrank_j$ defined in the proof of \cref{lem:release_dates_I_I'}.
Further,  we denote by
\[
	\jobrank (\theta) \coloneqq \min \big\{ \jobrank_j (\theta) \bigm| j \in [n]: r_j < \theta \text{ and } C_j^{\alg} (I') \ge \theta
	 \bigr\}
\]
the minimal rank among all jobs that are released and not completed at time~$\theta$. If there is no such job at time~$\theta$, we set $\jobrank(\theta) = \infty$.
Note that the function $\jobrank(\theta)$ decreases whenever new jobs are released and increases when the minimum rank of all released and not completed jobs increases. The definition is made so that $q$ is left-continuous.
We also define
\begin{align*}
\jobrank_{\min} &\coloneqq \min_{j \in [n]} \min \bigl\{\jobrank(r_j),\ q(C_j^{\alg}(I'))\bigr\} = \min_{j \in [n]} \min \bigl\{\jobrank(r_j),\ \jobrank_j(\infty)\bigr\}, \\
\jobrank_{\max} &\coloneqq \max_{j \in [n]} \jobrank_j(\infty).
\end{align*}

\begin{lemma}\label{lem:release_dates_I'_I''}
	For an instance $I' = (\vec p', \vec w, \vec r')$ with the properties described in \cref{lem:release_dates_I_I'} the instance $I'' = (\vec p'', \vec w, \vec r')$ with $\vec p'' \coloneqq \frac{b}{b-1} \vec p'$ satisfies $\alg(I') \le \wsetf(I'')$.
\end{lemma}

\begin{proof}

Observe that

\[
	p''_j =
 	 \frac{b}{b-1} \, p'_j = w_j \frac{b^{\jobrank_j + 1}}{b - 1} = \sum_{i=-\infty}^{\jobrank_j} w_j b^i = Y_j^{\alg}(I',C_j^{\alg}(I'))
 ,
\]
i.e., $p''_j$ is obtained by adding all failing times of $j$ to the processing time $p'_j$.
We show for all jobs $j$ that
\[
C_j^{\alg} (I') \leq C_j^{\wsetf} (I'').
\]
Let $0 = r^{(0)} < \cdots < r^{(n')}$ be all distinct release dates, and for every $i \in \{0,1,\dotsc,n'\}$ and $\jobrank \in \Z$ let $J^{(i)} \coloneqq \{j \in [n] \mid r_j = r^{(i)}\}$ and
\[
e^{(i)} (\jobrank) \coloneqq \inf \big\{ t \bigm| t > r^{(i)} \text{ and } \jobrank(t) \geq \jobrank +1\big\}
\]
be the end time of round $\jobrank$ of jobs released at $r^{(i)}$, i.e., at time $e^{(i)}(q)$ every job~$j \in J^{(i)}$ not completed at an earlier round has been probed for $\sum_{\widehat \jobrank = -\infty}^\jobrank w_jb^{\widehat \jobrank}=w_j\frac{b^{\jobrank+1}}{b-1}$ in total in $\alg (I')$.

In the following we will show that for all $\jobrank \in \{\minrank - 1, \minrank, \dotsc,\jobrank_{\max}\}$ and for all $i \in \{0,\dotsc,n'\}$ every job~$j\in J^{(i)}$
has the same elapsed time at $e^{(i)}(\jobrank)$ in both schedules, i.e., $Y^{\alg}_j(I', e^{(i)}(\jobrank)) = Y^{\wsetf}_j(I'', e^{(i)}(\jobrank))$. This will be done by induction on~$\jobrank$.
\Cref{fig:release_dates_I'_I''} illustrates the end times of rounds $e^{(i)}(q)$ as well as the value $\minrank$.
\begin{figure}
 \centering
    \begin{tikzpicture}
        \draw[thick, -{stealth}] (0,0) -- (15,0) node[right] {\footnotesize $t$};

        \foreach \x/\i/\col in {0/0/strong1,3/1/strong2,7/2/strong3!80!black,10/3/strong4} {
            \draw[\col, ultra thick] (\x,.2) -- (\x,-.2) node[below] {\footnotesize $r^{(\i)}$};
        }

        \foreach \x/\i in {.9/0,3.9/1,7.9/2,11.1/3} {
            \draw[{stealth}-] (\x, -.6) -- (\x, -1) node[below] {\footnotesize $e^{(\i)}(\minrank - 1)$};
        }

        \foreach \x/\q in {.9/4,1.8/5} {
            \draw[thick, strong1] (\x, .075) -- (\x, -.075) node[below] (e-0-\q) {\scriptsize $e^{(0)}(\q)$};
        }

        \fill[pattern=north east lines, pattern color=strong5]
            (0,.1) rectangle (.9, -.1);

        \foreach \x/\q in {3.9/4,5.2/5,6.1/6} {
            \draw[thick, strong2] (\x, .075) -- (\x, -.075) node[below] (e-1-\q) {\scriptsize $e^{(1)}(\q)$};
        }

        \fill[pattern=north east lines, pattern color=strong5]
            (3,.1) rectangle (5.2, -.1);

        \foreach \x/\q in {7.9/4,9.1/5} {
            \draw[thick, strong3!80!black] (\x, .075) -- (\x, -.075) node[below] (e-2-\q) {\scriptsize $e^{(2)}(\q)$};
        }

        \fill[pattern=north east lines, pattern color=strong5]
            (7,.1) rectangle (7.9, -.1);

        \foreach \x/\q in {11.1/4,12/5,13.2/6,14.3/7} {
            \draw[thick, strong4] (\x, .075) -- (\x, -.075) node[below] (e-3-\q) {\scriptsize $e^{(3)}(\q)$};
        }

        \fill[pattern=north east lines, pattern color=strong5]
            (10,.1) rectangle (12, -.1);

        \node[strong1, below=-1ex of e-1-6] (e-0-6) {\scriptsize $e^{(0)}(6)$};

        \node[strong3!80!black, below=-1ex of e-3-6] (e-2-6) {\scriptsize $e^{(2)}(6)$};

        \node[strong3!80!black, below=-1ex of e-3-7] (e-2-7) {\scriptsize $e^{(2)}(7)$};
        \node[strong2, below=-1ex of e-2-7] (e-1-7) {\scriptsize $e^{(1)}(7)$};
        \node[strong1, below=-1ex of e-1-7] (e-0-7) {\scriptsize $e^{(0)}(7)$};

    \end{tikzpicture}

    \caption{Illustration of the situation in the proof of \cref{lem:release_dates_I'_I''}.
    The release dates $r^{(i)}$ and the end points of the rounds in $\alg$ subdivide the time axis into intervals.
    The shaded intervals do \emph{not} contain any successful probings.
    The value $\minrank = 5$ is chosen such that the end of the round $e^{(i)}(\minrank - 1)$ is before the next release date and such that between $r^{(i)}$ and $e^{(i)}(\minrank - 1)$ no job completes.
    \cref{clm:qmin} refers to the endpoints $e^{(i)}(\minrank - 1)$, while \cref{clm:YWSETF-YALG} refers to all subsequent endpoints of rounds. Note, that some endpoints (e.g, $e^{(2)}(6)$ or $e^{(2)} (7)$) lie after one or more subsequent release dates.
    }

    \label{fig:release_dates_I'_I''}
\end{figure}
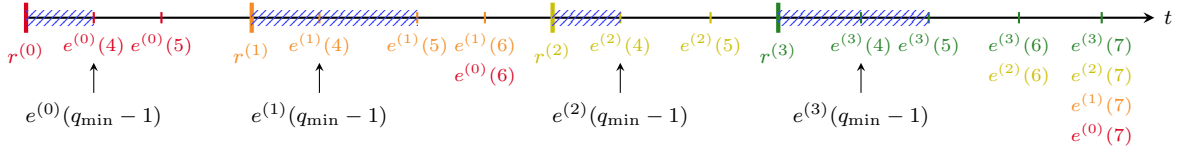
For the sake of simplicity, we write for the remainder of the proof $Y^{\alg}_j (t) = Y^{\alg}_j (I',t)$ and $ Y^{\wsetf}_j (t) = Y^{\wsetf}_j (I'',t)$. We start with the base case $\jobrank = \minrank-1$.

\begin{claim} \label{clm:qmin}
For all $i\in\{0,\ldots,n'\}$ and for all $j\in J^{(i)}$
we have
\[
	Y^{\alg}_j (e^{(i)}(\minrank-1)) = Y^{\wsetf}_j (e^{(i)}(\minrank-1))=w_j\frac{b^{\minrank}}{b-1}
	.
\]
\end{claim}
\newcommand{\nextrankH}[2]{\psi^{(#2)}_{#1}}
\begin{proof}[Proof of \cref{clm:qmin}]
	For every $i$, the definition of $e^{(i)}$ and the left-continuity of $q$ imply that $\jobrank(t) < \minrank$ for all $t \in [r^{(i)}, e^{(i)}(\minrank-1)]$. Therefore, no job is released or completed by $\alg$ during $[r^{(i)},e^{(i)}(\minrank-1)]$. Hence,
 \begin{equation}
  e^{(i)}(\minrank-1) < r^{(i+1)}. \label{e_before_r}
 \end{equation}

 Assume for a contradiction that the claim is wrong, and let $i$ be the minimum index for which the statement is violated.
 Clearly, at time $r^{(i)}$ all jobs from $J^{(i)}$ have not been processed in either $\wsetf$ and $\alg$, while all earlier released jobs with $r_j < r^{(i)}$ have either been completed in both $\wsetf$ and $\alg$ or have in both schedules an elapsed time of at least $w_j \frac{b^{\minrank}}{b-1}$, using \eqref{e_before_r} and the minimality assumption. Thus, both schedules start by processing only jobs from $J^{(i)}$. For $\alg$ we know that these are contiguously processed at least until time $e^{(i)}(\minrank-1)$.
 Since no job from $J^{(i)}$ is completed before $e^{(i)}(\minrank - 1)$, we know that $p_j' \ge w_j b^{\minrank}$, and thus $p_j'' \ge w_j \frac{b^{\minrank+1}}{b-1}$ for all $j \in J^{(i)}$.
 Moreover, for any job~$j\in J^{(i)}$ we have
	\begin{equation} \label{eq:Y_ALG_b}
		 Y_j^{\alg} (e^{(i)}(\minrank-1)) = \sum_{q=-\infty}^{\minrank-1} w_j b^q = w_j \frac{b^{\minrank}}{b-1}.
	\end{equation}
	Therefore,
		\begin{equation} \label{eq:e^0_qmin}
			e^{(i)}(\minrank - 1) - r^{(i)} = \sum_{j \in J^{(i)}} Y_j^{\alg} (e^{(i)}(\minrank-1))  \stackrel{\eqref{eq:Y_ALG_b}}= \frac{b^{\minrank}}{b-1}  \sum_{j \in J^{(i)}} w_j
			.
	\end{equation}
	Until the first completion time \wsetf{} processes each job~$j \in J^{(i)}$ with a rate of $w_j/\sum_{k \in J^{(i)}} w_k$. Suppose the first completing job~$j$ is completed at time $ t < e^{(i)}(\minrank-1)$. Then it has elapsed time
	\[
	Y_j^{\wsetf}(t) = \frac{w_j}{\sum_{k \in J^{(i)}} w_k} \cdot  (t - r^{(i)}) < \frac{w_j}{\sum_{k \in J^{(i)}} w_k} \cdot \bigl(e^{(i)}(\minrank - 1) - r^{(i)}\bigr) \stackrel{\eqref{eq:e^0_qmin}}= w_j \cdot \frac{b^{\minrank}}{b-1} \le p_j'',
	\]
	a contradiction since $j$ would have elapsed time strictly less than its processing time. Therefore, we have $C_j^{\wsetf}(I'') \ge e^{(i)}(\minrank-1)$ for all $j \in J^{(i)}$, and the elapsed times of these jobs~$j$ at this time can be computed as
	\begin{align*}
		Y_j^{\wsetf}(e^{(i)}(\minrank-1))
		= \frac{w_j}{\sum_{k \in J^{(i)}} w_k} \cdot \bigl(e^{(i)}(\minrank - 1) - r^{(i)}\bigr)
		\stackrel{\eqref{eq:e^0_qmin}}= w_j \frac{b^{\minrank}}{b-1}
		\stackrel{\eqref{eq:Y_ALG_b}}= Y_j^{\alg} (e^{(i)}(\minrank-1)).
	\end{align*}
 This contradicts the assumption that the claim is wrong for $i$, concluding the proof of \cref{clm:qmin}.
\end{proof}

 Now we generalize the claim to arbitrary $q$. We use the notation from \cref{lem:release_dates_I_I'}, so that $p_j' = b^{\jobrank_j}$ for every $j \in [n]$, and thus $p_j'' = \frac{b^{\jobrank_j+1}}{b-1}$.

	\begin{claim} \label{clm:YWSETF-YALG}
  For all $q \in \{\minrank - 1,\minrank,\dotsc,\maxrank\}$, for all $i \in \{0,\dotsc,n'\}$, and for all $j \in J^{(i)}$ we have
  \[
  Y^{\wsetf}_j (e^{(i)} (q)) = Y^{\alg}_j (e^{(i)} (q) )
  .
  \]
  Moreover, if $\jobrank_j \ge q$, then this quantity is equal to $w_j\frac{b^{q+1}}{b-1}$, and if $\jobrank_j = q$, then $C_j^{\wsetf}(I'')=e^{(i)}(q)$.
	\end{claim}
\begin{proof}[Proof of \cref{clm:YWSETF-YALG}]
 As announced, this is shown by induction on $q$. The case $q = \minrank-1$ follows from \cref{clm:qmin} because every job has $\jobrank_j > \minrank - 1$. Now we assume that $q \ge \minrank$ and the claim is true for all $q'$ with $\minrank - 1 \le q' < q$. For all $i \in \{0,\dotsc,n'\}$, all jobs $j \in J^{(i)}$ with $\jobrank_j < q$ are completed in $\alg$ before time $e^{(i)}(\jobrank_j)$. The induction hypothesis implies that $Y_j^{\alg}(e^{(i)}(\jobrank_j)) = Y_j^{\wsetf}(e^{(i)}(\jobrank_j)) = w_j \frac{b^{\jobrank_j + 1}}{b-1} = p_j''$. Hence, also in the \wsetf{} these jobs are completed until time $e^{(i)}(\jobrank_j) \le e^{(i)}(q-1)$. In the following we thus restrict to jobs with $\jobrank_j \ge q$.

  Assume for a contradiction that the claim is wrong for $q$, and let $i$ be the smallest index for which the statement fails. Let $e \coloneqq e^{(i)}(q)$, and let $\mathcal I(e) \coloneqq \bigl\{i' \in \{0,\dotsc,n'\} \mid e^{(i')}(q) = e\bigr\} = \{\underline i,\underline i+1,\dotsc,\overline i\}$.
  Let \[J' \coloneqq \biggl\{j \in \bigcup_{i \in \mathcal I(e)} J^{(i)} \biggm\vert \jobrank_j \ge q\biggr\}.\]
 All not completed jobs~$j \in J^{(i')}$ with $i' < \underline i \le i$ have $e^{(i')}(q) < r^{(\underline i)}$, and hence, as $i$ was chosen to be minimal, elapsed time of at least $w_j\frac{b^{q+1}}{b-1}$ in the schedules constructed by both $\alg$ and $\wsetf$. Therefore, none of these jobs are probed in the interval $[r^{(\underline i)}, e]$ by $\alg$.
 We know by the property of $I'$ that release dates are only at ends of probings, which implies that during the union of intervals
	\begin{equation}\label{eq:busy_interval}
	\mathcal{T}\coloneqq\bigcup_{i' = \underline i}^{\overline i-1} (e^{(i')}(q-1),\,r^{(i'+1)}] \,\cup\, (e^{(\overline i)}(q-1),\,e]
	\end{equation}
	each job $j \in J'$ is probed for $w_j b^\jobrank$ by \alg. Observe that in this formula it is possible that the left bound of an interval is not smaller than the right bound in which case the interval is empty. Therefore, using the property of $I'$ that the duration of each probing is exactly the probing time itself we have
	\begin{equation}\label{eq:busy_interval_length}
	\sum_{i' = \underline i}^{\overline i-1} \big(r^{(i'+1)} - e^{(i')}(q-1)\bigr)^+ + e -e^{(\overline i)}(q-1) = \sum_{j\in J'} w_j b^{\jobrank},
	\end{equation}
	where $(\cdot)^+$ denotes the positive part. We next argue that in the \wsetf{} schedule for all $t \in \mathcal T \setminus \{e\}$ every job~$j \in J'$ has elapsed time~$Y_j^{\wsetf}(t) < w_j \frac{b^{q+1}}{b-1}$. Since all $\jobrank_j \ge q$ for $j \in J'$, this, along with the lower bound for the elapsed processing times of previously released jobs shown above, implies that also \wsetf{} is not running any jobs released before $r^{(\underline i)}$ during this time, and, in view of the construction of the instance~$I''$, that no job from $J'$ is completed within this time. To show the claimed statement, consider the first moment in time~$t \in \mathcal T$ at which some job~$j \in J'$ has elapsed processing time~$Y_j^{\wsetf}(t) \ge w_j \frac{b^{q+1}}{b-1}$. Let $i^* \in \mathcal I(e)$ be the index such that $t \in (e^{(i^*)}(q-1), r^{(i^*+1)})$. Since \wsetf{} processes $j$ directly before $t$, the job~$j$ must have minimum weighted elapsed time among all released and unfinished jobs, so in particular among all jobs~$k \in \bigcup_{i'=\underline i}^{i^*} J^{(i')}$ with $\jobrank_{k} \ge q$, i.e.,
	\[
		\frac{Y^{\wsetf}_{k} (t)}{w_k}
		\geq \frac{Y^{\wsetf}_j (t)}{w_j}
		\geq \frac{b^{\jobrank+1}}{b-1}.
	\]
 Furthermore, no job~$k \in J' \cap J^{(i')}$ for $i' \in \mathcal I(e)$ is processed by \wsetf{} during an interval $(r^{(i'')}, e^{(i''+1)}(q-1)]$ for some $i'' > i'$, so that the time spent on $k$ during $\mathcal T \cap (0,t)$ is simply $Y_k^{\wsetf}(t) - Y_k^{\wsetf}(e^{(i')}(q-1))$.
	Therefore, $\wsetf$ has in total spent
	\[
		\sum_{i'=\underline i}^{i^*} \sum_{k \in J' \cap J^{(i')}} \bigl(Y_k^{\wsetf}(t) - Y_k^{\wsetf}(e^{(i')}(q-1))\bigr)
		\geq{} \sum_{i'=\underline i}^{i^*} \sum_{k \in J' \cap J^{(i')}} w_k \frac{b^{\jobrank+1}}{b-1}
		- w_k \frac{b^{\jobrank}}{b-1}
		=	\sum_{i'=\underline i}^{i^*} \sum_{k \in J' \cap J^{(i')}} w_k b^\jobrank
		,
	\]
	on jobs~$k \in J' \cap \bigcup_{i'=\underline i}^{i^*} J^{(i')}$ during $\mathcal T \cap (0,t)$, where we used the induction hypothesis for $q$. Thus, within $\mathcal T \cap (0,t)$, the strategy~$\alg$ has had enough time to probe all jobs from this set for $w_k b^{\jobrank}$. In particular it has probed the jobs from $J^{(\underline i)}$ for this time, so by definition of $e^{(\underline i)}(q)$, we have $t \ge e^{(\underline i)} (\jobrank) = e$, hence $t = e$.

 By continuity, all jobs~$j \in J'$ have $Y_j^{\wsetf}(e) \le w_j \frac{b^{q+1}}{b-1}$. As $\mathcal D_b$ manages to process all these jobs long enough during $\mathcal T$ so that at time~$e$ they have accumulated exactly this amount of elapsed time and \wsetf{} always processes jobs from $J'$ during $\mathcal T$, it must have assigned in total the same amount of time within $\mathcal T$ to jobs from $J'$, i.e., $\sum_{j \in J'} Y_j^{\wsetf}(e) = \frac{b^{q+1}}{b-1} \sum_{j \in J'} w_j$. Together with the above upper bound for each individual job~$j \in J'$, this implies that $Y_j^{\wsetf}(e) = w_j \frac{b^{q+1}}{b-1}$ for all $j \in J'$. Consequently, $C_j^{\wsetf}(I'') = e$ if $\jobrank_j = q$. Thus, the statement is true for all $i' \in \mathcal I(e)$, so in particular for $i$, contradicting our assumption.
\end{proof}
\cref{clm:YWSETF-YALG} implies that $C^{\alg}_j (I') \le e^{(i)}(\jobrank_j) = C^{\wsetf}_j (I'')$ for all jobs~$j \in J^{(i)}$ for all $i$, and thus, $\alg (I') \leq \wsetf (I'')$.
\end{proof}

\begin{proof}[Proof of \cref{thm:trivial_bound_release_dates}.]
By \cref{thm:WSETF} we know that \wsetf{} is $2$-competitive for \threefield{1}{r_j,\,\mathrm{pmtn}}{w_jC_j}. Moreover, using \cref{lem:release_dates_I_I',lem:release_dates_I'_I''} together with \cref{lem:b_factor_opt}, we obtain
\[
\alg(I) \leq \alg(I') \leq \wsetf(I'') \leq 2 \, \opt(I'') \leq \frac{2b}{b-1} \opt(I') \leq \frac{2b^4}{2b^2 - 3b + 1} \opt(I). \qedhere
\]
\end{proof}

\subsection{Parallel Machines} \label{subsec:parallel}
In this subsection we consider instances $(\vec{p}, \vec{1}, \vec{0}, m)$ of \threefield{\mathrm P}{}{C_j} with processing times~$\vec{p}$, unit weights, trivial release dates, and $m$ identical parallel machines. For this setting we have to extend the definition of kill-and-restart strategies and in particular of the $b$-scaling strategy. The set of all currently active probings are added to each state. The intervals chosen by an action need not be disjoint anymore, but instead it is required that any point in time~$t$ be covered by at most $m$ intervals. Finally, the transition function~$T_I$, mapping a state and an action to a new state, requires several modifications: It has to deal with the situation that multiple jobs are simultaneously completed, the elapsed probing times of the active probings given in the state have to be taken into account in the determination of the next completion time, and the active probings at the next decision time have to be determined. We do not go into more detail for the general definition.

For the $b$-scaling strategy~$\alg$, we perform the same sequence of probings, but assign them to the parallel machines in a list scheduling manner, i.e., every probing is scheduled on the first available machine. Moreover, at the moment when the number of remaining jobs becomes less than or equal to the number of machines, the jobs are not aborted anymore. The formalization of the actions chosen upon the completion of any job is straightforward. However, it may be not obvious how to formalize the action chosen at time~$0$ if $n$ is not divisible by $m$, because then the last planned probings of each round~$q$ need not be synchronized. The infinitesimal probing makes it impossible to define the probing intervals in an inductive way. This can be resolved by observing that the planned probings are scheduled in an SPT manner, and it is well-known that in an SPT schedule, every $m$th job goes to the same machine. Hence, it is possible to split the family of all probing operations into $m$ subsequences, each containing every $m$th element. More precisely, if the probings of the single-machine strategy are denoted by $\pi_{(q,j)} = (t_{(q,j)}, j, \tau_{(q,j)})$ for $(q,j) \in \Z\times [n]$, then consider the bijection~$\iota \colon \Z \times [n] \to \Z$ with $\iota(q,j) \coloneqq nq + j$, specifying the probing order, and the subfamilies $a_i = (\pi_{\iota^{-1}(k m + i)})_{k \in \Z}$ of probings to be assigned to each machine~$i \in [m]$. We define the probings of the initial action for $m$ machines as
 \[\pi_{(i,k)} = \bigg(\sum_{\ell=-\infty}^{k-1} \tau_{\iota^{-1}(\ell m+i)}, (\iota^{-1}(km+i))_2, \tau_{\iota^{-1}(km+i)}\biggr),\]
 where $\sum_{k=-\infty}^{k-1} \tau_{\iota^{-1}(km+i)} < \infty$ because of the absolute convergence of the sequence of all probing times smaller than any given one. The entire initial action of the $m$-machine $b$-scaling strategy is then the family~$(\pi'_{(i,k)})_{i \in [m], k \in \Z}$.

\TrivialBoundParallel

\begin{proof}
Let $I = (\vec{p}, \vec{1}, \vec{0}, m)$ be an arbitrary instance for \threefield{\mathrm P}{}{C_j}. We assume that $n > m$ as otherwise $\alg$ is optimal.
We define a new instance $I' = (\vec{p}', \vec{1}, \vec{0}, m)$ with $p_j':=b^{\jobrank_j}$ and $\jobrank_j = \lceil \log_b(p_j) \rceil$.
By definition of $\alg$,
the last job executed last on each machine is
run non-preemptively, i.e., it is probed for an infinite amount of time.
We denote by $\infprobingjobs$ this set of $m$ jobs
run non-preemptively in instance $I'$.

We first show that $\alg(I)\leq\alg(I')$.
	Denote by $\pi_0=(t_0,j_0,\tau_0)$ the first probing in $\alg(I)$ in which a job is completed and denote by
	$\pi_0,\pi_1,\ldots,\pi_N$ the sequence of all probings started at or after time $t_0$, ordered by starting time of the probing operation. Each probing $\pi_k=(t_k,j_k,\tau_k)$ is in one-to-one correspondence with
	a probing $\pi_k'=(t_k',j_k,\tau_k)$ in $\alg(I')$.
	Denote by $\lambda_{k1}\leq \lambda_{k2} \leq \ldots \leq \lambda_{km}$ the ordered loads of the $m$ machines in $\alg(I)$ at time $t_k$, where the load of a machine at time $\theta$ is the last end of a probing started before $\theta$ on that machine. In particular, $\lambda_{k1}=t_k$ because probing~$\pi_k$ starts at time $t_k$ on the least loaded machine. Similarly, denote by
	$\lambda_{k1}' \leq \ldots \leq \lambda_{km}'$ the $m$ ordered machine loads at time $t_k'$ in $\alg(I')$.

We show by induction on $k$ that $\lambda_{ki}\leq \lambda_{ki}'$ holds for all $i\in[m]$.
	The base case $k=0$ is trivial, since the probings are identical in $\alg(I)$ and $\alg(I')$ until time $t_0$, hence $\lambda_{0i}=\lambda_{0i}', \forall i\in[m]$.
	Then, let $k\in\{0,\ldots,N-1\}$, and denote by $\delta_k=\min(p_{j_k},\tau_k)$ the actual duration of the probing operation $\pi_k$, so the loads $\vec{\lambda}_{k+1}$ at time $t_{k+1}$ in $\alg(I)$ are a permutation of $(\lambda_{k1}+\delta_k,\lambda_{k2},\ldots,\lambda_{km})$. Similarly, the loads
	$\vec{\lambda}_{k+1}'$ at time $t_{k+1}'$ in $\alg(I')$ are a permutation of $(\lambda_{k1}'+\delta_k',\lambda_{k2}',\ldots,\lambda_{km}')$, with $\delta_k'=\min(p_{j_k}',\tau_k)\geq\delta_k$. By induction hypothesis we have
	$\lambda_{k1}+\delta\leq\lambda_{k1}'+\delta\leq\lambda_{k1}'+\delta'$ and
	$\lambda_{ki}\leq\lambda_{ki}'$ for all $i=2,\ldots,m$. This shows the existence of two permutations $\sigma$ and $\sigma'$ such that
	$\lambda_{k+1,\sigma(i)}\leq \lambda_{k+1,\sigma'(i)}'$ holds for all $i\in[m]$, which in turn implies that the ordered loads satisfy $\lambda_{k+1,i}\leq \lambda_{k+1,i}'$, for all $i\in[m]$.
	This concludes the induction.
	The end of probing~$\pi_k$ in $\alg(I)$ is
	$t_k+\delta_k=\lambda_{1k}+\delta_k\leq \lambda_{1k}'+\delta_k'=t_k'+\delta_k'$, where the latter corresponds to the end of probing~$\pi_k'$ in $\alg(I')$. Clearly, this implies $\alg(I)\leq\alg(I')$.

Let $\minrank \coloneqq \min_j \jobrank_j$ and $\maxrank \coloneqq \max_{j \notin \infprobingjobs} \jobrank_j$.
For $q<\maxrank$
denote by $\tALG_i (\jobrank)$ the
last end of a $b^\jobrank$-probing operation
on machine~$i$ in $\alg(I')$. For $q=\maxrank$ we need to define $T_i'(\maxrank)$ differently to take into account the $m$ jobs probed for an infinite time. Note that the jobs $j\in\infprobingjobs$
have $p_j'\geq b^{\maxrank}$ and for each $i\in[m]$ there is exactly one job $j(i)\in\infprobingjobs$ that is completed on machine~$i$.
If job $j(i)$ has already been probed for $b^{\maxrank}$ before being run non-preemptively, we define $T_i'(\maxrank)$ as the last end of a
$b^{\maxrank}$-probing operation on $i$; otherwise
we define $T_i'(\maxrank)$ as the first point in time where $j(i)$ has been processed for at least $b^{\maxrank}$. This ensures that every job with $p_j'\geq b^{\maxrank}$ is processed on some machine $i\in[m]$ during one interval of length $b^{\maxrank}$ contained in $[T_i'(\maxrank-1),T_i'(\maxrank)]$.

We define another instance $I'' = (\vec{p}'', \vec{1}, \vec{0}, m)$ with processing times
\[
	p_j'' \coloneqq
	\displaystyle \sum_{\jobrank = -\infty}^{\jobrank_j} b^{\jobrank}
	\text{ if } \jobrank_j\leq \maxrank,
	\quad \text{and} \quad
	p_j'' \coloneqq
	\displaystyle \sum_{\jobrank = -\infty}^{\maxrank} b^{\jobrank} + p_j'
	\text{ if } \jobrank_j\geq \maxrank+1.
\]
For all $j$ such that $\jobrank_j\leq \maxrank$ we have $p_j'' = \frac{b^{\jobrank_j+1}}{b-1} = \frac{b}{b-1} p_j'$ and otherwise we have $p_j'' = \frac{b^{\maxrank+1}}{b-1} + p_j' \leq  \bigl( \frac{1}{b-1} + 1 \bigr) p_j' = \frac{b}{b-1} p_j'$, i.e., $p''_j \leq \frac{b}{b-1} p'_j$ holds for all~$j$.

Consider the schedule~$S''$ of $\rr$ for the instance $I''$.
Denote by $\tSETF (\jobrank)$ the first point in time~$t$ with $Y_j^{\rr} (t) = \frac{b^{\jobrank + 1}}{b-1}$ for all jobs~$j$ with $p_j'' \geq \frac{b^{\jobrank + 1}}{b-1}$.
In particular, every job with $\jobrank_j \leq \maxrank$ is completed at $\tSETF (\jobrank_j)$ in the schedule $S''$.
We first prove by induction on $\jobrank$ that $\sum_{i=1}^m \tALG_i (\jobrank) = m \, \tSETF (\jobrank)$, for all $\jobrank=\minrank,\ldots,\maxrank$.
Since for all jobs~$j$ we have $\jobrank_j \geq \minrank$ no job is completed in $\alg$ and \rr{} until $\tALG_i (\minrank-1)$ for any machine~$i$ and $\tSETF (\minrank-1)$, respectively. Thus, we have
\[
	\sum_{i=1}^m \tALG_i (\minrank) = \sum_{j=1}^n \sum_{\jobrank=-\infty}^{\minrank} b^{\jobrank} = n \frac{b^{\minrank + 1}}{b-1} = m \frac{n}{m} \frac{b^{\minrank + 1}}{b-1} = m \tSETF (\minrank)
\]
Let $\jobrank > \minrank$.
At $\tSETF (\jobrank-1)$, there are, by definition, more than $m$ remaining jobs and each job~$j$ with $\jobrank_j \geq \jobrank$ has already been processed for $\frac{b^{\jobrank + 1}}{b-1}$ in the \rr{}-schedule for instance $I''$.
Therefore, each job receives a rate of $\frac{m}{n_{\geq \jobrank}}<1$ in this round, where $n_{\geq \jobrank}$ is the number of jobs with $p'' \geq \frac{b^{\jobrank + 1}}{b-1}$ and,
thus, $\tSETF (\jobrank)=\tSETF (\jobrank-1)+\frac{n_{\geq \jobrank}}{m}\cdot b^{\jobrank}$. On the other hand,
if $q<\maxrank$
$\alg$
probes $n_{\geq \jobrank}$ jobs for exactly $b^{\jobrank}$, so $\sum_{i=1}^m \tALG_i (\jobrank) = \sum_{i=1}^m \tALG_i(\jobrank-1) + n_{\geq \jobrank} b^{\jobrank}$ because there is no idling-time in the schedule.
For $q=\maxrank$, our definition of $T_i'(\maxrank)$ ensures that
$\sum_{i=1}^m \tALG_i (\maxrank) = \sum_{i=1}^m \tALG_i(\maxrank-1) + n_{\geq \maxrank} b^{\maxrank}$
holds as well.
Then, the claim follows  from the induction hypothesis.

Consider a job $j \notin \infprobingjobs$. We have
\[
	C^{\alg}_j (I')
	\leq \max_i \tALG_i (\jobrank_j)
	\leq \min_i \tALG_i (\jobrank_j) + b^{\jobrank_j}
	\leq \frac{1}{m}\sum_{i=1}^m \tALG_i (\jobrank_j) + b^{\jobrank_j}= \tSETF (\jobrank_j)+b^{\jobrank_j}=C_j^{\rr}(I'')+p_j'
	,
\] where the second inequality comes from the fact that the probing operations are done in a list-scheduling manner.
For a job $j \in \infprobingjobs$, we have $C^{\alg}_j (I') \leq \tALG_{i(j)} (\maxrank) +p_j'$,
where~$i(j)$ denotes the machine on which~$j$ is probed for an infinite amount of time.
This implies
\[
\sum_{j\in \infprobingjobs}C^{\alg}_j (I')
\leq\sum_{j\in \infprobingjobs} T'_{i(j)}({\maxrank}) + p_j' = m T''(\maxrank) + \sum_{j\in \infprobingjobs} p_j' =
\sum_{j\in\infprobingjobs} T''(\maxrank) + p_j'
=\sum_{j\in \infprobingjobs} C_j^{\rr}(I'').
\]

Since \rr{} is $2$-competitive for \threefield{\mathrm P}{\mathrm{pmtn}}{C_j} (see~\cite{MPT94}), we overall obtain
\begin{align*}
 \alg(I)\leq \alg(I') &\leq
 \rr(I'') + \sum_{j \notin \infprobingjobs} p_j' \leq 2 \opt(I'') + \opt(I')\\
 &\leq \Big( \frac{2b}{b-1}+1 \Big) \opt(I') \leq \Big(\frac{2b}{b-1}+1\Big)\cdot b \cdot \opt(I) = \frac{3b^2-b}{b-1}\cdot \opt(I)
\end{align*}
where we used \cref{lem:b_factor_opt} for the last two inequalities.
\end{proof}

\section{Conclusion}
We studied kill-and-restart as well as preemptive strategies for the problem of minimizing the sum of weighted completion times and gave a tight analysis of the deterministic and randomized version of the natural $b$-scaling strategy for \threefield{1}{}{w_jC_j} as well as of \wsetf{} for \threefield{1}{r_j,\,\mathrm{pmtn}}{w_jC_j}. 

We hope that this work might lay a basis for obtaining tight bounds on the performance of the $b$-scaling strategy for more general settings such as non-trivial release dates and parallel machines. 
Moreover, we think that the class of kill-and-restart strategies combines the best of two worlds. On the one hand, they allow for interruptions leading to small competitive ratios in contrast to non-preemptive algorithms, on the other hand, they reflect the non-preemptive property of only completing a job if it has been processed as a whole.

\paragraph{Acknowledgements.}
We thank Sungjin Im for helpful comments on an earlier version of this manuscript.

 \bibliography{forgetful}

 \appendix
\section{Technical Lemmas}

\begin{lemma} \label{lem:Cholesky}
 Let $L \in \N_{<0}$, $b \ge 1$, and $\vec B = (\frac 1 2 b^{\min(\ell, m)})_{0 \le \ell, m \le L}$ Then the Cholesky decomposition of $\vec B$ is $\vec B = \vec Y^\top \vec Y$ with $\vec Y = \bigl(\sqrt{\frac{b^\ell - b^{\ell-1} \cdot \mathds 1_{\ell \ge 1}}{2}} \cdot \mathds 1_{m \ge \ell}\bigr)_{0 \le \ell, m \le L}$.
\end{lemma}
\begin{proof}
 Obviously,
 \[
  \vec{Y}^{\vphantom{\top}}=\frac{1}{\sqrt{2}}
  \begin{pmatrix}
   1 & 1 &  \cdots      & 1\\
   & \sqrt{b-1} &  \cdots      & \sqrt{b-1} \\
   &  & \ddots & \vdots\\
   &  &  & \sqrt{b^L-b^{L-1}}
  \end{pmatrix}.
 \]
 is an upper triangular matrix with positive diagonal elements. An easy computation shows that $\vec Y^\top \vec Y = \vec B$.
\end{proof}
\begin{lemma} \label{supremum ratio}
 Let $L \in \N_{>0}$, $b \ge 1$, and $0 = a_0 < a_1 \le \cdots \le a_L$. Let $\vec A = (\frac 1 2 a_{|m-\ell|} \cdot b^{\min(\ell,m)})_{0 \le \ell,m \le L}$, and $\vec B = (\frac 1 2 b^{\min(\ell,m)})_{0 \le \ell,m \le L}$
 . Then
 \[\sup_{\vec x \in \R^{\{0,\dotsc,L\}}} \frac{\vec x^\top \vec A \vec x}{\vec x^\top \vec B \vec x} = \lambda_{\max}(\vec Z),\]
 where $\vec Z 
 = (Z_{\ell,m})_{0\le \ell,m \le L}$ with
 \[Z_{\ell m} = \begin{cases}
  0 & \text{ if } \ell=m=0, \\
  \frac{a_k- a_{k-1}}{\sqrt{b^k - b^{k-1}}} & \text{ if } \ell=0,m=k, \text{ or } \ell=k,m=0 \text{ for } k \in [L],\\
  -\frac{2a_1}{b-1} & \text{ if } 1 \le \ell = m \le L, \\
  \frac{(b+1)a_k-b a_{k-1} - a_{k+1}}{b^{k/2}(b-1)}  & \text{ if } \ell,m \geq 1, |m-\ell|=k, \text{ for } k \in [L-1].
 \end{cases}\]
\end{lemma}
\begin{proof}
 Let $\vec B = \vec Y^\top \vec Y$ be the Cholesky decomposition of $\vec B$. We can rewrite
 \begin{equation}\label{eq:forgetful:rho_L_lambda_max}
		\sup_{\vec{x}\in \mathbb{R}^{\{0,\dotsc,L\}}}\  \frac{\vec{x}^\top \! \vec{A}^{\phantom{\mathclap{-\!\top}}}\vec{x}}{\vec{x}^\top \! \vec{B}^{\phantom{\mathclap{-\!\top}}}\vec{x}}
		=
		\sup_{\vec{x}: \vec{x}^\top \! \vec{B}\vec{x} = 1}\  \vec{x}^\top \! \vec{A}^{\phantom{\mathclap{-\!\top}}}\vec{x}
		=
		\sup_{\vec{x}: \|\vec{Y} \vec{x} \| = 1}\  \vec{x}^\top \! \vec{A}^{\phantom{\mathclap{-\!\top}}}\vec{x}
		=
		\sup_{\vec{z}: \|\vec{z} \| = 1}\  \vec{z}^\top \vec{Y}^{-\!\top} \! \vec{A}^{\phantom{\mathclap{-\!\top}}}\vec{Y}^{-1} \vec{z}
		=
		\lambda_{\max}(\underbrace{ \vec{Y}^{-\!\top} \!\vec{A}^{\phantom{\mathclap{-\!\top}}} \vec{Y}^{-1}}_{\coloneqq\vec{Z}}).
	\end{equation}
 
 So it remains to compute the matrix~$\vec Z$. The matrix $\vec Y$ is given in \cref{lem:Cholesky}, and its inverse is given by
	\begin{equation}
	\vec{Y}^{-1} = \sqrt{2}\begin{pmatrix}
		1  & -(b-1)^{-\frac{1}{2}}  &    &   & \\
		& \phantom{-}(b-1)^{-\frac{1}{2}}  & -(b^2-b)^{-\frac{1}{2}}  &   & \\
		&  & \phantom{-}(b^2-b)^{-\frac{1}{2}}  & \ddots  & \\
		&    & &\ddots & -(b^L-b^{L-1})^{-\frac{1}{2}}\\
		&    &    &  & \phantom{-}(b^L-b^{L-1})^{-\frac{1}{2}}
	\end{pmatrix}. \label{eq:inverse Y}
\end{equation}
 By computing the product $\vec Z = \vec Y^{-\top} \vec A \vec Y^{-1}$, we see that $\vec Z$ has the form claimed in the \lcnamecref{supremum ratio}, i.e.,
 \[
  \vec Z = \begin{pmatrix} 0 & \frac{a_1}{\sqrt{b-1}} & \frac{a_2-a_1}{\sqrt{b^2-b}} & \cdots & \frac{a_L - a_{L-1}}{\sqrt{b^L - b^{L-1}}} \\
 \frac{a_1}{\sqrt{b-1}} & -\frac{2a_1}{b-1} & \frac{a_1(b+1) - a_2}{b^{1/2}(b-1)} & &\frac{a_{L-1}(b+1) - a_{L-2} b - a_L}{b^{(L-1)/2}(b-1)} \\
 \frac{a_2-a_1}{\sqrt{b^2-b}} & \frac{a_1(b+1) - a_2}{b^{1/2} (b-1)} & -\frac{2a_1}{b-1} & & \frac{a_{L-2}(b+1) - a_{L-3} b - a_{L-1}}{b^{L/2-1}(b-1)} \\
 \vdots & & & \ddots \\
 \frac{a_L-a_{L-1}}{\sqrt{b^L-b^{L-1}}} & \frac{a_{L-1}(b+1) - a_{L-2} b - a_L}{b^{(L-1)/2}(b-1)} & \frac{a_{L-2} (b + 1) - a_{L-3}b - a_{L-1}}{b^{L/2-1}(b-1)} & & -\frac{2a_1}{b-1}
 \end{pmatrix}.\qedhere
 \]
\end{proof}

\begin{lemma} \label{lem:ellstar}
 Let $b \ge 1$. For $L \in \N_{>0}$ define the matrix $\vec Y_L \coloneqq \bigl(\sqrt{\frac{b^\ell - b^{\ell-1} \cdot \mathds 1_{\ell \ge 1}}{2}} \cdot \mathds 1_{m \ge \ell}\bigr)_{0 \le \ell, m \le L}$ and the vectors $\vec z_L \coloneqq \bigl(\sqrt{\frac{2}{L+1}} \cdot \sin\bigl(\frac{\ell \pi}{L+1}\bigr)\bigr)_{0 \le \ell \le L}$ and $\vec x_L = (x_\ell^{(L)}) \coloneqq \vec Y^{-1} \vec z_L$. Then $|x_\ell^{(L)}| \le \frac{2(\sqrt b + 1)}{\sqrt{(L+1)b^\ell(b-1)}}$ for all $0 \le \ell \le L$, and there is an $\ell^* \in \N_{>0}$ such that for all $L \ge \ell \ge \ell^*$ we have that $x^{(L)}_\ell \ge 0$. Moreover,
 $\lim_{L \to \infty} \sum_{\ell=\ell^*}^L x_\ell^{(L)} = 0$.
\end{lemma}
\begin{proof}
 Set $\ell' \coloneqq \bigl\lceil\frac{2}{\sqrt b - 1}\bigr\rceil$.
 We bound
 \[\lim_{L \to \infty} \frac{\sin \bigl( \frac{(\ell'+1) \pi}{L+1} \bigr)}{\sin \bigl( \frac{\ell' \pi}{L+1} \bigr)} = \frac{\ell' + 1}{\ell'} < \frac{\frac{2}{\sqrt{b}-1}+2}{\frac{2}{\sqrt{b}-1}} = \sqrt b,\]
 so there is an $L^*$ such that for all $L \ge L^*$ the left hand side is bounded by $\sqrt b$. Set $\ell^* \coloneqq \max\{\ell', L^*\}$, and let $L \ge \ell^*$ be fixed. By computing the product of the matrix in \cref{eq:inverse Y} with $\vec z_L$ we obtain
 \[
    x_\ell^{(L)}
    = \frac{2}{\sqrt{(L+1) b^{\ell} (b-1)}} \, \bigg( \sqrt b \cdot \sin\bigg( \frac{\ell \pi}{L+1} \bigg) - \sin\bigg( \frac{(\ell + 1) \pi}{L+1} \bigg) \bigg)
        \qquad\text{for }\ell=0, \dotsc, L
    .
 \]
 This implies the bound on the absolute values.
 Since the function
 \[\ell \mapsto \frac{\sin \bigl( \frac{(\ell+1) \pi}{L+1} \bigr)}{\sin \bigl( \frac{\ell \pi}{L+1} \bigr)}\]
 is decreasing on $[1,L]$, for all $\ell \in \{\ell^*,\dotsc,L\}$ it holds that
 \[\frac{\sin \bigl( \frac{(\ell+1) \pi}{L+1} \bigr)}{\sin \bigl( \frac{\ell \pi}{L+1} \bigr)} \le \frac{\sin \bigl( \frac{(\ell^*+1) \pi}{L+1} \bigr)}{\sin \bigl( \frac{\ell^* \pi}{L+1} \bigr)} \le \sqrt b \quad\implies\quad x_\ell^{(L)} \ge 0.\]
 To prove the last claim, observe that $\sum_{\ell=\ell^*}^L x^{(L)}_\ell$ is a telescoping sum for every $L \ge \ell^*$, and hence
 \[
 \sum_{\ell=\ell^*}^L x^{(L)}_\ell = \frac{2}{\sqrt{(L+1)b^{\ell^*-1} (b-1)}} \cdot \sin\biggl(\frac{\ell^*\pi}{L+1}\biggr) \xrightarrow{L \to \infty} 0
 . \qedhere
 \]
\end{proof}

\begin{lemma} \label{lem:generalized_Toeplitz_pos_semidefinite}
 For $k, L \in \N$ denote by $\vec T_{k.L}$ the $L \times L$ Toeplitz matrix with $2$ on the main diagonal and $-1$ on the $k$th and the $(-k)$th superdiagonal. Let $k, L \in \N$ with $k \mid L$, let $\vec v \in \R^k$, and let $\alpha \ge \Vert \vec v \Vert^2$. Then the matrix
 \[\vec H_L(\alpha, \vec v) \coloneqq
 \begin{pmatrix} 
  \alpha & v_1 & \cdots &    v_k \\
     v_1 &   2 &        &        &     -1 \\
  \vdots &     & \ddots &        &        & \ddots \\
     v_k &     &        & \ddots &        &        & \ddots \\
         &  -1 &        &        & \ddots &        &        & \ddots \\
         &     & \ddots &        &        & \ddots &        &        & -1 \\
         &     &        & \ddots &        &        & \ddots \\
         &     &        &        & \ddots &        &        & \ddots \\
         &     &        &        &        &     -1 &        &        &  2
 \end{pmatrix} = \left(\begin{array}{c|c} \alpha & \vec v^\top\ \ \vec 0^\top \\\hline \begin{matrix} \vec v \\ \vec 0 \end{matrix} & \vec T_{k,L} \end{array}\right)\]
 is positive semidefinite.
\end{lemma}
\begin{proof}
 We use the Schur complement lemma to show that $\vec{H}_L(\alpha, \vec v)$ is positive semidefinite. 
Let $a \coloneqq L/k \in \N$, and observe that the matrix $\vec T_{k,L}$ is of the form $\vec T_{1,a} \otimes \vec I_k$,
where $\otimes$ denotes the Kronecker product. The matrix $\vec{T}_{1,a}$ is a symmetric tridiagonal Toeplitz matrix, which has minimum eigenvalue $\lambda_{\min}(\vec T_{1,a}) = 2 \bigl(1-\cos\bigl(\frac{\pi}{a+1}\bigr)\bigr) > 0$ (see \cite[Theorem 2.4]{BG05}), and is thus positive definite. The reader may verify that the inverse is given by $(\vec{T}_{1,a}^{-1})_{ij}=\frac{1}{a+1} \cdot \min(i,j) \cdot (a+1-\max(i,j))$. Since the eigenvalues of the Kronecker product are the products of the eigenvalues, we have $\lambda_{\min}(\vec T_{k,L}) = \lambda_{\min}(\vec T_{1,a}) > 0$, i.e., $\vec T_{k,L} \succ 0$, and, moreover, $\vec{T}_{k,L}^{-1}= \vec{T}_{1,a}^{-1} \otimes \vec{I}_k$. In particular, the upper left $k\times k$ block of $\vec{T}_{k,L}^{-1}$ is equal to
$(\vec{T}_{1,a})^{-1}_{11} \cdot \vec{I}_k=\frac{a}{a+1}\vec{I}_k$. So we can form the Schur complement
\[
 2\alpha -
 [\vec{v}^\top\ \  \vec{0}^\top]\cdot \vec{T}_{k,L}^{-1} \cdot \begin{bmatrix}\vec{v}\\\vec{0}\end{bmatrix}
 = 
 \alpha
 -\vec{v}^\top \Bigl(\frac{a}{a+1}\vec{I}_k\Bigr) \vec{v}
 \geq \alpha - \|\vec{v}\|^2 \geq 0.
\]
This concludes the proof that $\vec H_L(\alpha,\vec v)$ is positive semidefinite.
\end{proof}

\begin{lemma} \label{lem:double limit}
 Let $a_{mn} \in \R$ for all $m, n \in \N$. Assume that for every $n \in \N$ the sequence $(a_{mn})_{m \in \N}$ converges to some $a_n \in \R$ and that the sequence $(a_n)_{n \in \N}$ converges to some $a \in \R$. Then there is $s \colon \N \to \N$ so that $a_{s(n)n} \xrightarrow{n \to \infty} a$.
\end{lemma}
\begin{proof}
 For every $n \in \N$ there is an $s(n) \in \N$ such that $\vert a_{s(n)n} - a_n \vert < \frac 1 n$. Then the resulting sequence $(a_{s(n)n})_{n \in \N}$ converges to $a$ because for every $\varepsilon > 0$ there is an $N \in \N$ such that $\vert a_n - a \vert < \frac \varepsilon 2$ and $\frac 1 n < \frac \varepsilon 2$ for all $n \ge N$, and hence, $\vert a_{s(n)n} - a \vert \le \vert a_{s(n)n} - a_n \vert + \vert a_n - a \vert < \frac 1 n + \frac \varepsilon 2 < \varepsilon$.  
\end{proof}
\end{document}